\documentclass[
 reprint,
superscriptaddress,
nofootinbib,
 amsmath,amssymb,
 aps,
 prx,
floatfix,
]{revtex4-2}
\usepackage{mathtools}
\usepackage{graphicx}
\usepackage{dcolumn}
\usepackage{bm}
\usepackage{amsmath,amssymb}
\usepackage{amsthm}
\usepackage{caption}
\usepackage{subcaption}
\usepackage{braket}
\usepackage{physics}
\usepackage{calc}
\newlength\myheight
\newlength\mydepth
\settototalheight\myheight{Xygp}
\newcommand*\inlinegraphics[1]{%
  \settototalheight\myheight{Xygp}%
  \settodepth\mydepth{Xygp}%
  \raisebox{-\mydepth}{\includegraphics[height=\myheight]{#1}}%
}

\usepackage{hyperref}
\usepackage[mathlines]{lineno}

\newtheorem{theorem}{Theorem}[section]
\newtheorem{corollary}{Corollary}[theorem]

\newtheorem{definition}{Definition}[section]

\begin{document}


\title{Algebraic Topology Principles behind Topological Quantum Error Correction}

\author{Xiang Zou}
\email{xiang.zou@mail.utoronto.ca}
\affiliation{%
 Department of Physics, University of Toronto, 60 St. George Street, Toronto, ON, M5S 1A7, Canada
}%

\author{Hoi-Kwong Lo}%
\email{hklo@ece.utoronto.ca}
\affiliation{%
 Department of Physics, University of Toronto, 60 St. George Street, Toronto, ON, M5S 1A7, Canada
}%
\affiliation{%
 Department of Electrical and Computer Engineering, University of Toronto, 10 King's College Road, Toronto, ON, M5S 3G4, Canada
}%
\affiliation{%
Center for Quantum Information and Quantum Control, University of Toronto, Toronto, Canada
}%
\affiliation{Quantum Bridge Technologies Inc., 108 College St., Toronto, Canada}
\affiliation{Department of Physics, National University of Singapore, Singapore 117551.}
\affiliation{Centre for Quantum Technologies, National University of Singapore, Singapore 117543}

\date{\today}

\begin{abstract}
Quantum error correction (QEC) is crucial for realizing scalable quantum technologies, and topological quantum error correction (TQEC) has emerged as the most experimentally advanced paradigm of QEC.  Existing homological and topological code constructions, however, are largely confined to orientable two-manifolds with simple boundary conditions, and higher-dimensional or non-manifold generalizations are typically treated on a case-by-case basis rather than within a single unifying framework. In this work, we develop a unified algebraic-topological framework for TQEC based on homology, cohomology, and intersection theory, which characterizes exactly when an arbitrary-dimensional manifold (with or without boundary) can serve as a quantum memory, thereby extending the standard 2D homological-code picture to arbitrary dimension and to manifolds with boundary via Poincaré-Lefschetz duality. Building on this classification, we introduce concrete code families that exploit nontrivial topology beyond the planar and toric settings. These include ``3-torus code'' and higher-dimensional ``volume codes'' on compact manifolds with mixed $X$- and $Z$-type boundaries. We further give a topological construction of qudit TQEC codes on general two-dimensional cell complexes using group presentation complexes, which unifies and extends several known quantum LDPC and homological-product-like constructions within a single geometric language. Finally, we combine the theoretical framework with numerical simulations to demonstrate that changing only the global topology - for example, from a torus to a Klein bottle - can yield improved logical performance at fixed entanglement resources. Taken together, our results provide a systematic set of topological design principles for constructing and analyzing TQEC codes across dimensions, boundaries, and local dimensions, and they open new avenues for topology-aware fault-tolerant quantum architectures.
\end{abstract}

\maketitle


\section{Introduction}

Quantum computation has demonstrated significant theoretical potential in addressing various computational challenges, spanning both classical and quantum domains. Notable examples include prime factorization, which underpins the security of the standard RSA encryption scheme~\cite{Shor_algorithm1}, and quantum phase estimation, a critical tool for numerous quantum applications~\cite{quantum_phase_estimation1}. However, the practical realization of quantum systems faces substantial challenges due to their inherent fragility to noise and decoherence. This vulnerability arises from the quantum no-cloning theorem~\cite{No_cloning1, No_cloning2}, which fundamentally limits the applicability of classical error correction techniques to quantum systems. Thus, quantum error correction (QEC) has emerged as a fundamental component for reliable and secure quantum computation, which also has significant impacts on quantum communication~\cite{QECforCommunication1, QECforCommunication2, QECforCommunication3}. As quantum technologies advance toward practical deployment in computational tasks, QEC provides essential mechanisms for protecting quantum information against noise~\cite {QEC_importance1, QEC_importance2, Google1, Google2, Xanadu1, Xanadu2, Xanadu3}, and perhaps paves the pathway to achieve fault-tolerant quantum computation. In QEC, a qubit error can be represented as $Err = e_0I+e_1\sigma_X+e_2\sigma_Z+e_3\sigma_X\sigma_Z$, where bit flips $\sigma_X$, phase flips $\sigma_Z$, and their combinations constitute a basis for describing error spaces. QEC codes encode logical qubits with more physical qubits to make them robust against both bit-flip and phase-flip errors~\cite{Principle_of_QEC1}. Different institutions have explored diverse QEC techniques to achieve low logical error rates, thereby advancing the prospects of quantum technologies. Notable approaches include Google's implementation of surface code architectures on superconducting qubits~\cite{Google1, Google2, XZZX1}, IBM's pursuit of quantum Low-Density Parity Check codes~\cite{IBM1, IBM2}, Xanadu's development of Gottesman-Kitaev-Preskill (GKP) encoded photonic systems~\cite{Xanadu1, Xanadu2, Xanadu3}, and Microsoft's utilization of topological qubit architectures based on Majorana quasi-particles~\cite{Microsoft1, Microsoft2}. Among these paradigms, topological quantum error correction (TQEC) has arguably achieved the most experimental success with current technological advancements. Google has recently demonstrated a below-threshold surface code error rate in their superconducting qubit system~\cite{Google1, Google2}, a milestone sought by the industry for over three decades. The surface code employed by Google is part of the broader class of TQEC methods derived from the toric code framework~\cite{Toriccode1, Surfacecode1, TQEC_review1}.

The concept of topological quantum error correction (TQEC) was initially introduced by Kitaev~\cite{Toriccode1}, wherein two qubits are encoded on the two-dimensional topological manifold of a torus ($T^2 = S^1 \times S^1$). This pioneering framework laid the foundation for subsequent advancements in TQEC codes. In pursuit of practical implementations, several later studies have focused on models of TQEC that are suitable for experimental realization. Notably, surface codes defined on planar geometries with boundaries, commonly referred to as planar codes, have been extensively investigated. These studies have shown that TQEC codes can be effectively realized on planar surfaces with non-contractible loops, thereby extending the applicability of the original toric code framework~\cite{Toriccode1, TQEC_review1, HomologicalQEC1}. In parallel, Bombin and Martin-Delgado introduced a comprehensive theory of homological quantum error correction~\cite{HomologicalQEC1}. Their work rigorously established that non-orientable surfaces are unsuitable for TQEC codes defined on qudits, where a qudit is a quantum system with a finite-dimensional Hilbert space of dimension $D > 2$.

In this paper, we first focus on the theory behind qubit codes. In Section~\ref{topological_properties}, we give an overall review of algebraic topology and the topological properties of manifolds that are crucial for our discussions. Then in Section~\ref{sec: TQEC Codes on 2-Manifolds}, we formulate the theory of TQEC codes with the intersections between homology and cohomology loops, propose extending the TQEC paradigm beyond traditional topological manifolds, and use different manifolds that are seldom discussed, including the non-orientable manifolds like the Klein bottle. 

Furthermore, in Section~\ref{sec: TQEC Codes on Higher Dimensional Manifolds}, we explore the extension of TQEC into higher-dimensional manifolds that satisfy specific conditions. We believe that an $n$-dimensional manifold $M$ can encode qubits on the $i$-th dimension if the $i$-th homology group $H_i(M;\mathbb{Z}_2)$ is non-trivial, and we further propose an implementation for a 3-dimensional code. These higher-dimensional structures encompass a broader class of TQEC codes, which hold promise for future applications and are of substantial theoretical interest in advancing the understanding of TQEC principles. 

Subsequently, in Section~\ref{sec:TQEC_boundary}, we extend the theory to include TQEC codes on manifolds with boundary, which involves the ideas of relative cycles and relative homology reviewed in Section~\ref{topological_properties}, which then allows us to construct TQEC codes on manifolds with boundary in any arbitrary dimensions, analogous to the ``surface code'' in 2-dimensions.

Then, we shift our focus from qubit codes to qudit codes in Section~\ref{sec:TQEC_on_2_complex}, where we demonstrate a bottom-up construction of the TQEC code on a 2-dimensional cell-complex for a general $d$-dimensional qudit, which is of great theoretical interest. We further showed that some quantum Low Parity Density Check (qLDPC) codes are special cases of our code construction.

Finally, in Section~\ref{sec: Examples of TQEC Codes on 2-Manifolds and Simulation Results}, we demonstrate the simulation results of the TQEC codes using the Klein Bottle and $\mathbb{R}P^2$, which were seldom discussed in the literature. The result gives great potential for developing new TQEC codes for future applications.

As we are aware, our work is among the first few works that generalize the theory of TQEC to arbitrary-dimensional manifolds with and without boundaries, which represents a great step forward in understanding the theory behind TQEC and achieving fault-tolerant quantum computation and communication.

\section{Topological Properties of Manifolds}
\label{topological_properties}

To establish a robust foundation for the principles and properties of TQEC codes, we first provide a comprehensive overview of the characteristics of manifolds and their associated topological properties. This section systematically examines the properties of manifolds across arbitrary dimensions, with specific emphasis on 2-manifolds, which serve as the basis for most TQEC codes discussed in the literature~\cite{Toriccode1, TQEC_review1, HomologicalQEC1, TQEC_inefficient2}. These properties are crucial to subsequent discussions, particularly when assessing various theorems related to the quantum error correction capabilities of specific manifolds.

The topological properties of a manifold $ M $ can be characterized in several ways, including the fundamental group $ \pi_1(X, x_0) $ at a base point $ x_0 $, as well as homology groups $ H_n(X) $ and cohomology groups $ H^n(X) $~\cite{Algebraic_Topology1, Algebraic_Topology2}. In this work, our theoretical framework primarily utilizes $ H_n(X) $ and $ H^n(X) $, as homology groups are finitely generated Abelian groups.

\subsection{Chain Complexes, Relative Chain Complexes and Cycles}
\label{sec:Chain_complex}

Before discussing the homology, we first review the standard definition of the boundary and coboundary maps. Denote the set of chains complex on a manifold $M$ over a commutative ring $R$ as $C_n(M;R)=R\langle\sigma:\Delta^n\rightarrow M\rangle$, where $\Delta^n$ represents the $n$-simplex on $M$~\cite{Algebraic_Topology1, Algebraic_Topology2}, the formal definition of simplexes is in Appendix~\ref{appendix:simplex_cell_complex}. Then, we can define the boundary map, $\partial_n: C_n \rightarrow C_{n-1}$, where $\partial_n$ takes an $n$-simplex from $C_n(M)$ and return a $(n-1)$-simplex which is the boundary of the $n$-simplex~\cite{Algebraic_Topology1, Algebraic_Topology2}. Similarly, we can define the co-boundary map~\cite{HomologicalQEC1},
$\delta^i: C^{i-1} \rightarrow C^i$, which is the dual homomorphism of $\partial_i$.  This means, for some cell complex $X_M$, and some co-chain $c^n\in C^n(X_M;R)$, and some chain $c_n\in C_n(X_M;R)$~\cite{HomologicalQEC1}
\begin{equation}
    (\delta c^{i-1}, c_i):=(c^{i-1}, \partial c_i)
\end{equation}

Given a topological space \(X\), one typically takes \(C_n(X)\) to be the free abelian group generated by singular \(n\)-simplices in \(X\). For a subspace \(A \subseteq X\), the relative chain complex \(C_n(X, A)\) is defined as the quotient \(C_n(X) / C_n(A)\), whose boundary maps are induced from those of \(C_n(X)\). The relative chain complex thus measures chains in \(X\) modulo those that lie entirely in \(A\)~\cite{Algebraic_Topology1, Algebraic_Topology2}. 

Then, two groups can be defined, where Ker and Im mean kernel and image, respectively.

\begin{equation}
\begin{split}
    \text{$n$-cycle in X: }Z_n(X;R)&:=\text{Ker}(\partial_n)\\
    \text{$n$-boundary in X: }B_n(X;R)&:=\operatorname{Im}(\partial_{n+1})
\end{split}
\end{equation}

In a chain complex \((C_*, \partial_*)\), an element \(z \in C_n\) is called an \(n\)-cycle if it lies in the kernel of the boundary operator, that is, \(\partial_n(z) = 0\). The group of all \(n\)-cycles is denoted by \(Z_n = \ker(\partial_n)\). Similarly, an element \(b \in C_{n-1}\) is called an \((n-1)\)-boundary if it lies in the image of the boundary map from the next chain group, that is, \(b = \partial_n(c)\) for some \(c \in C_n\); the group of all such boundaries is denoted by \(B_{n-1} = \operatorname{Im}(\partial_n)\). Because \(\partial_{n-1} \circ \partial_n = 0\), every boundary is a cycle, hence \(B_n \subseteq Z_n\). 

For a pair of topological spaces \((X, A)\) with \(A \subseteq X\), one can define the corresponding relative chain complex:
\begin{definition}
    The relative chain cycle of a manifold $X$ relative to a sub-manifold $A\subseteq X$ is defined as 
    \begin{equation}
        C_n(X, A) = C_n(X)/C_n(A)
    \end{equation}
\end{definition} 

with induced boundary maps \textit{\(\partial_n : C_n(X, A) \to C_{n-1}(X, A)\)}. An element \([z] \in C_n(X, A)\) (the equivalence class of \(z \in C_n(X)\)) is called a relative \(n\)-cycle if \(\partial_n([z]) = 0\) in \(C_{n-1}(X, A)\), that is, if \(\partial_n z \in C_{n-1}(A)\). The group of relative \(n\)-cycles is therefore given by
\begin{equation}
Z_n(X, A) = \{ z \in C_n(X) \mid \partial_n z \in C_{n-1}(A) \},
\end{equation}
and the group of relative \(n\)-boundaries by
\begin{equation}
B_n(X, A) = \{ \partial_{n+1} c \mid c \in C_{n+1}(X), \ \partial_{n+1} c \in C_n(X) \}.
\end{equation}
Which means, relative $n$-cycles $z\in Z_n(X, A)$ are chain in $C_n(X)$ with their boundary in $C_{n-1}(A)$. The concept of relative cycles is crucial when we discuss TQEC codes on manifolds with boundary and holes in Section~\ref{sec:TQEC_boundary}.

\subsection{Homology and Relative Homology}

\begin{definition}
    The $n$-th Homology group of a topological space $X$, over some commutative ring $R$, denoted as $H_n(X;R)$ is defined as 
    \begin{equation}
        H_n(X;R) = Z_n(X;R)/B_n(X;R)
    \end{equation}
\end{definition}

In other words, $H_n(X)$ is defined as the group of all closed $n$-cycles modulo the $n$-cycles that are the boundary of an $(n+1)$-cell. Therefore, considering 2-manifolds, the term ``non-contractible loops'' is commonly used in literature to refer to the non-trivial members of $H_1(X)$, because they cannot be smoothly deformed into a point. For instance, the torus $T^2$ has 2 classes of non-contractible loops~\cite{Toriccode1}, which comes from the fact that $H_1(T^2)=\mathbb{Z}\times\mathbb{Z}$. 

For the encoding of qubits, only the parity of the winding number of the ``non-contractible loops'' is needed. Therefore, instead of considering the general homology groups, it is sufficient to consider the homology groups with $\mathbb{Z}_2$ coefficients~\cite{HomologicalQEC1}. Another way to see that only $\mathbb{Z}_2$ coefficients matter is that when we try to find a closed non-contractible loop on the surface $M$ for our logical $\sigma_X$ and $\sigma_Z$ operators, which will be discussed in more details in section~\ref{sec: definition of TQEC}, since $\sigma_X^2 = \sigma_z^2=\mathbb{I}$. Thus, considering the TQEC code for qubits, every closed non-contractible loop has to be its inverse. 

For a pair \( (X, A) \) consisting of a topological space \( X \) and a subspace \( A \subseteq X \), the relative chain complex is given by the quotient \( C_*(X, A) = C_*(X) / C_*(A) \). The corresponding boundary operators induce well-defined maps on the quotient. The relative homology groups are then defined by
\begin{equation}
H_n(X, A) = H_n(C_*(X, A)) = \operatorname{Ker}(\partial'_n) / \operatorname{Im}(\partial'_{n+1}),
\end{equation}
which measure how the topology of \( X \) differs from that of its subspace \( A \), capturing information about the structure of \( X \) relative to \( A \). The concept of relative homology will become crucial when we discuss TQEC codes on compact manifolds with boundary and holes.

\subsection{Cohomology and Relative Cohomology}

With the definition of homology groups and the co-boundary operators defined above, we can now define the cohomology groups $H^n(X)$. We follow the definition from~\cite{Algebraic_Topology1, Algebraic_Topology2, HomologicalQEC1}. On a $n$-manifold, for $i\in[0, n]$, the definition starts with the co-chain complex on $X$ over some commutative ring $R$, $C^i(X;R)$.
\begin{equation}
    C^i(X;R):=Hom\left(C_i(X;R),R\right)
\end{equation}
This is the group of singular $i$-cochains. And we quote the definition of the co-boundary map as $\delta^n:C^n(X;R)\rightarrow C^{n+1}(X;R)$. And then, we can define two groups,
\begin{equation}
\begin{split}
    \text{$n$-cocycle in X: }Z^i(X;R)&:=\text{ker}(\delta^n)\\
    \text{$n$-coboundary in X: }B^i(X;R)&:=\text{Im}(\delta^{n-1})
\end{split}
\end{equation}
Then, the definition of the cohomology group is as follows,
\begin{definition}
    The $n$-th Cohomology group of a space $X$, denoted as $H^n(X)$ is defined as \begin{equation}
        H^n(X)=Z^n(X;R)/B^n(X;R)
    \end{equation}
\end{definition}

For a pair \( (X, A) \) of topological spaces with \( A \subseteq X \), the relative cochain complex is defined as the subcomplex \( C^*(X, A; G) = \operatorname{Ker}( C^*(X; G) \to C^*(A; G) ) \), where the map is induced by restriction of cochains. The coboundary maps descend naturally to this subcomplex, giving a well-defined cochain complex whose cohomology groups are the relative cohomology groups.
\begin{equation}
    H^n(X, A; G) = H^n(C^*(X, A; G)) = \operatorname{Ker}(\delta^n) / \operatorname{Im}(\delta^{n-1}).
\end{equation}

Intuitively, \( H^n(X, A; G) \) measures the cohomological information of \( X \) that vanishes on \( A \), or equivalently, the obstructions to extending cochains on \( A \) to cochains on \( X \).

The isomorphism between the singular homology and the simplicial homology of a manifold $M$ is an important theorem in algebraic topology, stated as Theorem 2.27 from~\cite{Algebraic_Topology1}. With this, considering a cell-complex embedded on $M$ formed by attaching different simplexes, a corollary of the theorem leads to:

\begin{corollary}
    For a closed and compact manifold $M$, let $X_M$ be a cell complex embedded on $M$ and $R$ be a commutative ring; the homomorphisms $H^{Cell}_n(X;R) \rightarrow H_n(X;R)$ are isomorphisms for all $n$.
\end{corollary}

This corollary implies that for any manifold $M$, and any cell complex $X_M$ on $M$, the singular chain complex $C_i(M;R)$ and the cell chain complex $C_i^{cell}(X_M;R)$ are isomorphic as chain complexes. And then, because of this chain complex isomorphism, the singular homology $H(M;R)$ is isomorphic to the simplicial homology $H^{cell}(X_M;R)$. Therefore, by choosing different cell complexes on the manifold $M$, the homology is invariant. As we will later prove in section~\ref{sec:TQEC_on_2_manifolds}, the TQEC ability of a surface is independent of the cell complex we choose to embed on the manifold.

\subsection{Poincaré Duality and Poincaré–Lefschetz isomorphisms}

Another crucial theorem for our discussion is the famous Poincaré duality in algebraic topology~\cite{Algebraic_Topology1, Algebraic_Topology2, poincare_duality1}, and it is quoted here.

\begin{theorem} [Poincaré duality]
For any commutative ring $R$, if $M$ is a closed $R$-orientable $n$-manifold with fundamental class $[M]\in H_n (M;R)$, then the map $D:H^k(M;R)\rightarrow H_{n-k}(M;R)$ defined by $D(\alpha)=[M]\cap \alpha$ is an isomorphism for all $k$.
\label{Poincaré}
\end{theorem}

The above definition uses the cap product $\cap$, formally defined in Appendix~\ref{appendix:pairing}. Notice that the above definition involves homology groups with coefficients from a coefficient ring $R$, where $R$ is usually taken as $\mathbb{Z}$. Then the theorem~\ref{Poincaré} states that for any compact oriented $n$-manifolds, there is an isomorphism between $H_k(M)$ and $H^{n-k}(M)$. However, since we are interested in using a topological manifold as TQEC codes for qubits, for much of this paper, the coefficient ring $R$ is considered to be $\mathbb{Z}_2$ as mentioned above. The condition of the surface becomes $\mathbb{Z}_2$-orientable, where, by definition, all compact manifolds, including the non-orientable manifolds, are $\mathbb{Z}_2$-orientable~\cite{Z2_orientable1, poincare_duality_Z2_1}.

\begin{corollary}
For any closed and compact topological manifold $M$ with dimension $n$, there is a isomorphism between the $H_k(M;\mathbb{Z}_2)$ and $H^{n-k}(M;\mathbb{Z}_2)$.

Furthermore, if $M$ is a closed and compact 2-manifold, the isomorphism further simplifies to
\begin{equation}
    H_1(M; \mathbb{Z}_2)\cong H^{1}(M; \mathbb{Z}_2)
\label{H1hom_and_cohom}
\end{equation}
\label{H1duality}
\end{corollary}

Using the homology over $\mathbb{Z}_2$ is sufficient for characterizing the manifold's TQEC ability for qubits, and it allows us to consider non-orientable surfaces. Most of the previous discussion of TQEC was focused on orientable surfaces, including the original toric code~\cite{Toriccode1}. Utilizing non-orientable surfaces greatly widens the possibilities of manifolds that can be used for qubit quantum codes. The Poincaré duality is a well-known theorem, and its proof can be found in any algebraic topology textbook, for example~\cite{Algebraic_Topology1, Algebraic_Topology2}. 

The relationship between the homology and cohomology groups can be regarded as the unique pairing between the $[\gamma]\in C_i(M)$ and $[\beta]\in C^i(M)$, which can be interpreted as the number of intersections. For a manifold, given a $i$-chain $\gamma \in C_i(M, \mathbb{Z}_2), \partial\gamma=0$ and $[\gamma]\ne0$ and a class of functionals on loops $\beta\in C^i(M, \mathbb{Z}_2), \delta\beta=0$ and $[\beta]\ne 0$, we can ``evaluate'' the functionals on the loop gives a number in $\mathbb{Z}_2=\{0, 1\}$. This can be interpreted as the intersection points between $[\gamma]$ and $[\beta]$. Poincaré duality in equation~\ref{Poincaré} has non-trivial implications for the intersections between homology and cohomology groups~\cite{Intersection_homology1, Intersection_homology2}, for any closed and oriented n-manifold $M$. The Poincaré duality can be reformulated as,
\begin{itemize}
    \item There is a functional intersection product, where R is a commutative ring if $M$ is the R-orientable surface
    \begin{equation}
        H_i(M;R)\times H_j(M;R)\xrightarrow{\cap} H_{i+j-n}(M;R)
    \end{equation}
    \item When $i+j=n$, and R is a commutative ring, if $M$ is the R-orientable surface, then the pairing
    \begin{equation}
        H_i(M; R)\times H_j(M;R)\xrightarrow{\cap} H_0(M, R)\xrightarrow{\epsilon}R
    \label{pairing}
    \end{equation}
    is non-degenerate, where $\epsilon$ is named Argumentation, which counts the points of a 0-cycle according to their multiplicities~\cite{Intersection_homology1, Intersection_homology2}.
\end{itemize}

As a result of equation~\ref{H1hom_and_cohom} and equation~\ref{pairing}, the pairing between $H_i(M; \mathbb{Z}_2)$ and $H^i(M;\mathbb{Z}_2)$ is unique and non-degenerate, which means, the pairing can be interpreted as the intersection number. Furthermore, when we consider $\mathbb{Z}_2$ coefficients, the pairing condition provides the parity of the intersection number. Later in section~\ref{intersection_on_2_manifolds}, we will show that the parity of the intersection number is well-defined and invariant under curve homotopy.

When discussing the manifolds with boundary, the usual Poincaré duality does not apply. Therefore, here we introduce a generalization of Poincaré duality~\cite{Algebraic_Topology1, Algebraic_Topology2},

\begin{theorem}[Poincaré--Lefschetz duality]
Let \(M\) be a compact, $R$-orientable \(n\) manifold with boundary \(\partial M\), and let \(R\) be a commutative ring. 
For each \(0\le k\le n\), the cap product with the relative fundamental class induces natural isomorphisms
\begin{equation}
D_M \colon H^k(M;R)\xrightarrow{\ \ \frown [M,\partial M]\ \ } H_{n-k}(M,\partial M;R)
\end{equation}
\begin{equation}
D_M^\partial \colon H^k(M,\partial M;R)\xrightarrow{\ \ \frown [M,\partial M]\ \ } H_{n-k}(M;R).
\end{equation}
If \(\partial M=\varnothing\), these reduce to the classical Poincaré duality isomorphisms \(H^k(M;R)\cong H_{n-k}(M;R)\).
\label{thm:Poincare_Lefschetz_duality}
\end{theorem} 

The Poincaré-Lefschetz duality induced the following corollary, which is stated as Theorem 3.43 from Ref.~\cite{Algebraic_Topology1}

\begin{corollary}
    Given a compact $R$ orientable $n$ manifold $M$ with boundary $\partial M$, if $\partial M$ is further decomposed as the union of two compact $(n-1)$ dimensional manifolds $A$ and $B$ with a common boundary $\partial A = \partial B = A \cap B$. Then the cap product with a fundamental class $[M] \in H_n(M, \partial M; R)$ gives isomorphisms $D_M: H^k(M, A; R)\longrightarrow H_{n-k}(M, B; R)$ for all $k$.
    \label{thm:Lefschetz_duality}
\end{corollary}

It is clear to see that Theorem~\ref{thm:Lefschetz_duality} will reduce to the Poincaré-Lefschetz duality in Theorem~\ref{thm:Poincare_Lefschetz_duality} when $A=\emptyset$ or $B=\emptyset$; this then reduces to the usual Poincaré duality when $\partial M = \emptyset$. The duality in Theorem~\ref{thm:Lefschetz_duality} will be useful when we consider TQEC codes on compact manifolds with boundaries.

\section{TQEC Codes on 2-Manifolds}
\label{sec: TQEC Codes on 2-Manifolds}
\subsection{Definition of TQEC Codes on 2-Manifolds}
\label{sec: definition of TQEC}

With the properties of manifolds established, we proceed to define topological quantum error correction (TQEC) codes on cell complexes embedded in 2-manifolds, which will subsequently be extended to higher-dimensional manifolds. While the formal definition of a cell complex is provided in Appendix~\ref{appendix:simplex_cell_complex}, for simplicity, it can be conceptualized as decomposing a two-dimensional manifold into constituent parts, formed by vertices, edges, and faces, and organizing these components into corresponding sets. In this framework, we focus on quantum error correction codes for qubits~\cite{Toriccode1}.

First, the TQEC code construction process starts by using a 2-dimensional manifold. For instance, in Kitaev's original paper~\cite{Toriccode1}, the toric code uses a torus, and the surface code uses a flat surface~\cite{Surfacecode1, XZZX1}. At the same time, there are some proposals using other 2-manifolds, but not widely discussed, like the 2-dimensional real projective plane $\mathbb{R}\text{P}^2$~\cite{RP2code1} and Klein bottles $K$. Then, we consider two-dimensional cell complexes embedded on the 2-manifold $M$, this can be characterized by a 3 tuple, $G=(V, E, F)$, which is made up of the 0-cell vertices $V = \{v_i\}$, 1-cell edges $E = \{e_j\}$, and 2-cell faces $F = \{f_k\}$~\cite{TQEC_review1}.

Following that, the dual lattice space can be defined as $\bar{G}=(\bar{V}, \bar{E}, \bar{F})$, where $\bar{V}=F$, $\bar{E}=E$ and $\bar{F}=V$~\cite{TQEC_review1}. This means that, intuitively, to construct the dual lattice space, every vertex is expanded into a face and every face is shrunk to a vertex~\cite{Algebraic_Topology1, Algebraic_Topology2}, a physical qubit is placed on every edge $e\in E$.

Then, boundary and co-boundary maps, defined in section~\ref{topological_properties}, are applied to the definitions of the local stabilizer operators that provide an implementation for TQEC codes~\cite{HomologicalQEC1, Algebraic_Topology1, Algebraic_Topology2}.

The Vertex and Plaquette Operators of a topological quantum error correction code, for vertices $v\in V$ and faces $f\in F$, can be defined as follows~\cite{Toriccode1, TQEC_review1}. 

\begin{figure}
    \centering
    \includegraphics[width=\linewidth]{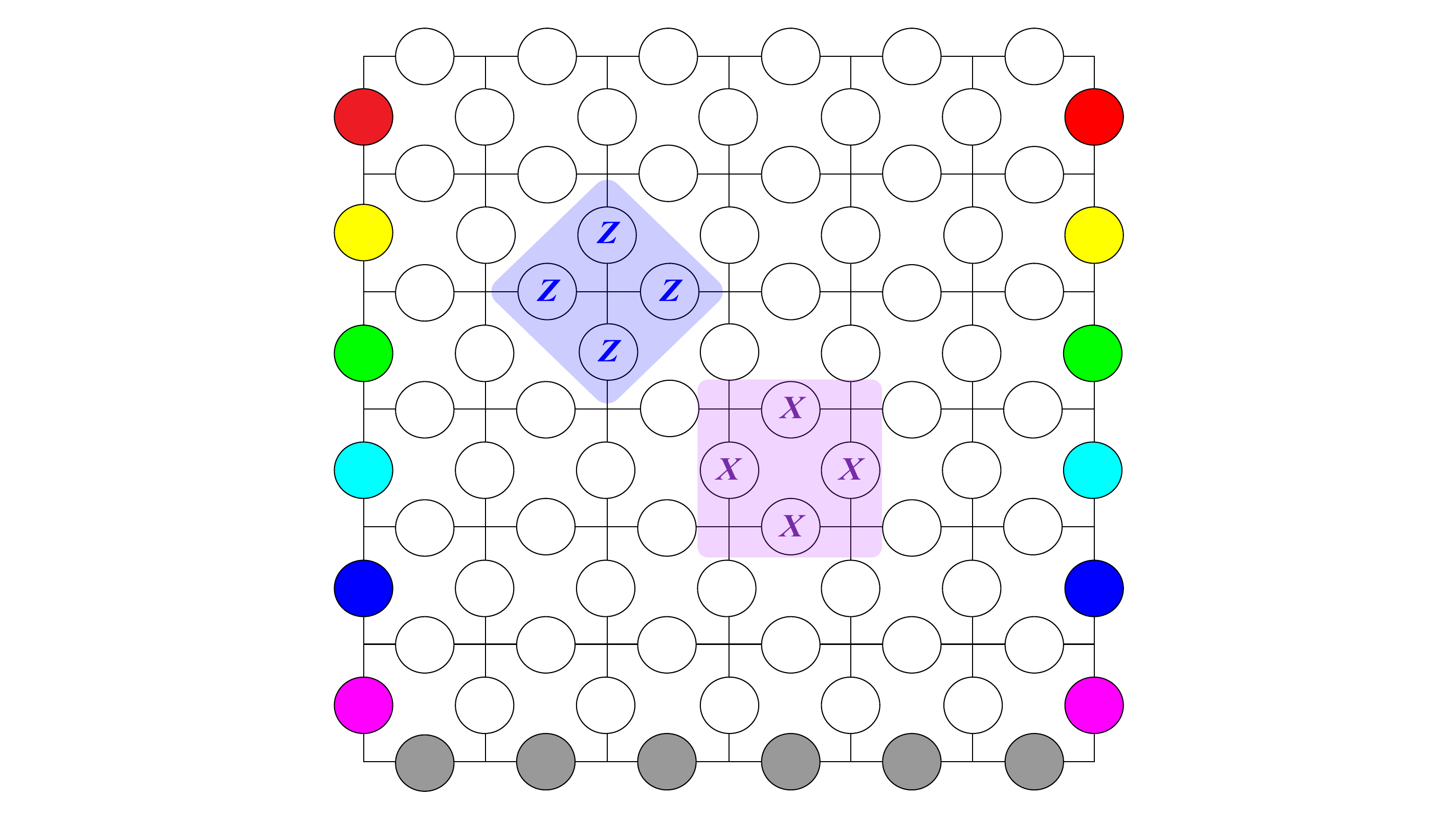}
    \caption{\raggedright Two types of operators are shown, and the torus topology is attained by identifying the qubits on both sides. For clarity, we only showed the identification of one pair of edges.}
    \label{fig:Toric_code_stabilizers}
\end{figure}

\begin{itemize}
    \item Vertex Operator: \begin{equation}
        A_i = \sigma_Z\left(\delta v_i\right) =\sigma_Z\left(\partial\bar{f_i}\right)
    \end{equation}
    \item Plaquette Operator:\begin{equation} B_j = \sigma_X\left(\partial f_j\right)\end{equation}
\end{itemize}

In simple terms, the vertex operators $A_i$ are defined as the stabilizer formed by the tensor product of $\sigma_Z$ of qubits on edges with boundary $v$; and the plaquette operators $B_j$ are defined as the stabilizer formed by the tensor product of $\sigma_X$ at the boundary of each face $f$. It is easily checkable that $\Pi_v A_v=\Pi_pB_p=\mathbb{I}$, and all the stabilizers commute, where $\mathbb{I}$ is the identity. Figure~\ref{fig:Toric_code_stabilizers} shows a well-established example of defining the two types of stabilizers on a 2-dimensional lattice used by the Toric code~\cite{Toriccode1}.

And, we define the code space of our TQEC to be the ground state of the following Hamiltonian, for some constant $J$: 
\begin{equation}
    H = -J\sum_v A_v - J\sum_p B_p
    \label{TQEC hamiltonian}
\end{equation}

A logical $\sigma_X$ is defined to operate on any cycles $\gamma\in C_1(M; \mathbb{Z}_2)$, such that $\gamma \subset E$, where $\partial[\gamma]=0$ and $[\gamma]\ne 0$. Similarly, a logical $\sigma_Z$ is defined to operate on any co-cycles, $\beta\in C^1(M; \mathbb{Z}_2)$, such that $\beta \subset E$, where $\delta[\beta]=0$ and $[\beta]\ne 0$. 

Thus, the above definition of logical $\sigma_X$ and $\sigma_Z$ satisfies the properties of Pauli operations. Firstly, by the definition of $\mathbb{Z}_2$ coefficient, $[\gamma]+[\gamma]=0$ and $[\beta]+[\beta]=0, \forall \gamma, \beta$, meaning that operate logical $\sigma_X^2 = \sigma_Z^2 = \mathbb{I}$. Also, using the pairing property as presented in equation~\ref{pairing}, $[\beta]([\gamma]) = 1$, meaning that $[\beta]$ and $[\gamma]$ intersect at an odd number of sites, thus, the logical $\sigma_X$ and $\sigma_Z$ anti-commutes, the idea is illustrated in Figure~\ref{fig:X_Z_Logical_intersection}.

\begin{figure}
    \centering
    \includegraphics[width=\linewidth]{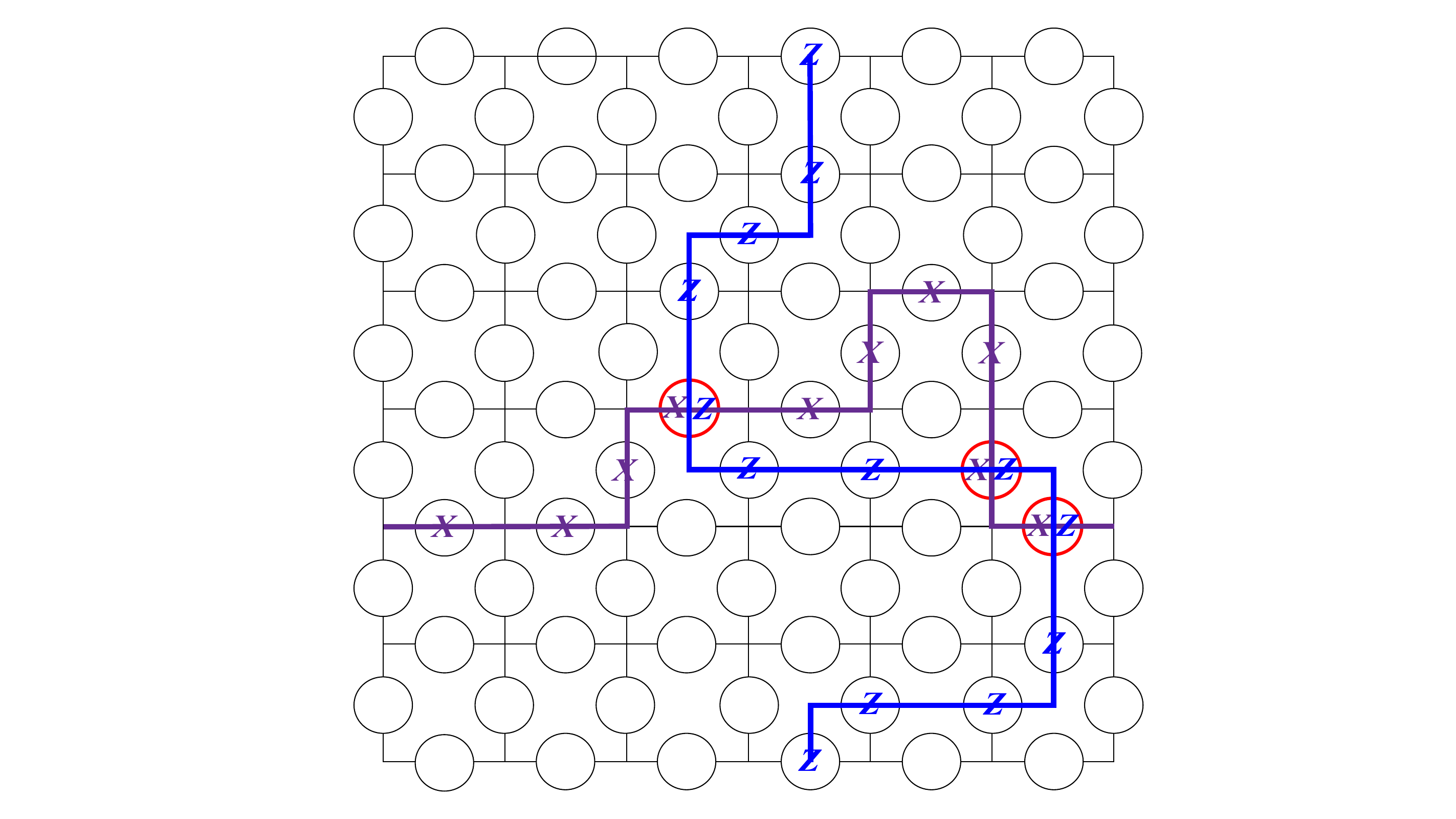}
    \caption{\raggedright Illustration of the intersection between the homology and cohomology cycles. An example of a 1-cycle is in blue, and a 1-cocycle is in red; they intersect at 3 sites. The parity of the intersection number is conserved up to any curve homotopy.}
    \label{fig:X_Z_Logical_intersection}
\end{figure}

\subsection{Intersections on 2-Manifolds and TQEC Ability}
\label{intersection_on_2_manifolds}

The original toric code~\cite{Toriccode1, TQEC_review1} can be used as quantum memory because the ground state of the Hamiltonian in equation~\ref{TQEC hamiltonian} has a four-fold degeneracy. Therefore, no local stabilizer checks can tell apart the degenerate ground states. The four ground states can be transformed into one another by applying logical $\sigma_X$ and $\sigma_Z$ operations, which anti-commute because the chain of logical $\sigma_X$ and $\sigma_Z$ operators intersects at an odd number of sites.

Thus, we propose that the TQEC ability of any manifold $M$ is related to the intersection properties of sub-manifolds of $M$. The definition of homotopy of curves and transverse intersection is presented in Appendix~\ref{appendix:homotopy_transversality}, and we shall only consider transverse intersection points. Here, the formal definition of the modulo-2 intersection is presented, which is useful for proving the theorems in later sections~\cite{Differential_Topology1, Algebraic_Topology1, Algebraic_Topology2}.

\begin{definition}
    Let $A$ be a compact manifold, $M$ be a manifold, and $B \subset M$ be a submanifold, all of them are without boundary and satisfying $dim(A)+dim(B) = dim(M)$. Let $f : A \rightarrow M$ be a smooth map, then the $mod\ 2 $ intersection number $I_2(f, B)$ of $f$ with $B$ is defined as,
    \begin{equation}
        I_2(f, B):= \#f^{-1}(B) \ (mod\ 2)
    \end{equation}
    Where $\#$ counts the number of intersections. If $f$ is the inclusion of $A$ as a sub-manifold, then we may use the notation $I_2(A, B):= I_2(f, B)$.
\end{definition}

In the above definition, $f^{-1}(B)$ is the preimage of $B$ under $f$, which is a finite set of points due to transversality and dimensional constraints. And, with the above definition, we quote the well-known modulo-2 intersection theorem~\cite{Differential_Topology1}.
\begin{theorem}
    If 2 maps, $f_0, f_1:A\rightarrow M$ are homotopic and both transversal to $B$, then, $I_2(f_0, B)=I_2(f_1, B)$. In other words, $I_2(f, B)$ is invariant under curve homotopies.
    \label{mod2intersection}
\end{theorem}

The proof can be found in any differential topology textbook, for instance~\cite{Differential_Topology1}, and the mathematical definition of homotopy is in Appendix~\ref{appendix:pairing}. Intuitively, a homotopy of a curve means a smooth and continuous deformation of that curve. The theorem~\ref{mod2intersection} is crucial when we investigate the necessary conditions for a topological manifold to be used for TQEC purposes. Restricting to 2-manifolds, the following corollary is obtained.

\begin{corollary}
    Given a closed and compact 2-manifold $M$. For any closed loops $A, B\subset M$, and $\partial A = \partial B = 0$. Then, the parity of the number of intersection points between loops $A$ and $B$ is invariant under homotopies of $A$ and $B$.
\end{corollary}

This corollary sets limitations on the surfaces that can be used for TQEC purposes. Theorem~\ref{mod2intersection} shows that the parity of intersection points of 2 loops on a 2-manifold is preserved under homotopy. With all the mathematical fundamentals, we will now develop the theoretical structure for TQEC codes on 2-manifolds in the following subsection. Later, we will generalize the construction to higher-dimensional manifolds.

\subsection{TQEC on 2-Dimensional Manifolds}
\label{sec:TQEC_on_2_manifolds}

In this paper, we considered the ``loop-based TQEC'', where logical $\sigma_X$ is defined by operating $\sigma_X$ on any cycles $\gamma\in C_1(M; \mathbb{Z}_2)$, such that $\gamma \subset E$, where $\partial[\gamma]=0$ and $[\gamma]\ne 0$; a logical $\sigma_Z$ is defined to operate on any co-cycles, $\beta\in C^1(M; \mathbb{Z}_2)$, such that $\beta \subset E$, where $\delta[\beta]=0$ and $[\beta]\ne 0$, just as we defined in section~\ref{sec: definition of TQEC}. Then, we have the following theorem.

\begin{theorem}
    A 2-dimensional closed and compact manifold can be used as a TQEC code for qubits if and only if there exist at least two closed loops with an odd number of intersections.
    \label{Theorem: TQEC intersection}
\end{theorem}

The proof is in Appendix~\ref{appendix:Proofs}. With the theorem~\ref{Theorem: TQEC intersection}, and applying basic algebraic topology properties of surfaces, the first trivial observation is that all simply connected surfaces cannot be used for TQEC purposes. We formalize it using the following definition and a corollary~\cite{Algebraic_Topology1}.

\begin{definition}
    A space $X$ path-connected if given any two points $A, B\in X$, there exists a continuous path $P(t)$, $t\in[a, b]$, that $P(a)=A, P(b)=B$. 
    
    A space $X$ is said to be simply connected if it is path-connected and every loop in $X$ can be shrunk continuously to a point in $X$. Equivalently, a space $X$ is simply connected if and only if there is a unique homotopy class of paths connecting any two points in $X$.
\end{definition}

\begin{corollary}
    A simply connected closed surface cannot be used for TQEC purposes.
    \label{Theorem: Simply connected surface}
\end{corollary}

The proof is demonstrated in Appendix~\ref{appendix:Proofs}. As a result, simply connected surfaces, such as spheres ($S^2$), cannot be used for TQEC purposes. This may seem rather trivial, but more theorems are building up from it.

The relationship between the number of intersection points and the TQEC ability is stated in Theorem~\ref {Theorem: TQEC intersection}. Also, in section~\ref{topological_properties}, we examined, on a manifolds $M$, the relation between the intersection points of sub-manifolds $A, B\subset M$ and the homology $H_i(M;\mathbb{Z}_2), H^i(M;\mathbb{Z}_2)$. Therefore, we can rephrase Theorem~\ref {Theorem: TQEC intersection} using homology and cohomology theory, which gives the following more fundamental theorem of TQEC.

\begin{theorem}
A closed 2-manifold $M$ can be used for TQEC purposes if and only if the group $H_1(M; \mathbb{Z}_2)$ is non-trivial. Furthermore, the number of encoded qubits equals the rank of $H_1(M; \mathbb{Z}_2)$, or equivalently, the first Betti number $b_1(M;\mathbb{Z}_2)$.
\label{main_theorem}
\end{theorem}

The proof of this theorem is in Appendix~\ref{appendix:Proofs}. We can visualize the effect of the co-homology cycle on the homology cycle using the Toric code as an example, for a torus $T^2$. In Figure~\ref{fig:X_Z_Logical_intersection}, where the blue cycle corresponds to $\gamma\in C_1(T^2; \mathbb{Z}_2), \partial\gamma = 0, [\gamma]\ne0$, and the red cycle corresponds to $\beta\in C^1(T^2; \mathbb{Z}_2), \delta\beta = 0, [\beta]\ne0$. Since both $\gamma$ and $\beta$ are non-trivial, they intersect at an odd number of sites, and the parity of the intersection site is invariant under curve homotopy. Therefore, the torus can be used as a TQEC code for qubits.

Theorem~\ref{main_theorem} characterizes the logical qubit encoding capabilities of different 2-manifolds. The 2-manifolds can be classified into two primary categories. The first consists of orientable surfaces $M_g$ of genus $g$, constructed iteratively via the connected sum operation with a torus, i.e., $M_g = M_{g-1} \# T^2$, see Appendix~\ref{appendix:connected_sum} for a formal definition of connected sum. By convention, $M_g$ is equivalently denoted as $gT^2 := (g-1)T^2 \# T^2$. The second category comprises non-orientable surfaces, denoted as $gP^2$, defined recursively through connected sums with the real projective plane: $gP:= (g-1)P \# \mathbb{R}P^2$ for $g \geq 1$~\cite{HomologicalQEC1}. Another famous theorem in Algebraic Topology, which has been discussed widely~\cite{2d_surface_homeo, HomologicalQEC1}, is quoted below. 

\begin{theorem}
    Any orientable surface is homeomorphic to $gT^2$ for some integer $g \ge 0$, with $0T^2:=S^2$. Any non-orientable surface is homeomorphic to $gP^2$ for some integer $g\ge1$.
\end{theorem}

The proof of the theorem can be found in many Topology textbooks, for instance~\cite{2d_surface_homeo}. With this knowledge, we illustrate several examples below.
\begin{enumerate}
    \item \textbf{2-dimensional real projective plane $\mathbb{R}\text{P}^2$:} This is an non-orientable surface, The homology group of $\mathbb{R}\text{P}^2$ is $H_1(\mathbb{R}\text{P}^2; \mathbb{Z}_2)=\mathbb{Z}_2$, and since $\mathbb{R}\text{P}^2$ is a topological manifold, by poincaré duality, $H^1(\mathbb{R}\text{P}^2; \mathbb{Z}_2)=\mathbb{Z}_2$. Therefore, we can use the surface of $\mathbb{R}\text{P}^2$ to encode 1 qubit, as formerly discussed in~\cite{RP2code1}.

    \item \textbf{The Klein bottle $K$:} This is another non-orientable surface embedded in $\mathbb{R}^3$. Through a simple calculation, we have $H_1(K;\mathbb{Z}_2)=\mathbb{Z}_2\oplus\mathbb{Z}_2$, and by Poincaré duality, $H^1(K)=\mathbb{Z}_2\oplus\mathbb{Z}_2$. Thus, with a cell complex embedded on $K$, we should be able to define a TQEC code for 2 qubits.
    
    \item \textbf{The surface $M_g$ with genus $g$:} The first homology group of the surface $H_1(M_g; \mathbb{Z}_2)=\mathbb{Z}_2^{2g}$. Thus, by theorem~\ref{main_theorem}, it can encode $2g$ qubits. This matches with previous results, that the Hilbert space dimension of the code space of $M_g$ is $2^{2g}$~\cite{Toriccode1, TQEC_review1}.

    \item \textbf{The surface $gP^2$:} The first homology group of the surface $H_1(gP; \mathbb{Z}_2)=\mathbb{Z}_2^{g}$. Thus, by theorem~\ref{main_theorem}, it can encode $g$ qubits. Also, when $g = 2$, the surface $\mathbb{R}P^2\#\mathbb{R}P^2\cong K$ .
\end{enumerate}

Remark that the 2-dimensional real projective plane $\mathbb{R}\text{P}^2$, the Klein bottle $K$, and their generalization $gP^2$ surface are all non-orientable, and they cannot be used for TQEC codes for qudits with dimension $D>2$ as proven in~\cite{HomologicalQEC1}. If we recall the definition of TQEC codes as in section~\ref{sec: definition of TQEC}, it is defined on 2-dimensional cell complexes, which can be represented as a graph, $G = (V, E, F)$. In the discussion above, we only considered the cell complexes that are embedded in 2-manifolds, because of the properties of manifolds. However, in theory, TQEC codes can be defined on cell complexes not embedded in manifolds. The problem is, if a cell complex $X$ is not embeddable on any manifold, then a lot of nice properties of manifolds, like the Poincaré duality as in theorem~\ref{Poincaré}, do not hold anymore. We will discuss the formation of TQEC on general 2-dimensional cell complexes in Section~\ref{sec:TQEC_on_2_complex}.

There has been previous discussion on topological codes on non-orientable manifolds, including $\mathbb{R}P^2$ and Klein bottles $K$. Ref.~\cite{2D_TQEC_Clifford1} studies toric, surfaces, and the honeycomb Floquet codes on non-orientable surfaces with cross‑cap defects and shows that the emergent anyon‑exchange symmetry between $e$ type and $m$ type particles, which can be implemented as a constant‑depth local logical gate. Ref.~\cite{Higher_dimension_TQEC_Clifford1, Higher_dimensional_TQEC2} discussed higher-dimensional TQEC codes, for instance, the $\mathbb{Z}_2$ homological codes on $K\times S^1$. The discussion is from a more condensed matter flavor. Later, we will demonstrate the TQEC ability on higher-dimensional manifolds and the intersection properties.

\section{TQEC Codes on Higher Dimensional Manifolds}
\label{sec: TQEC Codes on Higher Dimensional Manifolds}

After the above discussion of TQEC on 2-manifolds, the primary question is whether higher-dimensional manifolds can be effectively utilized for TQEC purposes. In this subsection, we aim to extend the framework of TQEC beyond its conventional reliance on 2-manifolds. While applying higher-dimensional manifolds and cell complexes may present significant challenges in experimental implementations, their exploration holds substantial theoretical value. This investigation contributes to a deeper understanding of the mathematical foundations of TQEC and broadens the scope of potential topological structures applicable to QEC.

We start the section by presenting the construction of TQEC codes on higher-dimensional manifolds. The idea is that when generalizing the TQEC from 2-manifolds to an $n$-manifold $M$, the code is defined upon a cell complex $X_M$ embedded on the manifold $M$, where a physical qubit is placed on each $i$-cell $D^i_\alpha$ for $1\le i\le n-1$, where $\alpha$ denotes the indices, and logical qubits are encoded separately on each dimension.

Logical $\sigma_X$ is operated on homology cycles $\gamma\in C_i(M; \mathbb{Z}_2)$, that $\partial\gamma=0, [\gamma]\ne 0$. And logical $\sigma_Z$ is operated on cohomology cycles, $\beta\in C^i(M; \mathbb{Z}_2)$, that $\delta\gamma=0, [\gamma]\ne 0$. As shown in equation~\ref{pairing}, if $i+j=n$, then the pairing between $H_i(M, \mathbb{Z}_2)$ and $H_j(M,\mathbb{Z}_2)$ over $\mathbb{Z}_2$ is well-defined, meaning that it is possible to find $n$-manifold $M$ with $n>2$, and for some $1\le i \le n-1$ we can define logical $\sigma_X$ and $\sigma_Z$ operations that their modulo-2 intersections is well-defined, and as from the Poincaré duality~\ref{Poincaré}, $H^i(M, \mathbb{Z}_2)\cong H_{n-i}(M, \mathbb{Z}_2)$. Here, we will develop a general theory of TQEC using an $n$-dimensional cell complex embedded on an $n$-manifold.

\begin{theorem}
    A cell complex $X_M$ embedded on an $n$-dimensional closed and compact manifold $M$ can encode qubits on the $i$-th dimensional skeleton $X^i$, where $1 \le i \le n-1$, if and only if the homology group $H_i(M; \mathbb{Z}_2)$ is non-trivial. The number of qubits that can be encoded on the $i$-cycles equals the rank $H_i(M; \mathbb{Z}_2)$, or equivalently, the $i$-th Betti number $b_i(M;\mathbb{Z}_2)$.
    \label{thm:high_dim_main_theorem}
\end{theorem}

The definition of the $n$-th skeleton of the cell complex $X_M$ is formally defined in Appendix~\ref{appendix:simplex_cell_complex}. Remark: The reason why we only consider the $H_i(M)$ for $1\le i \le n-1$ is that all closed and path-connected manifolds having $H_0(M; \mathbb{Z}_2)\approx \mathbb{Z}_2$, and the 0-simplex is essentially a collection of points. By definition, points have no boundary~\cite{Algebraic_Topology1, Algebraic_Topology2}, but boundary operations are considered when defining the stabilizer checks in TQEC. Since we cannot formulate boundary maps and stabilizer checks on points, considering $H_0(M)$ and $H_n(M)$ is not useful for TQEC. This theorem is proven in Appendix~\ref{appendix:Proofs}.

Applying this theorem, some interesting examples will arise. Before we get into examples, please note that both the solid ball and the solid torus ($D^2\times S^1$) are not closed manifolds. Therefore, the Poincaré duality~\ref{Poincaré} does not apply to them. There are some examples of closed and compact $n$-manifolds, for instance, the $n$-sphere $S^n$. The homology group of spheres~\cite{Algebraic_Topology1, Algebraic_Topology2}, 
\begin{equation}
\begin{split}
  H_k(S^n;\mathbb{Z}_2) &=
    \begin{cases}
      \mathbb{Z}_2 & \text{if $k=0$ or $n=k$}\\
      0 & \text{otherwise}
    \end{cases}
\end{split}
\end{equation}

As demonstrated in Theorem~\ref{thm:high_dim_main_theorem}, the $n$-sphere $S^n$ cannot be utilized for TQEC purposes, irrespective of the value of $n$. 

\subsection{Examples of higher-dimensional TQEC: The 3-Torus Code}

The 3-torus, $T^3 \cong S^1 \times S^1 \times S^1$, is an example of a three-dimensional manifold suitable for TQEC for qubits. The 3-torus is a closed and compact 3-manifold where Poincaré duality applies. Its homology groups can be computed explicitly as described in~\cite{T3structure1}, making it a promising candidate for exploring TQEC codes for qubits.

\begin{figure}
    \centering
    \includegraphics[width=0.5\linewidth]{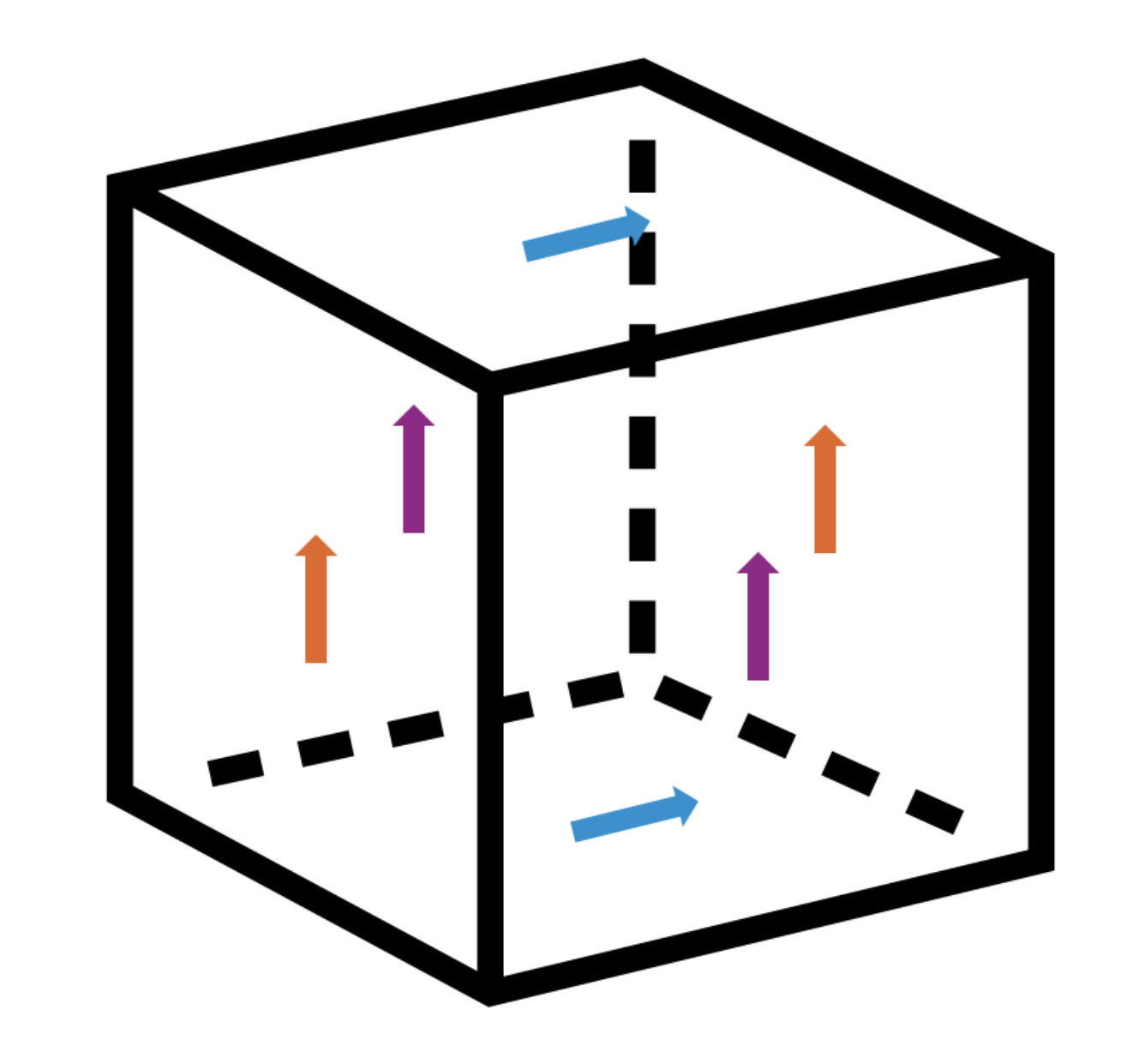}
    \caption{\raggedright A 3-torus ${T}^3$ can be obtained by identifying the opposite faces of a cube.}
    \label{fig:3torus}
\end{figure}

\begin{equation}
  H_k(T^3; \mathbb{Z}_2) =
    \begin{cases}
      \mathbb{Z}_2 & \text{if $k=0, 3$}\\
      \mathbb{Z}_2\oplus\mathbb{Z}_2\oplus\mathbb{Z}_2 & \text{if $k=1, 2$}\\
      0 & \text{otherwise}
    \end{cases}
\end{equation}

Thus, qubits can be encoded on the 3-torus $T^3$. Geometrically, the 3-torus $T^3$ can be obtained by identifying opposite faces of a cube, just as shown in Figure~\ref{fig:3torus}. As from the theorem~\ref{thm:high_dim_main_theorem}, for $k=1$, it can encode 3 qubits, as $H_1(T^3;\mathbb{Z}_2)$ has 3 independent generators; similarly, for $k=2$, it can encode 3 qubits separately.

\begin{figure}
    \centering
    \includegraphics[width=0.8\linewidth]{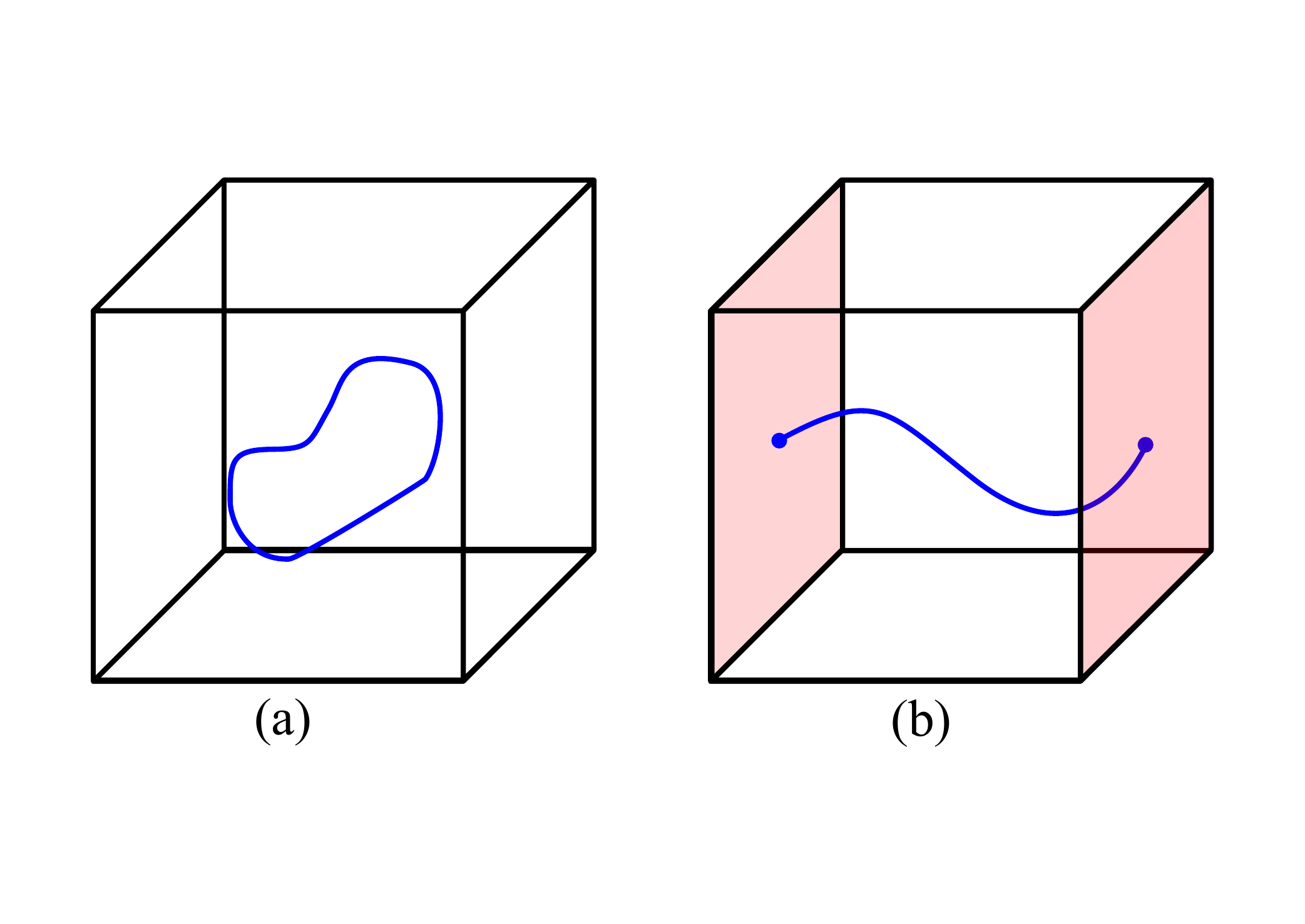}
    \caption{\raggedright Examples of possible closed loops $\gamma^1$ in the 3-torus $T^3$, where $\gamma^1 \in C_1(T^3;\mathbb{Z}_2), \partial\gamma^1=0$. \textbf{(a)} $[\gamma^1]=0$, which is a contractible loop. \textbf{(b)} $[\gamma^1]\ne0$, which is a non-contractible loop.}
    \label{fig:loops_in_T3}
\end{figure}

\begin{figure}
    \centering
    \includegraphics[width=0.8\linewidth]{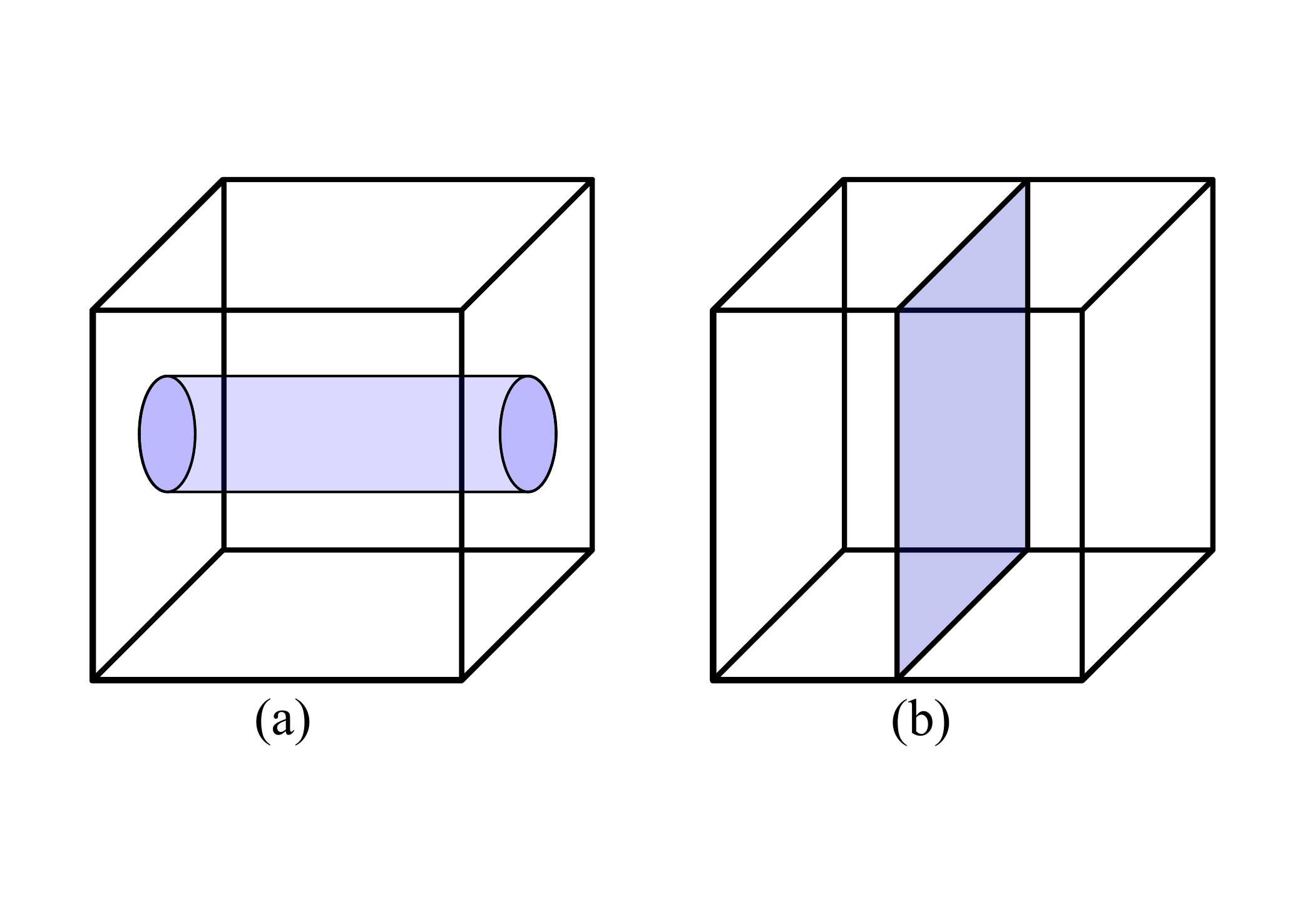}
    \caption{\raggedright Examples of possible closed 2-torus $\gamma^2$ in the 3-torus $T^3$, where $\gamma^2 \in C_2(T^3;\mathbb{Z}_2), \partial\gamma^2=0$. \textbf{(a)} $[\gamma^2]=0$, which is a compressible 2-torus. \textbf{(b)} $[\gamma^2]\ne0$, which is a non-compressible 2-torus.}
    \label{fig:torus_in_T3}
\end{figure}

Firstly, for $ i =\{1, 2\}$, we would like to provide some examples of  $\gamma^i\in C_i(T^3;\mathbb{Z}_2)$ and $\partial\gamma^i=0 $ for both $[\gamma^i]= 0$ and $[\gamma^i] \ne 0$. Figure~\ref{fig:loops_in_T3} shows examples of contractible and non-contractible loops in $T^3$, $\gamma^1\in C_1(T^3;\mathbb{Z}_2), \partial\gamma^1=0$, for both $[\gamma^1]=0$ and $[\gamma^1]\ne0$. And Figure~\ref{fig:torus_in_T3} shows examples of compressible and non-compressible 2-torus $T^2$ in $T^3$, $\gamma^2\in C_2(T^3;\mathbb{Z}_2), \partial\gamma^2=0$, for both $[\gamma^2]=0$ and $[\gamma^2]\ne0$. For simplicity, denote $\{\gamma^j_0\} := \{\gamma^j:[\gamma^j]=0\}$, and $\{\gamma_1^j\} := \{\gamma^j\}\setminus \{\gamma^j_0\}$.

It is relatively trivial to see that $\gamma_0^1$ and $\gamma_0^2$ will always intersect at an even number of sites. Furthermore, the number of intersection sites of $\gamma_i^1$ with $\gamma_0^2$ and $\gamma_0^1$ with $\gamma_i^2$ is always even. The only exception is the number of intersections between $\gamma_1^1$ and $\gamma_1^2$, which can be odd.

\subsection{Implementing 3-Torus Code: Encoding on Edges}

\begin{figure}
    \centering
    \includegraphics[width=0.8\linewidth]{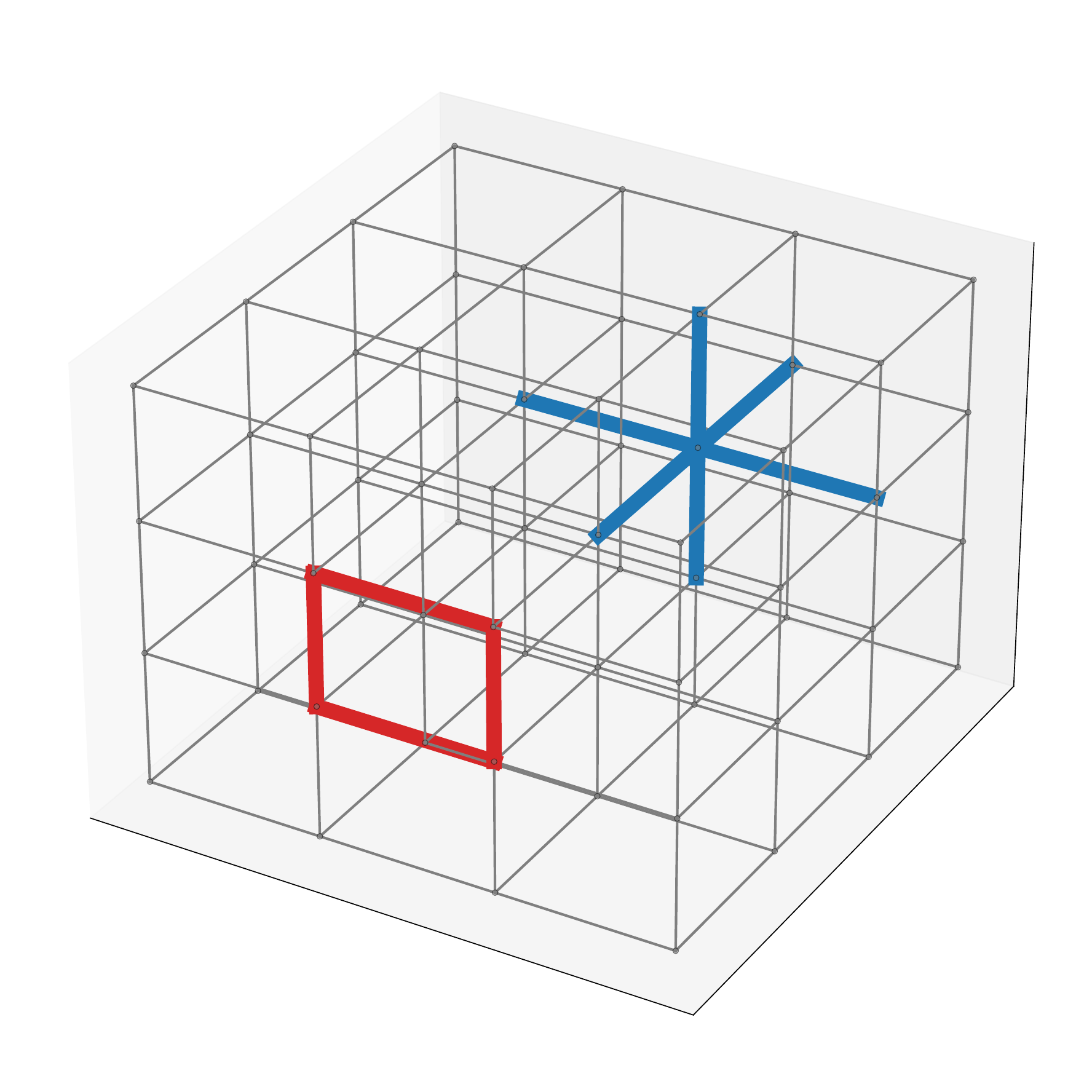}
    \caption{\raggedright The stabilizer checks defined on the 3-torus code for 1-dimensional cells. We check the $\sigma_Z$ of 6 adjacent qubits, the $A_i$ stabilizers, drawn in blue for each vertex. For each face, we check the $\sigma_X$ of 4 adjacent qubits, the $B_j$ stabilizers, drawn in red.}
    \label{fig:stabilizer_checks_on_3torus}
\end{figure}

\begin{figure}
    \centering
    \includegraphics[width=\linewidth]{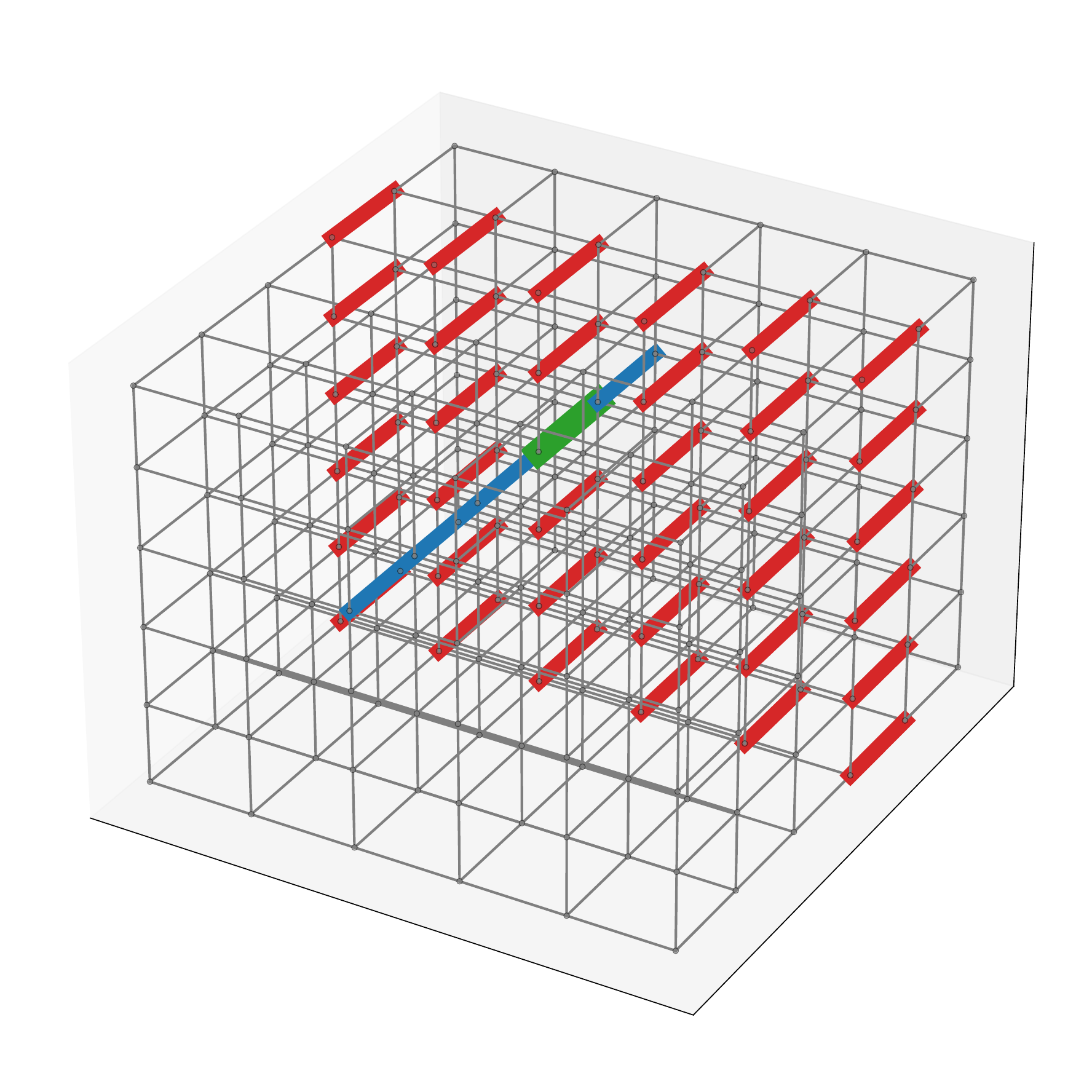}
    \caption{\raggedright The 3-Torus code, a qubit is placed on each edge. The blue lines represent qubits forming a non-contractible loop, and the red lines represent qubits forming a non-compressible 2-torus. They intersect at 1 site, which is drawn in green. And the blue qubits represent a possible logical $\sigma_X$, and the red qubits represent a possible logical $\sigma_Z$.}
    \label{fig:intersection_on_3torus}
\end{figure}

Here we provide an example to encode 3 qubits on the 1-cells, edges, of the 3-torus $T^3$, which we named the ``3-torus code''. Similar to the original toric code, we first discretize the lattice into an $N\times N\times N$ 3-dimensional grid with periodic boundary conditions as in Figure~\ref{fig:3torus}, which can be described by a 4-tuple $(V, E, F, C)$, for the set of vertices $V$, the set of edges $E$, the set of faces $F$ and the set of cubes $C$. 

Firstly, we place a physical qubit on every edge $e\in E$. In the next subsection, we will further show that we can encode more qubits by placing a physical qubit on every face $f\in F$. Then we define two types of stabilizer checks:

\begin{itemize}
    \item Vertex stabilizers: For every vertex $v\in V$
    \begin{equation}
        A_i^1 = \sigma_Z(\delta v)
        \label{duel_cell_operators}
    \end{equation}
    \item Face stabilizers: For every face $f\in F$
    \begin{equation}
        B_j^1 = \sigma_X(\partial f)
        \label{cell_operators}
    \end{equation}
\end{itemize}

Similar to the original toric code~\cite{Toriccode1}, the vertex stabilizer is defined to check the $\sigma_Z$ of the $6$ qubits on the edges adjacent to a given vertex; the blue edges in Figure~\ref{fig:stabilizer_checks_on_3torus} show an example. And the face stabilizer is defined to check the $\sigma_X$ of the $4$ qubits on the boundary of a given face, an example is demonstrated in red in Figure~\ref{fig:stabilizer_checks_on_3torus}. $A_i^1$ and $B_j^1$ commute for all $i, j$, because they either coincide on 2 edges or do not coincide at all. Also, not surprisingly,
\begin{equation}
    \Pi_i A_i^1=\Pi_j B_j^1=\mathbb{I}
\end{equation}

\begin{figure}
     \centering
     \begin{subfigure}[b]{0.5\textwidth}
         \centering
         \includegraphics[width=0.8\textwidth]{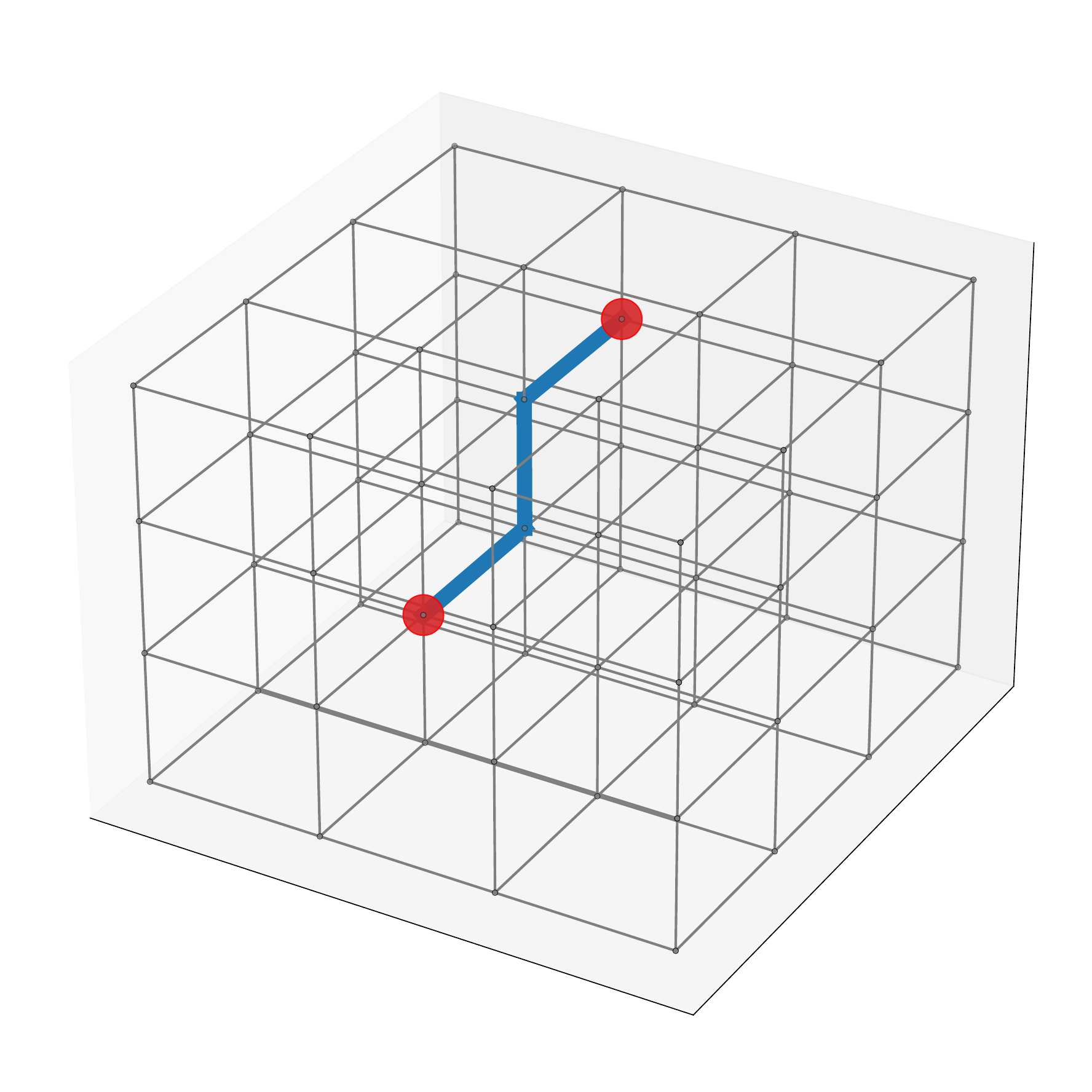}
         \caption{}
         \label{fig:3_torus_mistake_node}
     \end{subfigure}
     \hfill
     \begin{subfigure}[b]{0.5\textwidth}
         \centering
         \includegraphics[width=0.8\textwidth]{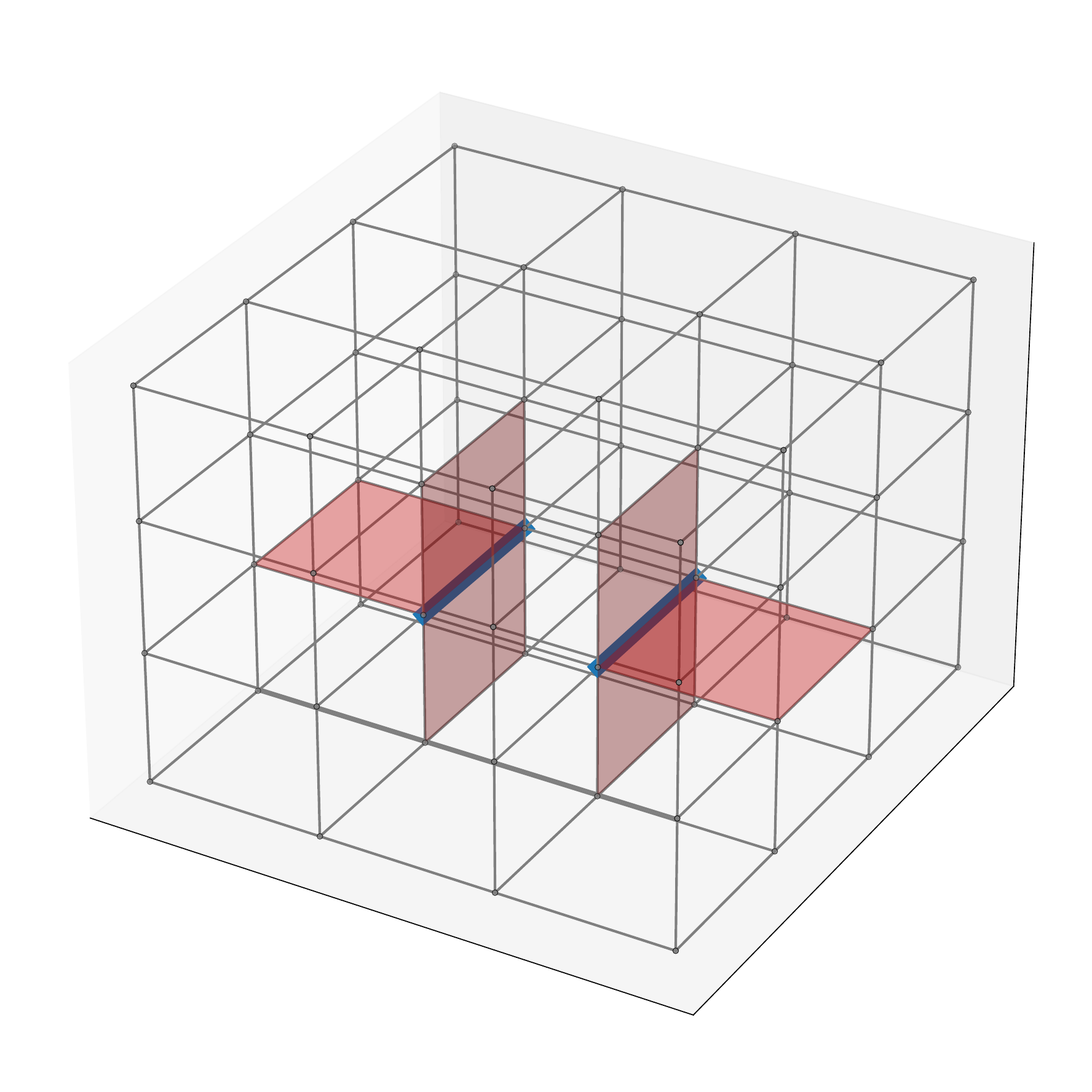}
         \caption{}
         \label{fig:3_torus_mistake_face}
     \end{subfigure}
     \caption{\raggedright Illustration of the stabilizer failures when physical qubits are subjected to $\sigma_X$ or $\sigma_Z$ errors. \textbf{(a)} $\sigma_X$ error on the qubits labelled blue, then 2 Vertex stabilizers checks will fail, which are labelled red. \textbf{(b)} $\sigma_Z$ error on the qubits labelled blue, then 6 Face stabilizers will fail, which are labelled red. There will always be an even number of stabilizer failures, which can then be corrected by various matching algorithms.}
     \label{fig:3_torus_mistake}
\end{figure}

Thus, the above vertex and face stabilizers are well-defined. And similar to the original toric code~\cite{Toriccode1}, a $\sigma_X$ error happening on any qubit will cause 2 mismatched $A_i^1$ stabilizers, where the error correction process can be done by various matching algorithms, such as Maximum Likelihood matching or Minimum Weight Perfect Matching~\cite{MLM1}. This process is similar to the original toric code, albeit it is more complicated, arising from the 3-dimensional lattice structure. And a $\sigma_Z$ error happening on any qubit will cause 4 mismatched $B_j^1$ stabilizers, because each edge is adjacent to 4 different faces. Additional $\sigma_Z$ errors will create mistakes in pairs, and we can do the matching accordingly, which is again more complicated. 

Figure~\ref{fig:3_torus_mistake} demonstrates an example of the failures of stabilizer checks when several physical qubits are subjected to $\sigma_X$ or $\sigma_Z$ errors. Figure~\ref{fig:3_torus_mistake_node} shows 3 qubits suffering from $\sigma_X$ error labelled blue. Then, there will be 2 failures in the vertex stabilizer $A^1_\alpha$ checks, where the corresponding vertices are labelled red. Figure~\ref{fig:3_torus_mistake_face} shows that 2 qubits suffer from $\sigma_Z$ error labelled blue, and in this case, there will be 6 failures in the face stabilizers, $B^1_\alpha$, where the corresponding faces are labelled red. Notice that for both cases, the stabilizer failures always appear in an even number of sites; thus, the error correction process can be theoretically done via different matching functions, as discussed above~\cite{MLM1, Google2}.

Figure~\ref{fig:intersection_on_3torus} illustrates the logical $\sigma_X$ and $\sigma_Z$ graphically. The blue lines in Figure~\ref{fig:intersection_on_3torus} forms a non-contractible loop $\gamma^1 \in C_1(T^3;\mathbb{Z}_2), \partial\gamma^1=0$ and $[\gamma^1]\ne 0$. The red lines forms a non-compressible 2-torus $\gamma^2 \in C_2(T^3;\mathbb{Z}_2), \partial\gamma^2=0$ and $[\gamma^2]\ne 0$. And they intersect at exactly 1 site, which is the edge drawn in green. These properties allow us to represent a logical $\sigma_X$ on the blue qubits and a logical $\sigma_Z$ on the red qubits.

\subsection{Implementing 3-Torus Code: Encoding on Faces}

\begin{figure}
    \centering
    \includegraphics[width=0.8\linewidth]{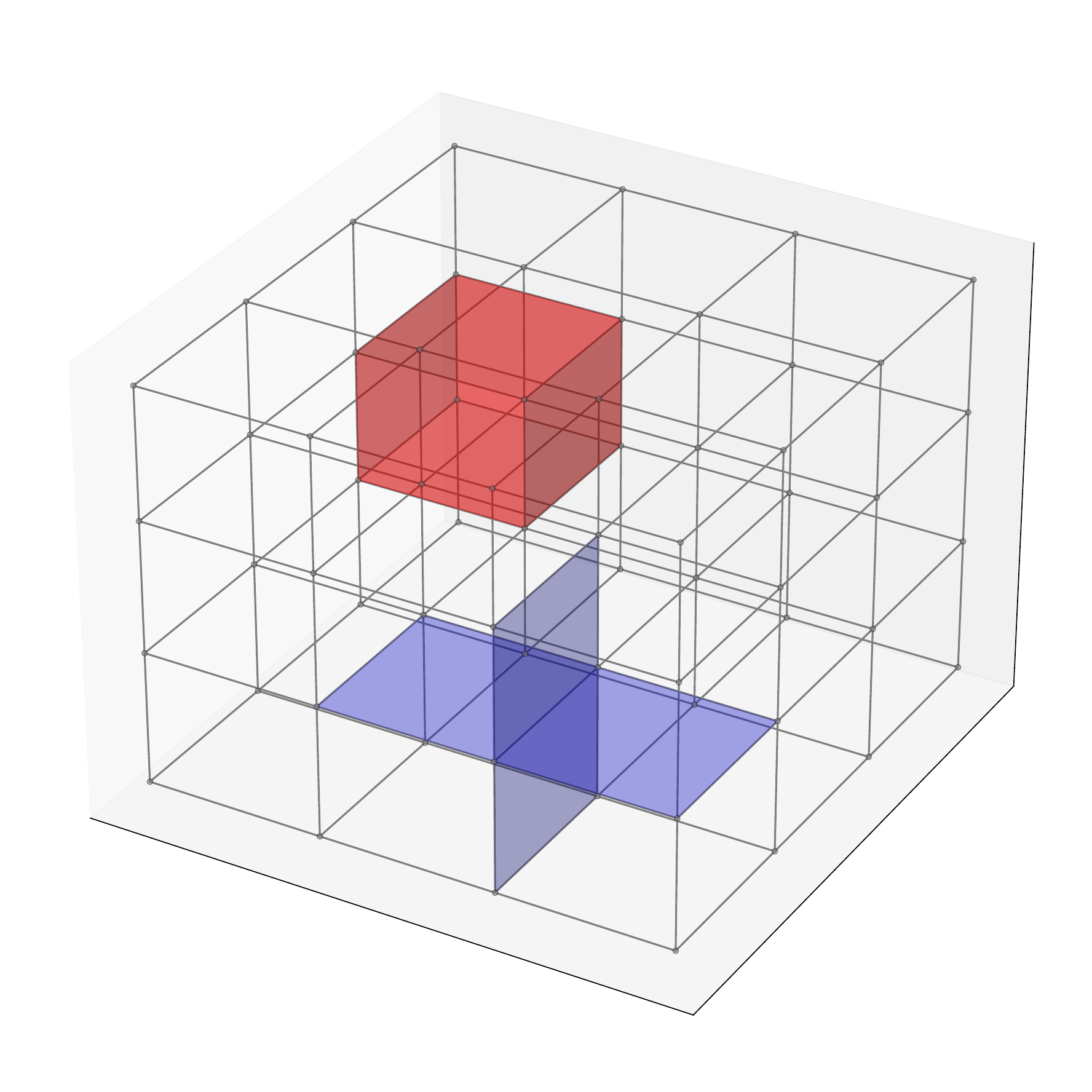}
    \caption{\raggedright The stabilizer checks defined on the 3-torus code for 2-dimensional cells. For each edge, we check the $\sigma_Z$ of 4 qubits reside on adjacent faces, the $A_i$ stabilizers, drawn in blue. For each cube, we check the $\sigma_X$ of 6 qubits residing on the boundary faces, the $B_j$ stabilizers, drawn in red.}
    \label{fig:stabilizer_checks_on_3torus_dim2}
\end{figure}

\begin{figure}
    \centering
    \includegraphics[width=\linewidth]{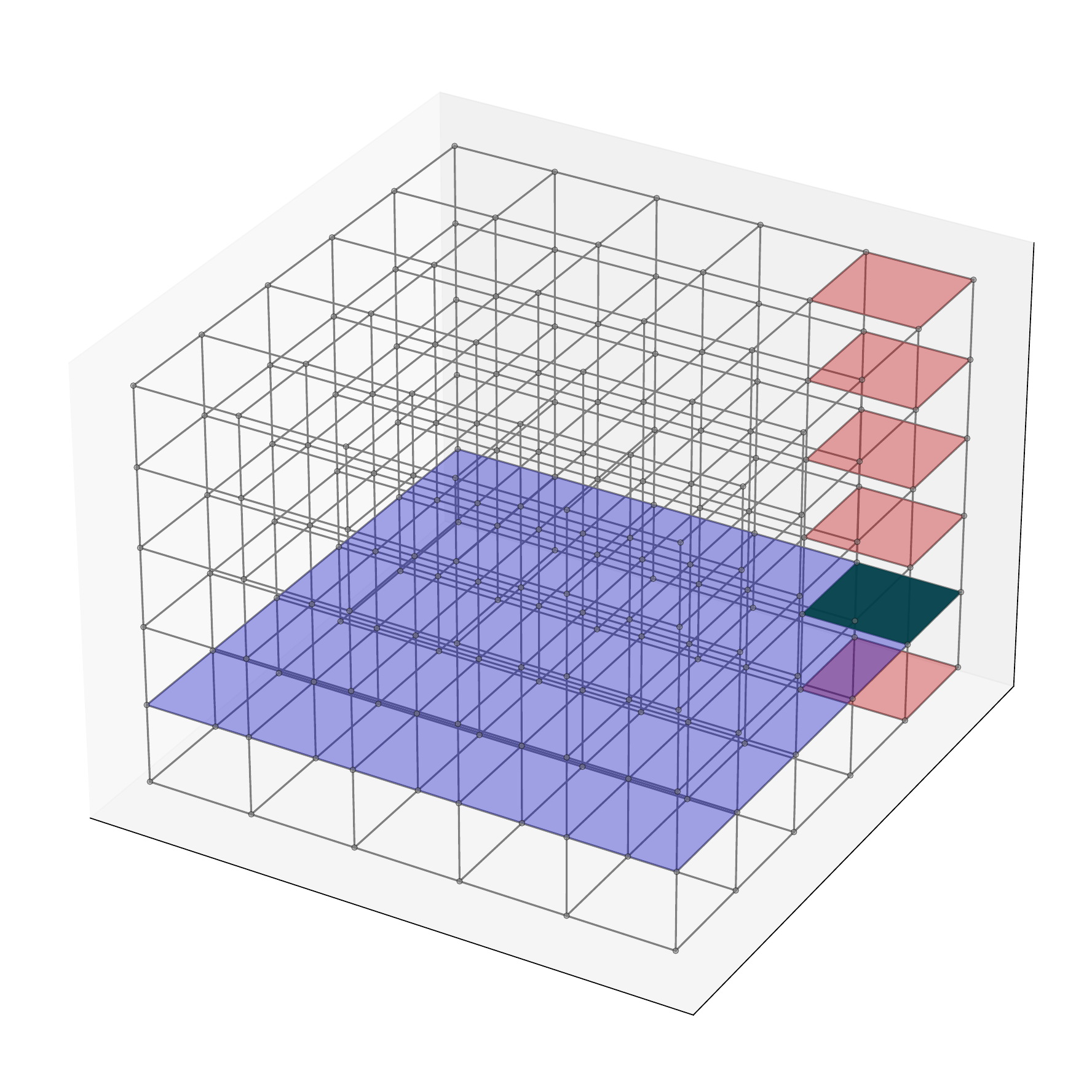}
    \caption{\raggedright The 3-Torus code encoding on 2-dimensional cells, a qubit is placed on each face. The blue faces represent qubits forming a non-compressible 2-torus, and the red faces represent qubits forming a non-contractible loop. They intersect at 1 face, which is drawn in green. And the blue qubits represent a possible logical $\sigma_X$, and the red qubits represent a possible logical $\sigma_Z$.}
    \label{fig:intersection_on_3torus_dim2}
\end{figure}

We have shown above that the 3-torus code can encode 3 qubits on the 1-cells. Here, we demonstrate that the 3-torus code can encode 3 extra qubits on the 2-cells, or faces. Note that the 3 qubits encoded on the faces are independent of the 3 qubits that are on the edges. 

We put 1 qubit on every face $f\in F$ of the 3-dimensional lattice. There are two types of stabilizers defined, which are illustrated in Figure~\ref{fig:stabilizer_checks_on_3torus_dim2}:
\begin{itemize}
    \item Edge stabilizers: For every edge $e\in E$
    \begin{equation}
        A_i^2 = \sigma_Z(\delta e)
    \end{equation}
    \item Cube stabilizers: For cube $c\in C$
    \begin{equation}
        B_j^2 = \sigma_X(\partial c)
    \end{equation}
\end{itemize}

The edge stabilizer $A^2_i$ is defined to check the $\sigma_Z$ of the $4$ qubits on the faces adjacent to a given edge, an example is drawn in blue in Figure~\ref{fig:stabilizer_checks_on_3torus_dim2}. The cube stabilizer $B^2_j$ is defined to check the $\sigma_X$ of the $6$ qubits on the boundary faces of a given cube; an example is drawn in red in Figure~\ref{fig:stabilizer_checks_on_3torus_dim2}. $A_i^2$ and $B_j^2$ commute for all $i, j$, because they either do not have any qubits in common or have 2 qubits in common. Not surprisingly,
\begin{equation}
    \Pi_i A_i^2=\Pi_j B_j^2=\mathbb{I}
\end{equation}

Thus, the above vertex and face stabilizers are well-defined. And similar to the original toric code~\cite{Toriccode1}, a $\sigma_X$ error happening on any qubit will cause 4 mismatched $A_i^2$ stabilizers because a face is adjacent to 4 edges, each of which defines a stabilizer. Matching can be done similarly to the original toric code, albeit it is more complicated. And a $\sigma_Z$ error happening on any qubit will cause 2 mismatched $B_j^2$ stabilizers, because each face is only adjacent to 2 cubes, and we can do the matching accordingly, similar to that of the 1-cell case, which is again more complicated. 

Meanwhile, Figure~\ref{fig:intersection_on_3torus_dim2} illustrates the logical $\sigma_X$ and $\sigma_Z$ on 2-dimensional cells graphically. Where the blue surfaces represent a non-compressible 2-torus in $C_2(T^3;\mathbb{Z}_2)$ and the red surfaces represent a non-compressible loop in $C_1(T^3;\mathbb{Z}_2)$.They intersect at 1 qubit, which is labeled green in the figure. Thus, we can treat the blue surfaces as the logical $\sigma_X$ and the red surfaces as the logical $\sigma_Z$ operators.

In summary, for a 3-torus code, when we use a 3-dimensional lattice, we can place a qubit on every edge and every face. Then, 3 qubits can be encoded on edges, and separately, 3 qubits can be encoded on faces. Therefore, 6 qubits can be encoded in total. Remark that this is only one possible implementation of the TQEC in higher-dimensional manifolds, and there are other proposals, such as~\cite{3DToric_code1, 3DToric_code2}. Theoretically, one can choose any manifold with non-trivial homology groups and use any cell-complex encoded on those manifolds for TQEC for qubits.

\subsection{3-Torus Code Performance Analysis}

In this subsection, we would like to analyse the performance of the 3-torus code that we proposed. The encoding on the 1-dimensional and 2-dimensional cells is independent. First, a quick remark, one can separately consider the distance of the code for both $\sigma_X$ and $\sigma_Z$ errors of the code, $d_x$ and $d_z$. Where $d_x$ or $d_z$ are the minimum weights of a logical $\sigma_X$ or $\sigma_Z$ operator. It quantifies the code's ability to detect and correct bit-flip or phase-flip errors. $d_x$ and $d_z$ equals the number of physical qubits in the shortest logical $\sigma_X$ and $\sigma_Z$ operation. 

Here, we consider the case of qubits encoded on edges. In a general $u\times v\times w$ lattice for $u, v, w\in \mathbb{Z}^+$, there are in total $4uvw$ edges in the lattice. Hence, there will be $4uvw$ physical qubits and three logical qubits can be encoded, where the $\sigma_Z$ distance of the 3 qubits equals $u, v$, and $w$ respectively. Furthermore, the $\sigma_X$ distance of the 3 qubits equals $vw, wu$, and $uv$ respectively, where Figure~\ref{fig:intersection_on_3torus} has demonstrated the idea. In the special case of a cubic lattice with dimensions $N\times N\times N$, there are $4N^3$ edges, the distance for logical $\sigma_Z$ is $N$, and the distance for logical $\sigma_X$ is $N^2$. Hence, the code is a $[[4N^3, 3, N]]$.

When compared to the original toric code~\cite{Toriccode1}, which is a $[[2N^2, 2, N]]$ code. It might seem that the 3-torus code performs way worse as the number of physical qubits scales as $N^3$. That is not true, because the 3-torus code proposed is asymmetric in the $d_x$ and $d_z$ for an $N\times N\times N$ lattice. It can be easily understood with the help of the Figure~\ref{fig:intersection_on_3torus}, where the blue qubits represent a logical $\sigma_Z$ and the red qubits represent a logical $\sigma_X$. Thus, on an $N \times N \times N$ lattice, the logical $\sigma_Z$ has a distance of $N$, while a logical $\sigma_X$ has a distance of $N^2$. If we want to achieve the same performance on a toric code, then we need a rectangular lattice of dimension $N\times N^2$, which makes the number of physical qubits also scale with $N^3$. Therefore, the newly proposed 3-torus code performed similarly to the original toric code in terms of encoding rate.

\section{TQEC Codes on Arbitrary Dimensional Manifolds with Boundary}
\label{sec:TQEC_boundary}

Realizing different topologies in current experimental qubit systems poses significant experimental challenges due to the non-trivial topologies that usually require long-range connectivity and entanglement, which are experimentally expensive and suffer from low fidelity. Thus, with the current technology in industry and academia, most successful quantum codes are planar quantum codes, which include the below-threshold surface code that Google implemented~\cite{Google1, Google2}, the $[[144, 12, 12]]$ QLDPC code that IBM is trying to implement~\cite{qLDPC1, IBM1, IBM2}. Thus, studying how to build quantum codes defined on surfaces with boundaries has great experimental research value.

Several previous papers have analyzed the requirements and performance of codes with a planar topology with holes on the surface, for example~\cite{Surfacecode1, Surfacecodetheory1, HomologicalQEC1}. Especially, the authors of Ref.~\cite{Surfacecodetheory1} develop generalized surface codes on punctured surfaces with mixed open and closed boundaries by formulating a relative homology/cohomology framework that unifies and extends Kitaev’s toric code and its boundary variants, enabling denser storage of quantum information and planar layouts with significantly reduced overheads. 

In this section, we aim to explain this construction in two dimensions and then generalize it to all higher-dimensional TQEC codes, demonstrating methods for constructing codes on higher-dimensional lattices with boundaries. We will start by reviewing the codes on 2-dimensional surfaces with smooth and rough boundaries, which is well-studied by authors in Ref.~\cite{Surfacecodetheory1}. Then, we generalize their construction to higher-dimensional manifolds with boundary and provide concrete examples in the 3-dimensional case. We will present follow-up papers to provide a more thorough discussion on considering TQEC in higher-dimensional manifolds with both boundaries and holes with arbitrary dimensions.

\subsection{Surface Code with Mixed Boundary Conditions}

\begin{figure}
    \centering
    \includegraphics[width=\linewidth]{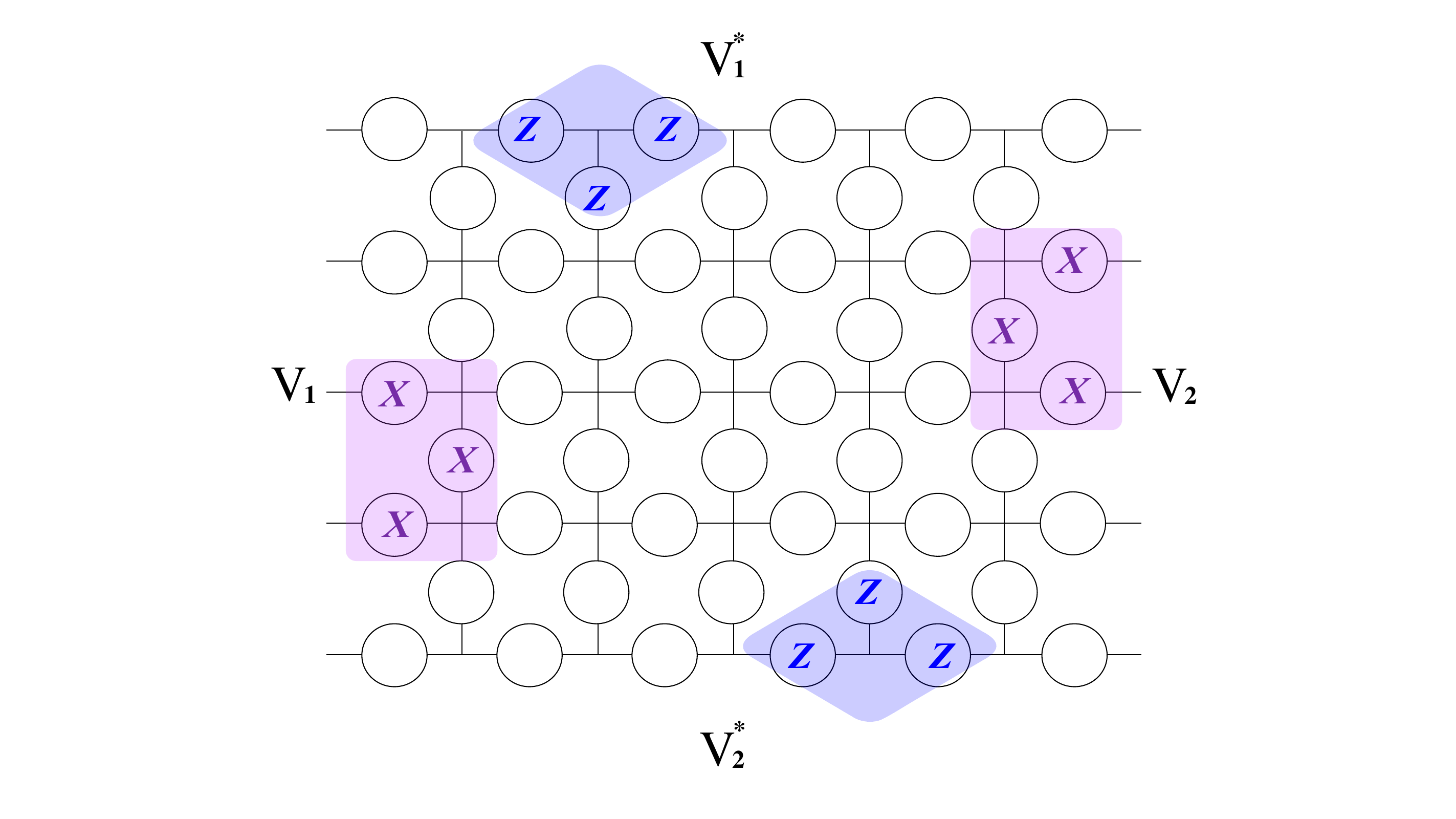}
    \caption{\raggedright The stabilizers of a surface code with 2 types of boundary conditions. A qubit is placed on each edge. Clearly, at the left and right boundaries, the plaquette stabilizers are only of weight-3 $\sigma_X$. Similarly, at the top and bottom boundaries, the vertex stabilizers are only of weight-3 $\sigma_Z$.}
    \label{fig:surface_code_stabilizer}
\end{figure}

\begin{figure}
    \centering
    \includegraphics[width=\linewidth]{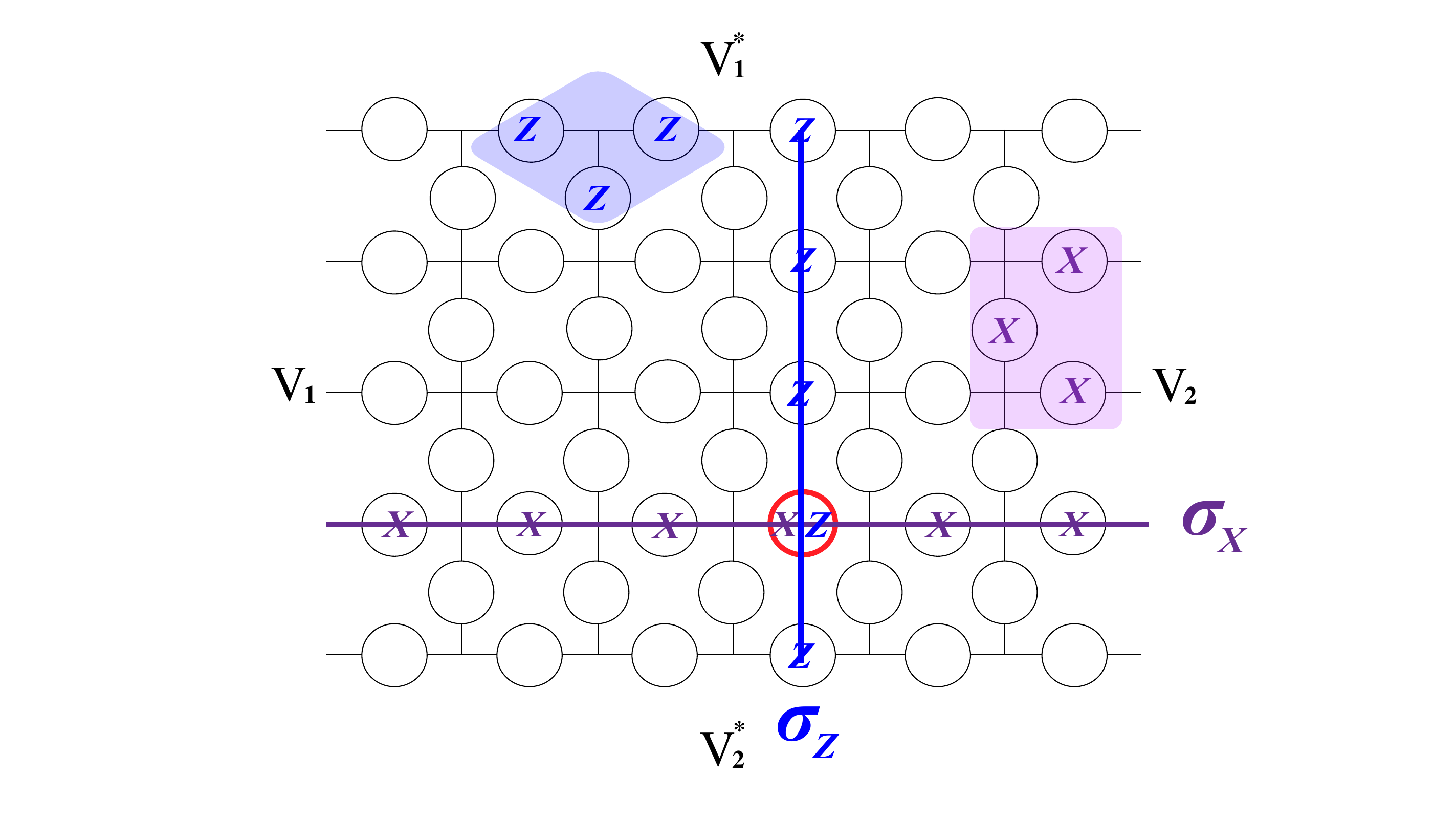}
    \caption{\raggedright Example of the operators of a surface code with 2 types of boundary conditions. Any logical operator commutes with all stabilizers. However, if there is a chain of $\sigma_X\ (\sigma_Z)$ operators that do not end on $ x$-boundary ($ z$-boundary), then they do not commute with some stabilizers, and errors will be spotted.}
    \label{fig:surface_code_operators}
\end{figure}

The surface code originates from Bravyi and Kitaev’s construction of quantum codes on a planar lattice with boundary conditions~\cite{Surfacecode1}. Their goal was to obtain a planar analogue of the toric code~\cite{Toriccode1}, which is experimentally more accessible: the toric code’s periodic boundary conditions effectively demand nonlocal connectivity (or, equivalently, long-range entanglement resources) that are difficult to realize in realistic hardware~\cite{Toriccode1}. By contrast, the surface code is defined on a planar lattice with boundaries, avoids periodic identification, and is therefore more amenable to experimental implementation.

A simply connected two-dimensional disk, however, is topologically insufficient for topological quantum error correction because its first homology group is trivial. In particular, on a disk, any two closed loops have an even mod-$2$ intersection number, consistent with Theorem~\ref{thm:high_dim_main_theorem_boundary} proved later. Consequently, to obtain nontrivial logical degrees of freedom, one must introduce appropriate boundary structure (or, more generally, defects) so that nontrivial relative homology classes exist.

Bravyi and Kitaev introduced two boundary types: (i) rough boundaries and (ii) smooth boundaries. In the convention adopted here---vertex stabilizers check $Z$-type stabilizers and plaquette stabilizers check $X$-type stabilizers---rough boundaries are also termed $x$-type boundaries: they preserve full weight-$4$ $Z$-type plaquette stabilizers while truncating adjacent $X$-type vertex stabilizers to weight $3$. Conversely, smooth boundaries are also termed $z$-type boundaries: they preserve full weight-$4$ $X$-type vertex stabilizers while truncating adjacent $Z$-type plaquette stabilizers to weight $3$. As illustrated in Figure~\ref{fig:surface_code_stabilizer}, the plaquette stabilizers along the left and right edges have weight $3$ and are of $X$-type, whereas the vertex stabilizers along the top and bottom edges have weight $3$ and are of $Z$-type. Accordingly, the left and right edges are $x$-type (rough) boundaries, and the top and bottom edges are $z$-type (smooth) boundaries.

As shown in Figure~\ref{fig:surface_code_operators}, a logical $X$ operator is supported on a nontrivial 1-chain whose endpoints lie on the $x$-type boundaries. Such supports are naturally described as relative 1-cycles; a precise definition is provided in Appendix~\ref{topological_properties}. Similarly, a logical $Z$ operator is supported on a nontrivial 1-cochain whose endpoints lie on the $z$-type boundaries, \emph{i.e.}, a relative 1-cocycle. By construction, any logical operator must commute with all stabilizers; in contrast, an $X$ (respectively $Z$) string whose endpoints do not terminate on an $x$-type (respectively $z$-type) boundary will anticommute with at least one stabilizer, producing a detectable syndrome and hence revealing an error.

The nomenclature ``$x$-type'' and ``$z$-type'' boundaries, as well as the definition of logical operators via boundary-to-boundary strings, follow from the commutation structure of truncated checks near the boundary. At an $x$-type boundary, the $X$-type vertex checks are truncated, so an $X$ string may terminate there without creating a violated vertex stabilizer; meanwhile it still commutes with all $Z$-type plaquette checks because each overlap contributes a single $XZ=-ZX$, and the total number of overlaps between a plaquette and a valid boundary-to-boundary logical string is even, yielding overall commutation. An analogous argument applies to a $Z$ string terminating on a $z$-type boundary. Geometrically, in the usual planar drawing with data qubits on edges, the $x$-type boundaries appear ``rough''. After all, some qubits extend to the boundary (as at $V_1$ and $V_2$ in Figure~\ref{fig:surface_code_stabilizer}), whereas the $z$-type boundaries appear ``smooth'' because the corresponding truncated structure does not exhibit such protruding qubits (as at $V_1^*$ and $V_2^*$).

It is shown in Refs.~\cite{Surfacecode1, Kitaev_anyons_computation1} that, for this class of planar stabilizer constructions, there are essentially only two boundary types compatible with the code’s anyonic commutation and condensation rules. The outer boundary is therefore partitioned into alternating $x$-type and $z$-type segments, and the relevant logical operators are represented by relative cycles and cocycles with endpoints.

\subsection{TQEC on Higher-Dimensional Manifolds with Boundaries}

In Section~\ref{sec: TQEC Codes on Higher Dimensional Manifolds}, we extended the construction of TQEC codes from \(2\)-dimensional manifolds to higher-dimensional closed, compact manifolds, and provided a concrete example given by the three-torus \(T^3\) code. In Theorem~\ref{thm:high_dim_main_theorem}, we proved that a cell complex \(X_M\) embedded in an \(n\)-dimensional closed, compact manifold \(M\) can encode \(b_i(M;\mathbb{Z}_2)\) qubits on the \(i\)-dimensional skeleton \(X^i\), where \(1 \le i \le n-1\).

In this subsection, we present a generalization of Theorem~\ref{thm:high_dim_main_theorem}, which focuses on closed and compact manifolds, that characterizes the qubit-encoding capability of a general \(n\)-dimensional manifold with boundary.

\begin{theorem}
For a cell complex \(X_M\) embedded in an \(n\)-dimensional compact manifold \(M\) with boundary \(\partial M\), suppose there exists \(A \subseteq \partial M\) and \(B = (\partial M \setminus A)\cup \partial A\) such that \(H_i(M, A; \mathbb{Z}_2)\neq \mathbf{0}\) and \(H_{n-i}(M, B; \mathbb{Z}_2)\neq \mathbf{0}\). Then \(X_M\) can encode qubits on the \(i\)-dimensional skeleton \(X^i\).

The number of qubits that can be encoded on the \(i\)-cycles equals the smaller of the ranks of \(H_i(M, A; \mathbb{Z}_2)\) and \(H_{n-i}(M, B; \mathbb{Z}_2)\), or equivalently, the smaller of \(b_i(M, A;\mathbb{Z}_2)\) and \(b_{n-i}(M, B;\mathbb{Z}_2)\). This quantity depends on the chosen decomposition \(\partial M = A \cup B\).
\label{thm:high_dim_main_theorem_boundary}
\end{theorem}

The proof is given in Appendix~\ref{appendix:Proofs}. Theorem~\ref{thm:high_dim_main_theorem_boundary} reduces to Theorem~\ref{thm:high_dim_main_theorem} when \(M\) is a compact manifold without boundary.

Applying Theorem~\ref{thm:high_dim_main_theorem_boundary} to a \(2\)-manifold with boundary recovers the original surface code introduced by Bravyi and Kitaev in Ref.~\cite{Surfacecode1}. It also recovers the result proved by Delfosse, Iyer, and Poulin in Ref.~\cite{Surfacecodetheory1}, while in Ref.~\cite{Surfacecodetheory1}, they explicitly discussed how introducing defects on 2-manifolds with boundaries changes the homology. More importantly, the theorem enables the definition of higher-dimensional TQEC codes that are analogous to the ``surface code'' and are potentially more amenable to experimental realization than constructions requiring more intricate global topology, such as the three-torus \(T^3\) considered in Section~\ref{sec: TQEC Codes on Higher Dimensional Manifolds}, since the presence of boundaries can eliminate the need for long-range entanglement.

\subsection{TQEC on 3-Dimensional Lattice with Mixed Boundaries}

Following the theoretical results developed from the previous subsection. In this subsection, we present concrete and experimentally realizable methods for constructing a 3-dimensional analogy to the "surface code," which we refer to as the "volume code."

As proven in Ref.~\cite{Surfacecode1, Kitaev_anyons_computation1}, there can only be two types of boundary conditions for qubits; we shall call them $x$-boundaries and $ z$-boundaries, sometimes also called smooth and rough boundaries. In three-dimensional lattices, as discussed in Section~\ref{sec: TQEC Codes on Higher Dimensional Manifolds}, logical operators are only well-defined on non-contractible loops and non-compressible tori. Thus, suppose we define our volume code on a 3-dimensional lattice, and we choose the convention that logical $\sigma_X$ operates on lines and logical $\sigma_Z$ operates on surfaces; a logical $\sigma_X$ will be a line with endpoints lying on $ x$-boundaries, and a logical $\sigma_Z$ will be a surface with boundary lines lying on $ z$-boundaries.

\begin{figure}
    \centering
    \includegraphics[width=0.75\linewidth]{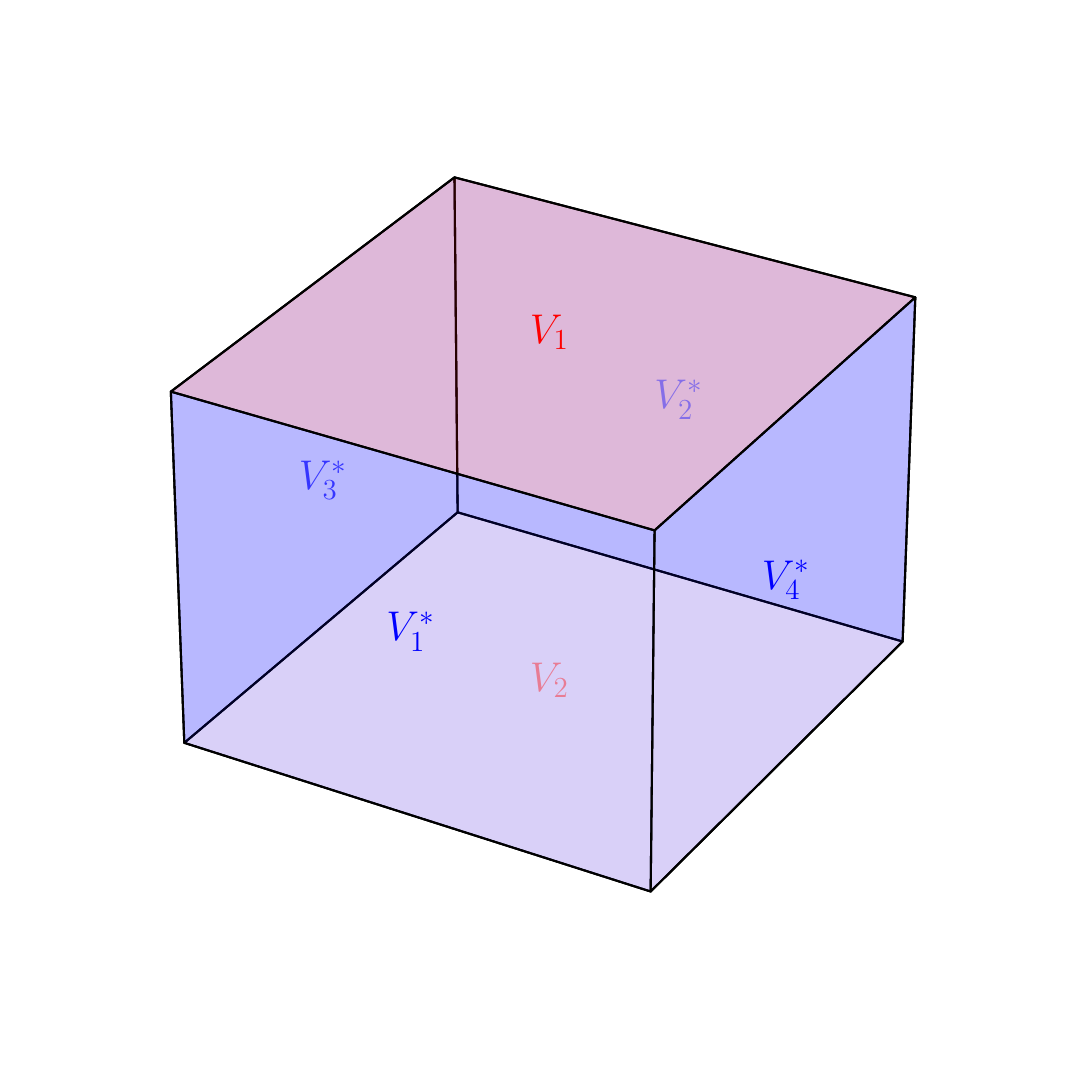}
    \caption{\raggedright A cube with 2 different boundary conditions, red boundaries $V_1, V_2$ are $ x$-boundaries and blue boundaries $V^*_1, V^*_2, V^*_3, V^*_4$ are $ z$-boundaries. Logical $\sigma_X$ are defined as lines with endpoints on $V_1, V_2$ and Logical $\sigma_Z$ are defined as surfaces with boundary loop lying entirely on $V^*_1, V^*_2,V^*_3, V^*_4$.}
    \label{fig:cube_with_different_boundary}
\end{figure}

\begin{figure}
     \centering
     \begin{subfigure}[b]{0.5\textwidth}
         \centering
         \includegraphics[width=0.9\textwidth]{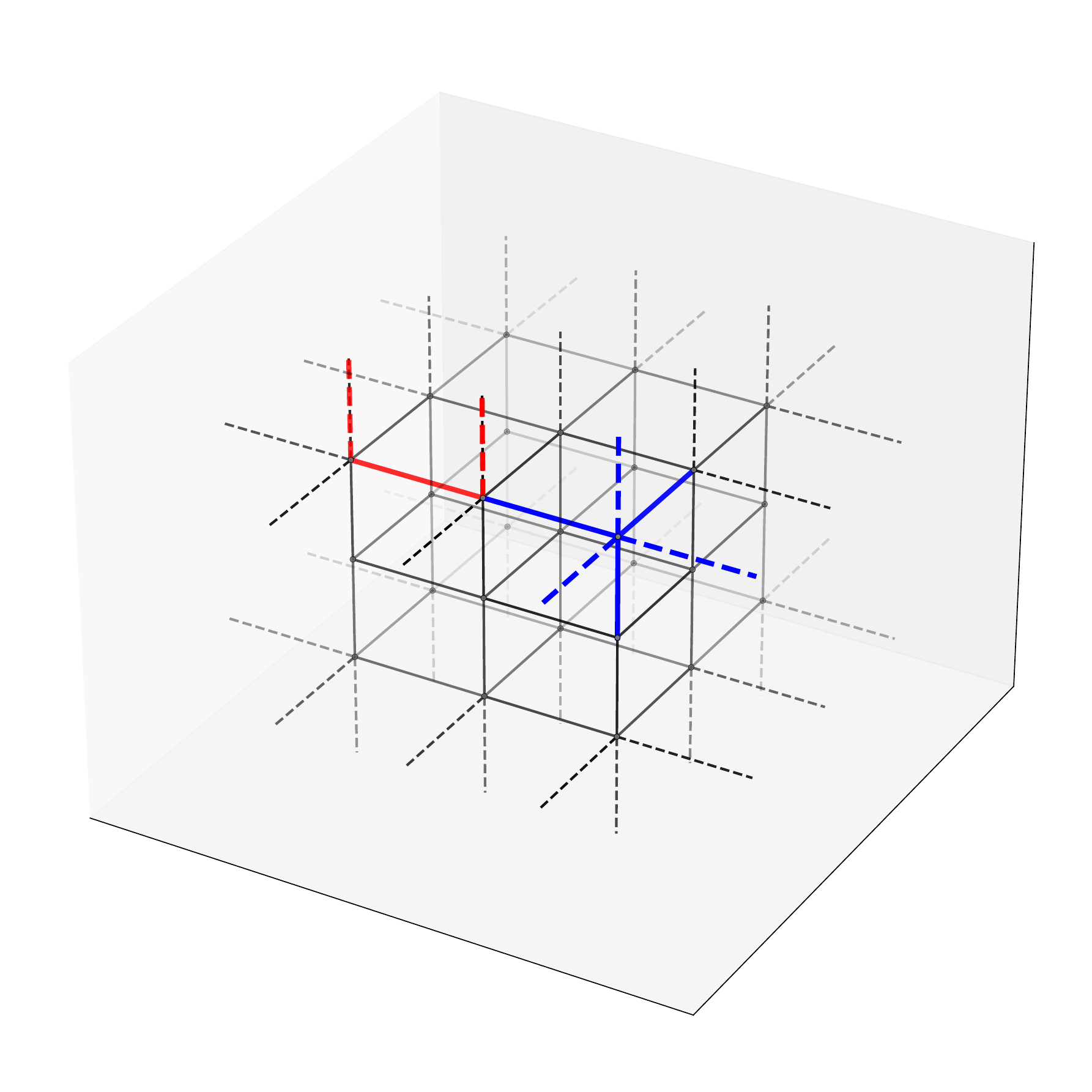}
         \caption{}
         \label{fig:3_torus_volume_code_qubits_rough_boundary}
     \end{subfigure}
     \hfill
     \begin{subfigure}[b]{0.5\textwidth}
         \centering
         \includegraphics[width=0.9\textwidth]{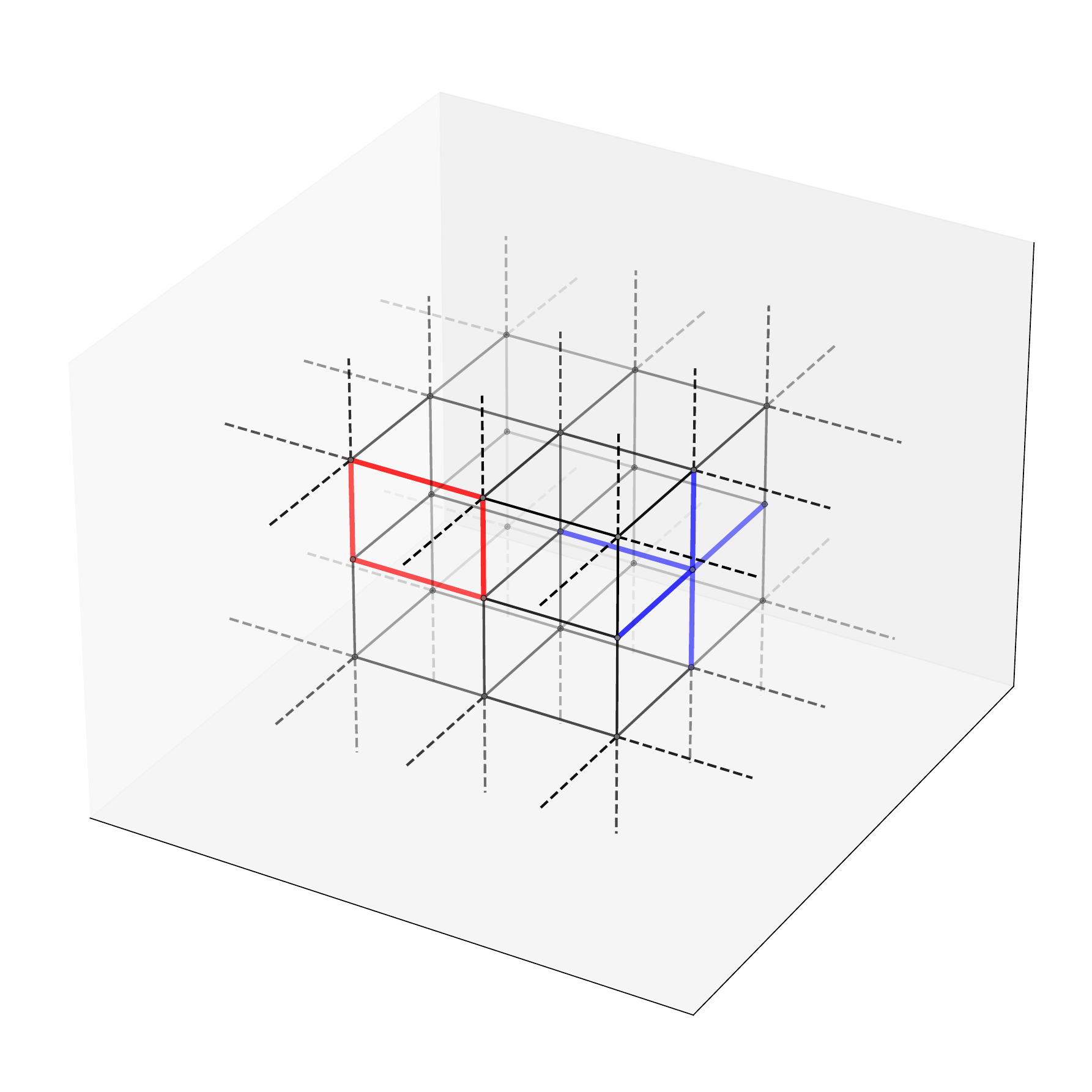}
         \caption{}
         \label{fig:3_torus_volume_code_qubits_smooth_boundary}
     \end{subfigure}
     \caption{\raggedright The qubit allocation for the quantum code defined on a 3-dimensional lattice with boundaries. Using the convention of Figure~\ref{fig:cube_with_different_boundary}, a qubit is placed on every edge. The top and bottom faces are $ x$-boundaries or rough boundaries, while the remaining 4 side faces are $ z$-boundaries or smooth boundaries. Dotted edges represent qubits at the rough boundaries. Figure \textbf{(a)} shows the stabilizers on the rough boundary, where the $z$-type stabilizers (vertex, in blue) retain full weights, while $x$-type stabilizers (face, in red) have less weight. Figure \textbf{(b)} showed the stabilizers on the smooth boundary, where the $x$-type stabilizers (face, in red) retain full weights, while $z$-type stabilizers (vertex, in blue) have less weight.}
     \label{fig:3_torus_volume_code_qubits}
\end{figure}

\begin{figure}
    \centering
    \includegraphics[width=0.6\linewidth]{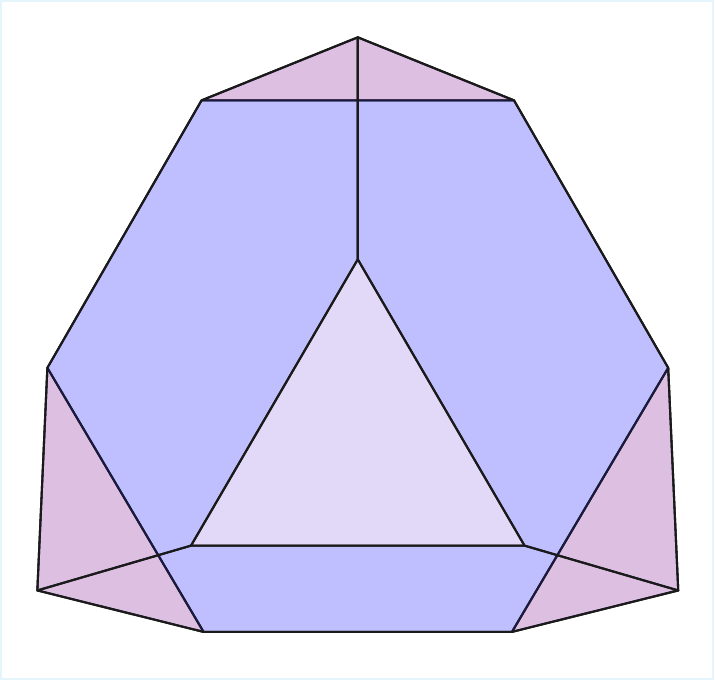}
    \caption{\raggedright A truncated tetrahedron of a total of 8 faces with 2 different boundary conditions, red boundaries (in triangular shape) are $ x$-boundaries and blue boundaries (in hexagonal shape) are $ z$-boundaries. It encodes 3 qubits.}
    \label{fig:truncated_tetrahedron_different_boundary}
\end{figure}

\begin{figure}
    \centering
    \includegraphics[width=\linewidth]{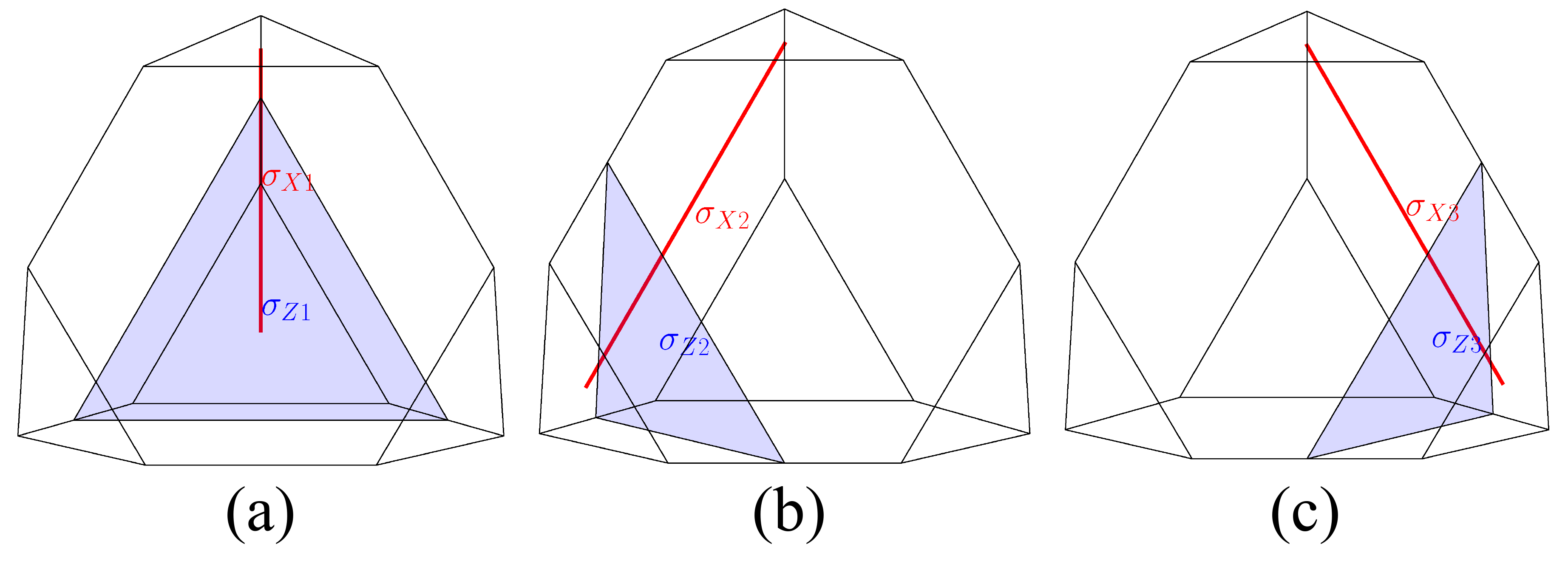}
    \caption{\raggedright The 3 qubits encoded on the truncated tetrahedron with $ x$-boundaries and $ z$-boundaries defined in Figure~\ref{fig:truncated_tetrahedron_different_boundary}, and it encodes 3 qubits. Each sub-figure shows the logical $\sigma_X$ (red line) and logical $\sigma_Z$ (blue surface) of one encoded qubit.}
    \label{fig:tetrahedron_operators}
\end{figure}

Hence, we have an almost trivial implementation of the volume code, which encodes only 1 qubit, as shown in Figure~\ref{fig:cube_with_different_boundary}, where 4 of the middle boundaries are $ z$-boundaries $V^*_1, V^*_2, V^*_3, V^*_4$, and the opposing faces are the $ x$-boundaries $V_1, V_2$, and Logical $\sigma_X$ are defined as lines with endpoints on $V_1, V_2$ and Logical $\sigma_Z$ are defined as surfaces with boundary loop lying entirely on $V^*_1, V^*_2, V^*_3, V^*_4$. Following Theorem~\ref{thm:high_dim_main_theorem_boundary}, in this example, the manifold with boundary $M$ is the big cube, and $A\subset \partial M$ is the top and bottom faces $V_1$ and $V_2$. While $B= \partial M\setminus A\cup\partial A$ are the four side faces $V^*_1, V^*_2, V^*_3, V^*_4$, and $\partial A = \partial B$ are the $8$ edges between the top/bottom faces and the side faces. 

The definition of stabilizers is similar to the 2-dimensional case, as shown in Figure~\ref{fig:3_torus_volume_code_qubits}, where on the $ x$-boundaries (rough boundary) in Figure~\ref{fig:3_torus_volume_code_qubits_rough_boundary}, the $z$-type stabilizers (defined on vertices, in blue) retain full weights, while $x$-type stabilizers (defined on faces, in red) have less weight. While on the $ z$-boundaries (smooth boundary) in Figure~\ref{fig:3_torus_volume_code_qubits_smooth_boundary}, the $x$-type stabilizers (defined on faces, in red) retain full weights, while $z$-type stabilizers (defined on vertices, in blue) have less weight.

In fact, apart from the cubic realization shown in Figure~\ref{fig:cube_with_different_boundary}, one may use any 3-dimensional shape $M$ with boundary $\partial M$, as long as there is a way to find a subset of the boundary $A\subset \partial M$ and $B = \partial M\setminus A \cup\partial A$ such that $H_i(M, A; \mathbb{Z}_2)\ne\mathbf{0}\text{ and }H_{n-i}(M, B; \mathbb{Z}_2)\ne \mathbf{0}$. 

For example, one can use a truncated tetrahedron, which is constructed by cutting the 4 corners of a usual filled tetrahedron, as shown in Figure~\ref{fig:truncated_tetrahedron_different_boundary}. There are a total of 8 faces, with 4 triangular faces and 4 hexagonal faces. We can define the $ x$-boundaries on the triangular faces and the $ z$-boundaries on the hexagonal faces just like in Figure~\ref{fig:truncated_tetrahedron_different_boundary}. This truncated tetrahedron encodes 3 qubits, because we can find 3 distinct sets of relative 1-cycles as logical $\sigma_X$ and relative 2-cycles as logical $\sigma_Z$, and each sub-figure of Figure~\ref{fig:tetrahedron_operators} shows the logical $\sigma_X$ (red line) and logical $\sigma_Z$ (blue surface) of the encoded qubit. 

In more detail, from the example of the ``truncated tetrahedral code'' presented in Figure~\ref{fig:truncated_tetrahedron_different_boundary} and~\ref{fig:tetrahedron_operators}, the manifolds with boundary $M$ are trivially the truncated tetrahedron, and $\partial M$ are the 8 boundary faces (4 triangular faces and 4 hexagonal faces) of a truncated tetrahedron. We defined $A\subset \partial M$ to be the 4 triangular faces (red faces in Figure~\ref{fig:truncated_tetrahedron_different_boundary}), and then, $B = \partial M\setminus A$ are the 4 hexagonal faces (blue faces in Figure~\ref{fig:truncated_tetrahedron_different_boundary}). 

Then, $A\subset \partial M$ are defined as the $ x$-boundaries, and $B = \partial M\setminus A\cup\partial A$ are defined as $ z$-boundaries; meaning that relative 1-cycles in $M$ are defined as closed cycles (with no endpoints) or lines with both endpoints lying on the triangular faces, and relative 2-cycles in $M$ are defined as closed surfaces (with no 1-dimensional boundary) or surface with its boundaries lying on the hexagonal faces. Thus, Figure~\ref{fig:tetrahedron_operators} lists out a possible combination of 3 distinct classes of 1-cycles and 2-cycles, which can be used as logical $\sigma_X$ and $\sigma_Z$ for 3 encoded logical qubits. 

This construction is general, and it has great implications in experiments, because the ``volume codes'' proposed above may not require long-range entanglement.

\section{Construction of TQEC codes for Qudit on 2-dimensional Cell Complex}
\label{sec:TQEC_on_2_complex}

In Section~\ref{sec:TQEC_on_2_manifolds}, we established a homological--intersection-theoretic formulation of topological quantum error-correcting (TQEC) codes on closed \(2\)-manifolds. We subsequently extended this viewpoint to closed manifolds of higher dimension in Section~\ref{sec: TQEC Codes on Higher Dimensional Manifolds}, and then to arbitrary-dimensional manifolds with boundary in Section~\ref{sec:TQEC_boundary}. While all the previous discussions still focus on smooth manifolds with and without boundaries. We want to highlight that the essential ingredients of the construction are not inherently smooth or even manifold-theoretic, but rather cellular and homotopy-theoretic in nature.

\subsection{Presentation Complex and Homology}

We therefore now move beyond the manifold setting and develop a systematic framework for constructing TQEC codes on general \(2\)-dimensional cell complexes. This shift is motivated both conceptually and practically. Conceptually, it isolates the topological content of the code, without requiring an embedding into a surface. Practically, a broad class of product-type quantum LDPC constructions admits a ``code-to-manifold'' or, more generally, a ``code-to-cell-complex'' correspondence~\cite{qLDPC1, qLDPC_review1}. Representative families include balanced product codes~\cite{qLDPC_review1, Balanced_Product_Quantum_Codes1}, hypergraph product codes~\cite{qLDPC_review1, Hypergraph_product_codes1}, quantum Tanner codes~\cite{qLDPC_review1, Quantum_Tanner_codes1}, and fibre bundle codes~\cite{qLDPC_review1, fibre_bundle1, fibre_bundle2}, among others. In this section, we present an alternative and topologically grounded route: we show how \(2\)-dimensional cell complexes can be used to realize TQEC codes encoding a general \(d\)-dimensional qudit, even in settings where no underlying manifold model is available, meaning the TQEC code defined on surfaces that are not locally flat~\cite{Algebraic_Topology1, Algebraic_Topology2}.

The construction we propose is guided by a well-known result in algebraic topology, proved in standard references such as Ref.~\cite{Algebraic_Topology1, Algebraic_Topology2}, which ensures that arbitrary group-theoretic data can be realized as the fundamental group of a \(2\)-complex.

\begin{theorem}
\label{thm:cell_complex_with_pi1_group}
For every group \(G\), there exists a connected \(2\)-dimensional CW complex \(X_G\) such that
\(
\pi_1(X_G)\cong G.
\)
\end{theorem}

The statement is well-developed, and its proof proceeds by constructing a cell complex from a presentation of \(G\). For completeness and to fix notation, we recall the standard ``presentation complex'' construction (see, e.g., Ref.~\cite{Algebraic_Topology1}).

\begin{proof}[Construction (Presentation complex)]
Let \(G=\langle S \mid R\rangle\) be a presentation of \(G\), where \(S\) is a generating set and \(R\subseteq F(S)\) is a set of relators in the free group \(F(S)\). We construct a connected \(2\)-dimensional cell complex \(X_G\) with a chosen basepoint \(v\) as follows.

\emph{(i) The \(0\)-skeleton.}  
Take a single \(0\)-cell \(v\).

\emph{(ii) The \(1\)-skeleton.}  
For each \(s\in S\), attach a \(1\)-cell \(e_s\) by identifying both endpoints with \(v\). The resulting \(1\)-skeleton is a wedge of circles,
\[
X_G^{(1)} \;=\; \bigvee_{s\in S} S^1,
\]
and admits the canonical identification \(\pi_1\!\left(X_G^{(1)},v\right)\cong F(S)\).

\emph{(iii) Attaching \(2\)-cells along relators.}  
For each relator \(r\in R\), choose a cyclically reduced representative
\(
r=s_1^{\varepsilon_1}\cdots s_k^{\varepsilon_k}
\)
with \(s_i\in S\) and \(\varepsilon_i\in\{\pm1\}\). Define an attaching map \(f_r:S^1\to X_G^{(1)}\) by subdividing \(S^1\) into \(k\) arcs and mapping the \(i\)-th arc homeomorphically onto \(e_{s_i}\), with orientation determined by \(\varepsilon_i\). Attach a \(2\)-cell \(D_r^2\) along \(f_r\). Let \(X_G\) denote the resulting CW complex.

By van Kampen's theorem and the standard cellular description of \(\pi_1\), the inclusion \(X_G^{(1)}\hookrightarrow X_G\) induces the quotient of \(F(S)\) by the normal closure of \(R\), hence
\[
\pi_1(X_G,v)\;\cong\;\langle S \mid R\rangle \;=\; G,
\]
as claimed.
\end{proof}

\begin{proof}[\textbf{Example}]
We first recall two concrete instances of the presentation-complex construction for cyclic groups. For \(G=\mathbb{Z}_2\), the standard presentation complex is homeomorphic to the closed \(2\)-manifold \(\mathbb{R}P^{2}\).

Following the geometric model described in~\cite{Algebraic_Topology1}, we next sketch a presentation complex for \(G=\mathbb{Z}_3\). Begin with an oriented \(1\)-cell complex \(S^{1}\) equipped with a distinguished basepoint \(v\). Choose a regular neighborhood \(N\) of \(S^{1}\) in the resulting \(2\)-complex whose local cross-section is a three-pronged ``asterisk'', \inlinegraphics{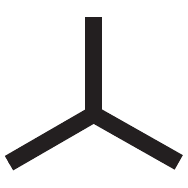}, and impose the attaching identification so that, upon traversing once around the core circle \(S^1\) and returning to \(v\), the neighborhood is glued back with a one-third twist. This produces a \(2\)-dimensional cell complex \(X_{\mathbb{Z}_3}\) realizing \(\pi_1(X_{\mathbb{Z}_3})\cong \mathbb{Z}_3\), as illustrated in Figure~\ref{fig:Z_3_Cell_Complex}.

The same geometric pattern extends to \(\mathbb{Z}_4\): one replaces the \inlinegraphics{3_point_asterisk.pdf}-shaped cross-section by a four-pronged cross-section in ``$\times$'' shape and performs a quarter-twist in the gluing. More generally, for \(\mathbb{Z}_n\) one may use an \(n\)-pointed asterisk neighborhood together with a \(1/n\) twist to obtain a presentation complex \(X_{\mathbb{Z}_n}\).
\end{proof}

It is important to emphasize that the above construction produces an oriented \(2\)-dimensional cell complex \(X_{\mathbb{Z}_n}\) for \(n\ge 3\); see, e.g.,~\cite{Algebraic_Topology1, Presentation_Complex1, Presentation_Complex2, Presentation_Complex3}. For completeness, we review the mathematical definition of orientation in Appendix~\ref{subsec:manifolds_and_orientation}. Once an orientation is fixed, the sign of a transverse intersection point between two oriented loops can be defined locally and consistently. This well-defined intersection sign will play a central role in the proof of Theorem~\ref{thm:TQEC_for_Qudit}.

At the same time, presentation complexes should not be conflated with surfaces. In general, a \(2\)-dimensional cell complex \(X_G\) need not be a \(2\)-manifold and may not admit an embedding into any \(2\)-manifold~\cite{Algebraic_Topology1, Algebraic_Topology2}. Equivalently, there may exist points in \(X_G\) whose neighborhoods are not locally homeomorphic to \(\mathbb{R}^2\) (for instance, vertices whose links are not circles). With this distinction in mind, we next review the TQEC construction for qudit stabilizer codes before specializing it to presentation complexes and related \(2\)-dimensional cell complexes.

\subsection{Reviewing TQEC on Qudits}

\begin{figure}
    \centering \includegraphics[width=0.6\linewidth]{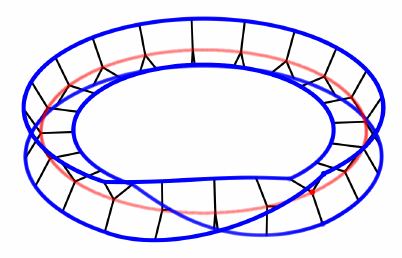}
\caption{\raggedright The 1-skeleton of the Presentation Complex of $\mathbb{Z}_3$, the complete presentation complex is obtained by attaching a disk to the skeleton.}
    \label{fig:Z_3_Cell_Complex}
\end{figure}

\begin{figure}
    \centering
    \includegraphics[width=0.75\linewidth]{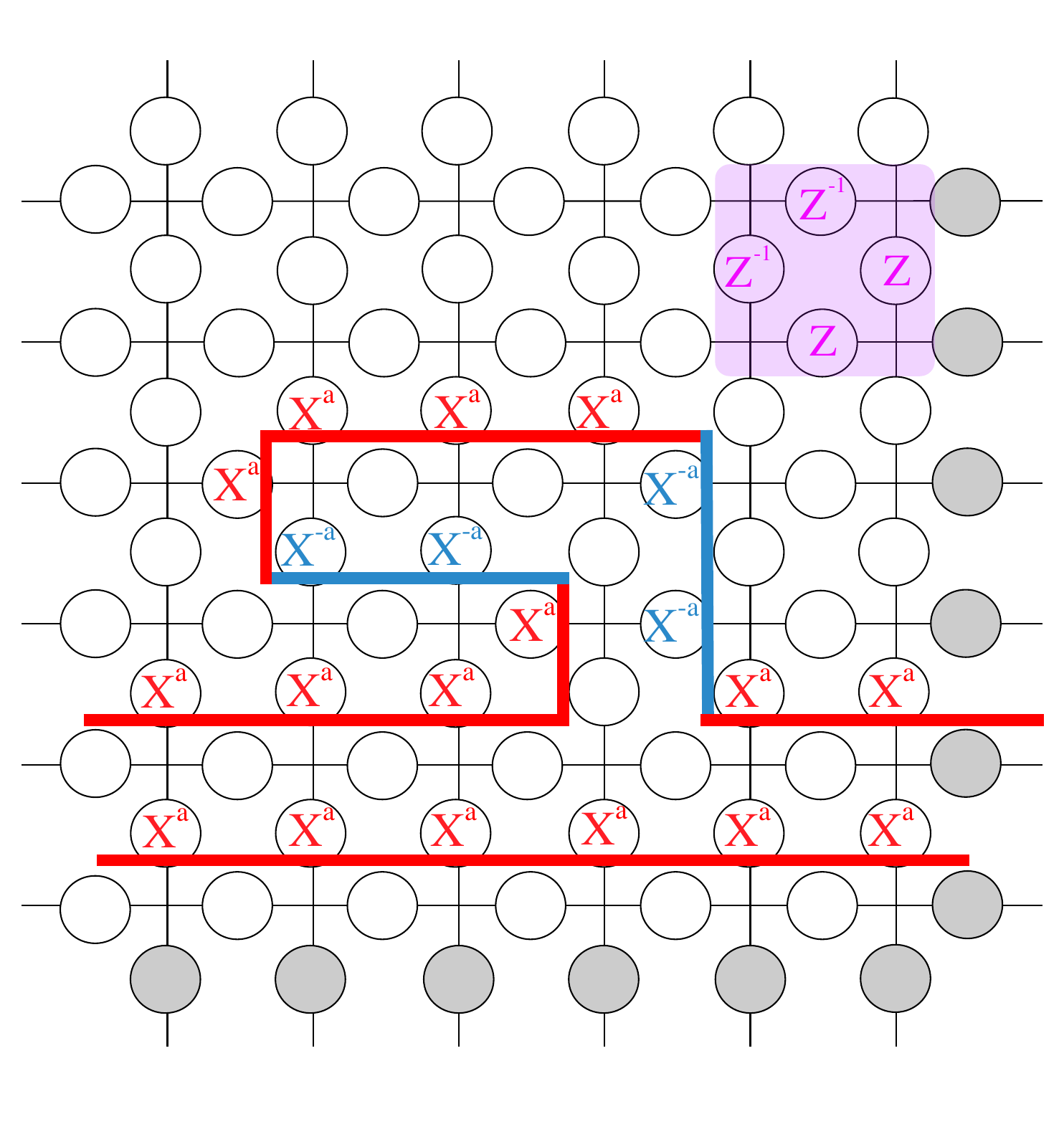}
    \caption{\raggedright The definition of logical $X^a$ operators of a Qudit TQEC code, the lower one is an example of the shortest logical $X^a$ operators, while the upper one is detoured. Both of them commute with the $Z$-type stabilizers, an example is shown in the top right corner.}
    \label{fig:Qudit_code_operators}
\end{figure}

\begin{figure}
    \centering
    \includegraphics[width=0.9\linewidth]{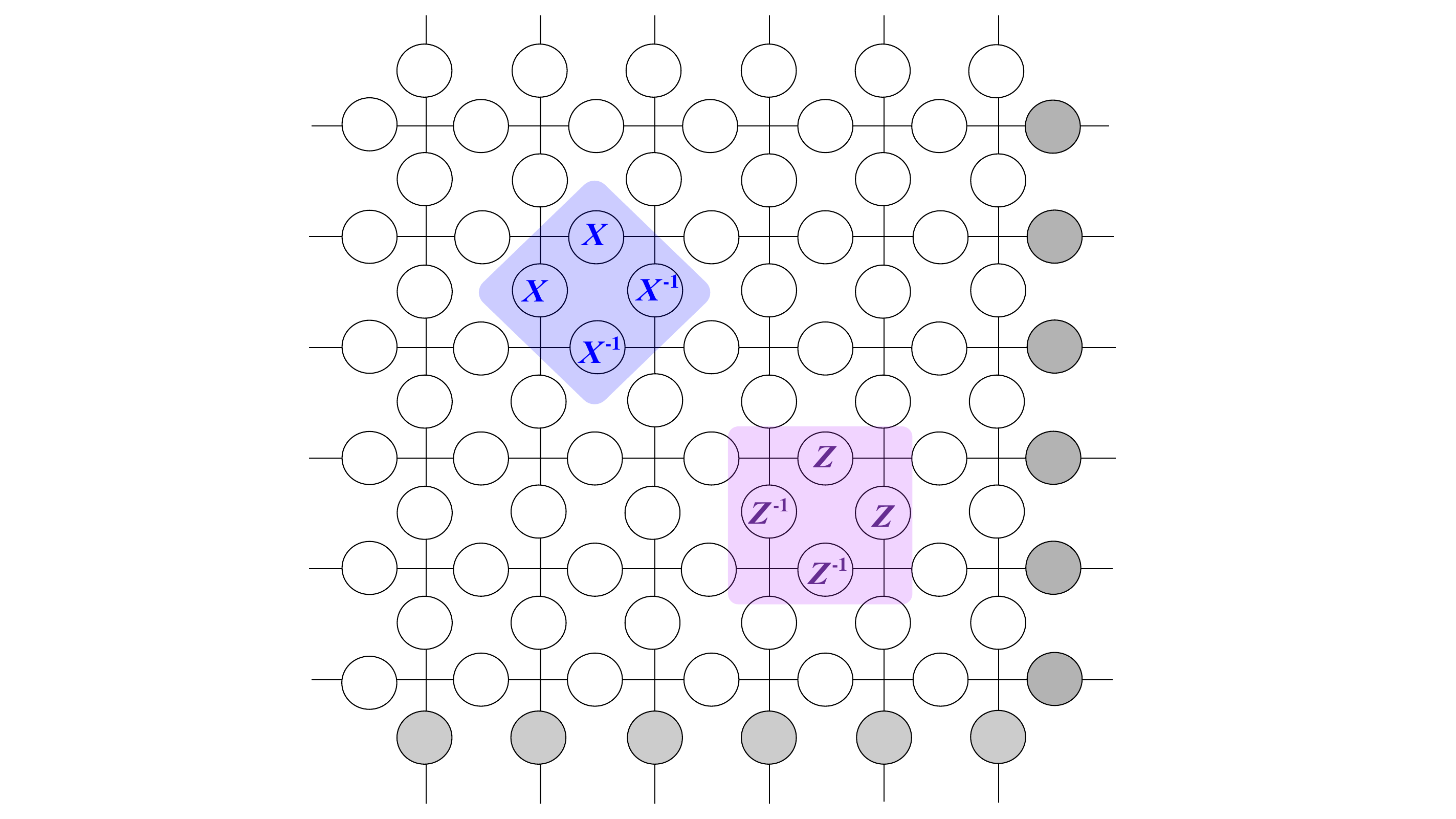}
    \caption{\raggedright The definition of two types of stabilizers of a Qudit TQEC code, which satisfy the requirement that all stabilizers commute, see, e.g., Ref.~\cite{Z_d_code1}}
    \label{fig:Qudit_code_stabilizer}
\end{figure}

The error correction for qudits living in a $d$-dimensional Hilbert space $\mathcal{H}_d$ spanned by $\{\ket{0}, \ket{1}, \cdots, \ket{d-1}\}$ can be generalized from encoding qubits, as discussed in many previous literature~\cite{Z_d_code1, Kitaev_anyons_computation1}, a generalized Pauli group is defined, $g\in[0, d-1]$,

\begin{equation}
    X=\sum_{j\in\mathbb{Z}_d}\ket{j+1}\bra{j} \qquad Z=\sum_{j\in\mathbb{Z}_d}\omega^j\ket{j}\bra{j}
    \label{eq:XZ_operator_qudit}
\end{equation}

\begin{equation}
        X\ket{g}=\ket{g+1 \mod d} \qquad Z\ket{g}=\omega^g\ket{g}
    \end{equation}
    where $\omega = e^{i2\pi/d}$. Thus, we have
    \begin{equation}
        X^a\ket{g}=\ket{g+a \mod d} \qquad Z^b\ket{g}=\omega^{bg}\ket{g}
    \end{equation}
    \begin{equation}
    X^aZ^b\ket{g}=\omega^{ab}Z^bX^a\ket{g}
    \end{equation}
The qudit toric code and its homological generalizations have been studied previously (see, e.g., Ref.~\cite{Z_d_code1}).
A key structural difference from the qubit case is that, for \(d>2\), the homological construction requires an oriented manifold or, more generally, an oriented cell complex: as shown in Ref.~\cite{HomologicalQEC1}, non-orientable manifolds do not support such qudit encodings when \(d>2\).
In our setting, this obstruction is avoided because the presentation complexes \(X_{\mathbb{Z}_n}\) are orientable for \(n>2\) as proven in Ref.~\cite{Algebraic_Topology1, Algebraic_Topology2}), which allows one to consistently distinguish ``forward'' versus ``backward'' traversal along oriented \(1\)-cells. Therefore, one can find orientable loops on the presentation complex and define the forward and backward directions. The logical operators are defined to operate on single qudits forming a loop; the single-qudit operations are simply $X$ in one direction and its inverse $X^{-1}$ in the reverse direction, just as shown in Figure~\ref{fig:Qudit_code_operators}. 

For a $d$-dimensional qudit, the two types of logical operators are $X^a$ and $Z^b$ for $a, b\in[0, d-1]$~\cite{Z_d_code1, Kitaev_anyons_computation1}. Figure~\ref{fig:Qudit_code_operators} have shown examples of logical $X^a$ for a qudit for $a\in[0, d-1]$. The shortest logical operator works the same as that of the usual, i.e., a closed loop of $X^a$ acting on single qudits. If there are detours, parts of the logical operators are $X^{-a}$ instead of $X^a$ because of the orientation conditions, and also to ensure that the detoured logical operators commute with all stabilizers. Logical $Z^a$ is defined similarly. 

\begin{equation}
    H = -J\sum_v A_v - J\sum_p B_p
    \label{eq:qudit_code_hamiltonian}
\end{equation}

\begin{equation}
    \Pi_{v\in V} A_v=\Pi_{p\in F}B_p=Id
\end{equation}

Physically and mathematically, any of the logical operators must commute with all stabilizers. Therefore, provided that the direction of the loops is well-defined, the stabilizers are also well-defined. Then, there is a well-defined set of stabilizers that agrees with the orientation convention of the surface as shown in Figure~\ref{fig:Qudit_code_stabilizer}. See~\cite{Z_d_code1} for a detailed analysis of the commutation relationship between the stabilizers. 

As proven by Theorem III.2 in~\cite{HomologicalQEC1}, a connected 2-complex $\Sigma$ can encode $k$ qudit with dimension $d$ if $H_1(\Sigma)\simeq H^1(\Sigma)\simeq Z_d^k$. In this section, we demonstrate how to build codes with the above presentation complex construction, and give a proof of why it works.

A key difference between the presentation complex and a cell complex tessellation of a manifold is that, for a presentation complex constructed as discussed above, we will always adopt a square tessellation, while for a cell complex tessellation of a manifold, the tessellation can be of any shape.

\subsection{Encoding Qudits on Presentation Complexes}

\begin{figure}
    \centering
    \includegraphics[width=0.85\linewidth]{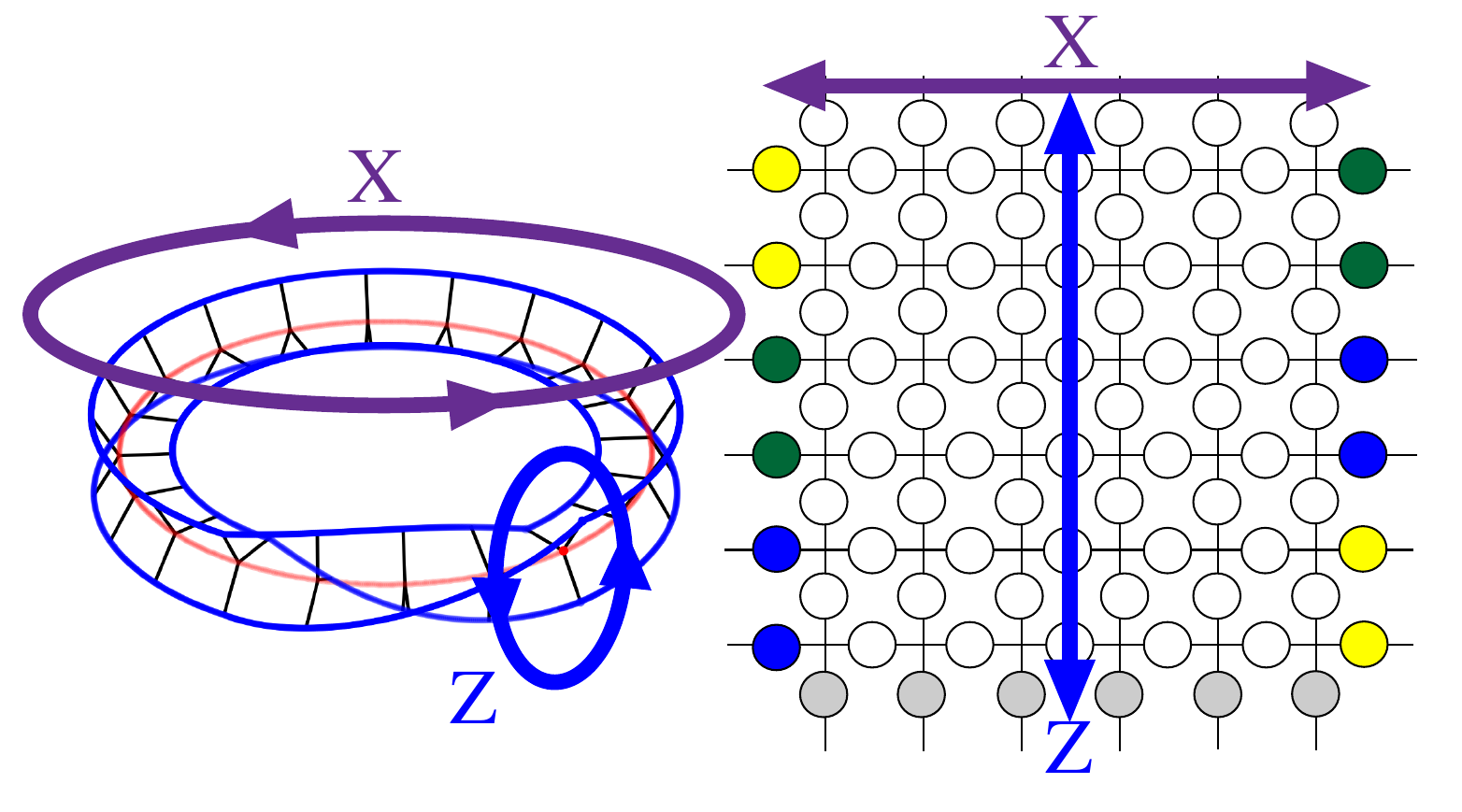}
    \caption{\raggedright An illustration of the directions of loops where the logical $X$ and logical $Z$ operate on a presentation complex built from a group $\mathbb{Z}_3$. The $1/3$ twist nature of the presentation complex $G_{\mathbb{Z}_3}$ is realized by identifying the qutrits labeled with different colors (yellow, green, and blue) at the boundary. A qutrit is placed on each edge of the lattice.}
    \label{fig:Z_3_operators}
\end{figure}

\begin{figure}
    \centering
    \includegraphics[width=\linewidth]{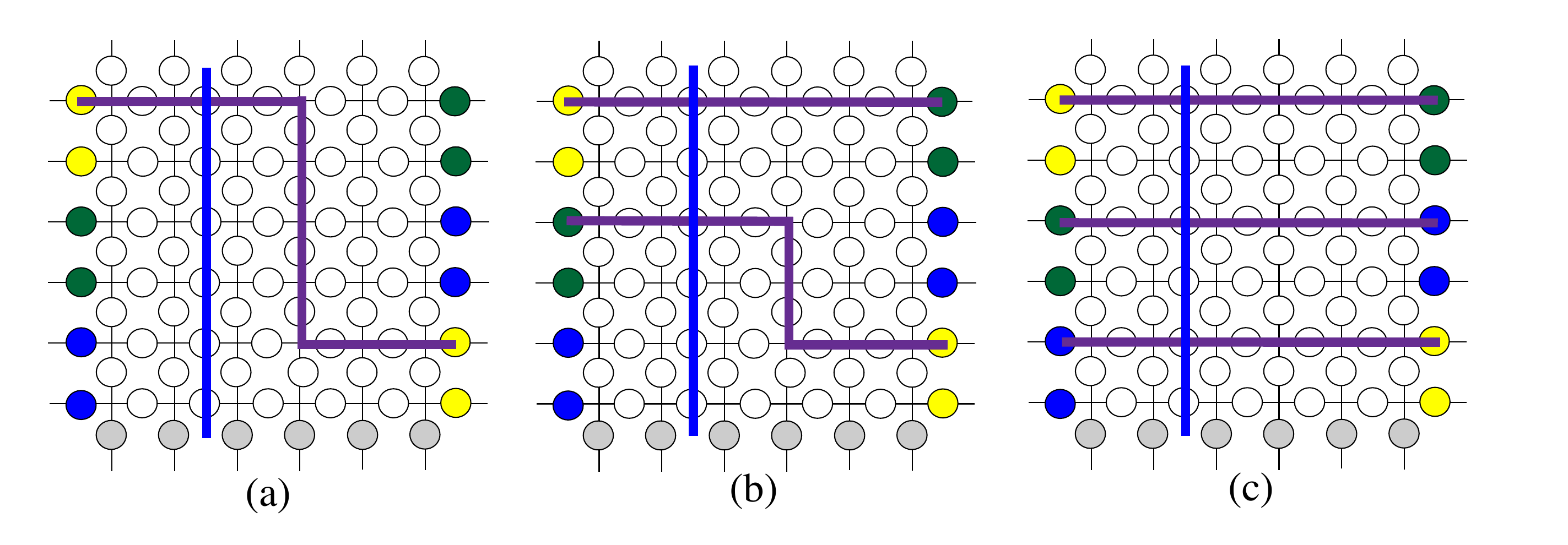}
    \caption{\raggedright An illustration of the three different class of logical $X$ operators on a lattice representing the presentation complex $X_{\mathbb{Z}_3}$ Logical $X$ are in purple while logical $X$ are in blue. Figure \textbf{(a)} shows the example logical $X^1$ operator that loops around for 1 time, thus it intersects with the logical $Z$ at one site, while Figure \textbf{(b)} and \textbf{(c)} represent logical $X^2$ and logical $X^3$ respectively.}
    \label{fig:Qutrits_operators}
\end{figure}

We have the general theorem regarding encoding a $d$-dimensional qudit on a presentation complex as follows.

\begin{theorem}
    A $[[n, k, \min(d_X, d_Z)]]$ TQEC code for $k$ qudit with dimension $m$ can be obtained by building a presentation complex $X_G=(V, E, F)$ from a group $G=\mathbb{Z}_m^k$ with $n = |E|$, $H_1(X_G)=H^1(X_G)=G$, and $d_X$ equals the shortest non-trivial closed cycle, and $d_Z$ equals the shortest non-trivial closed co-cycle on $X_G$.
    \label{thm:TQEC_for_Qudit}
\end{theorem}

An example of encoding a qutrit on a presentation complex constructed from $G=\mathbb{Z}_3$ is shown in Figure~\ref{fig:Qutrits_operators}, where the three different class of logical $X$ operators on a lattice representing the presentation complex $G_{\mathbb{Z}_3}$ Logical $X$ are in purple while logical $X$ are in blue. The $1/3$ twist nature of the presentation complex $G_{\mathbb{Z}_3}$ is realized by identifying the qutrits labeled with different colors at the boundary (yellow, green, and blue). Sub-figure \textbf{(a)} shows the example logical $X^1$ operator that loops around for 1 time, thus it intersect with the logical $Z$ at one site, while Figure \textbf{(b)} and \textbf{(c)} represents logical logical $X^2$ and logical $X^3$ respectively, which intersect with the logical $Z$ at 2 and 3 sites respectively.

We remark that, given the presentation complex of $X_{\mathbb{Z}_f}$, it can encode a qudit of dimension $d$ if $d$ is an integer factor of $f$. The reason being that, for an integer \(f > 0\) and let \(d > 0\) be a divisor of \(f\). Then \(\mathbb{Z}_f\) has a unique subgroup \(H_d\) of order \(d\); this subgroup is cyclic of order \(d\), hence a torsion group, and there is an isomorphism of groups \(H_d \cong \mathbb{Z}_d\). For example, if we construct a presentation complex $X_{\mathbb{Z}_{8}}$, then it should be able to encode 2-level qubits, 4-level qudits, and 8-level qudits.

It is interesting to see that this construction can be reduced to some qLDPC codes presented in previous literature. For example, our construction covers the simplest fibre bundle code constructed by the product of two classical codes with a twist at one of the ends~\cite {fibre_bundle1, fibre_bundle2, qLDPC_review1}, the reduction to previously known qLDPC codes provides another way to understand certain qLDPC codes, and many more qLDPC codes are not directly related to this construction, for example, the non-local qLDPC codes~\cite{qLDPC_review1}.

\section{Examples of TQEC Codes on 2-Manifolds and Simulation Results}
\label{sec: Examples of TQEC Codes on 2-Manifolds and Simulation Results}

After discussing the higher-dimensional TQEC codes, we would like to shift our focus back to codes on 2-manifolds, which are experimentally easier to realize. In prior studies, the toric code has been extensively discussed as a foundational example of TQEC code~\cite{Toriccode1, TQEC_review1}. While there is no inherent reason to restrict discussions solely to the torus, TQEC codes defined on other topological manifolds have received limited attention in the literature. For instance, the possibility of utilizing the real projective plane $\mathbb{R}P^2$ for TQEC was explored in~\cite{RP2code1}, but no proposals have been made regarding the use of the Klein bottle for TQEC purposes. Moreover, comprehensive studies on the error-correcting capabilities of various manifolds remain absent. 

This section aims to fill in this gap by presenting simulation results that examine the feasibility of employing the Klein bottle for TQEC purposes. Additionally, we compare its performance against the classical toric code to evaluate its potential advantages and limitations. Through this analysis, we seek to expand the understanding of TQEC codes on non-trivial topological manifolds and their practical implications for quantum error correction.

\subsection{Implementation of the TQEC Codes}

\begin{figure}
    \centering
    \includegraphics[width=0.8\linewidth]{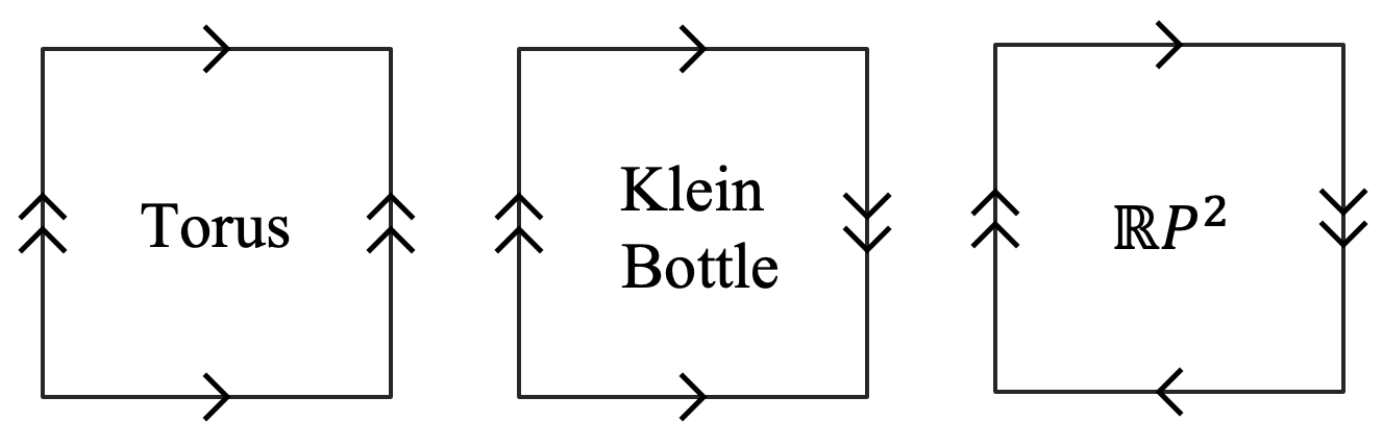}
    \caption{\raggedright The fundamental polygons of torus, Klein bottle and $\mathbb{R}P^2$.}
    \label{fig:fundamental_polygons}
\end{figure}

\begin{figure}
    \centering
    \includegraphics[width=0.8\linewidth]{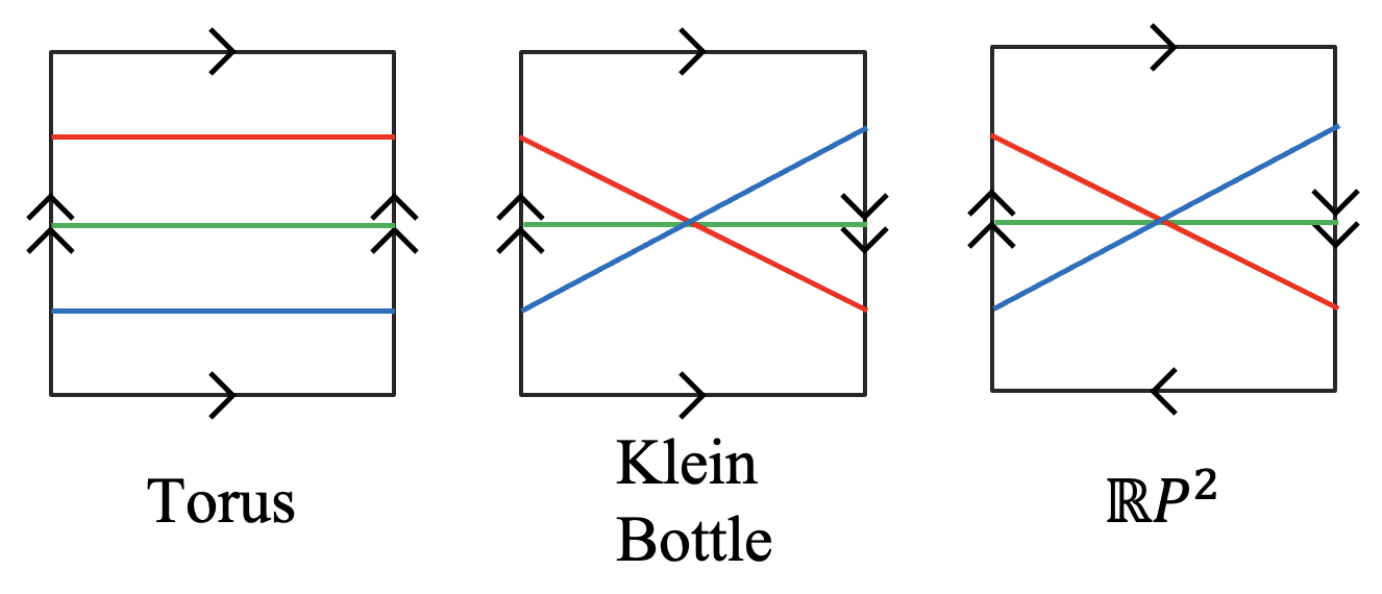}
    \caption{\raggedright Examples of 3 different non-contractible loops belonging to the group $H_1$ of torus, Klein bottle, and $\mathbb{R}P^2$, in which logical $\sigma_X$ may act on.}
    \label{fig:logical_Z_operators}
\end{figure}

The toric code is defined on a two-dimensional square lattice with periodic boundary conditions, effectively forming a torus. The arrangement of qubits is depicted in Figure~\ref{fig:Toric_code_stabilizers}~\cite{Toriccode1, TQEC_review1}. As discussed previously, logical $\sigma_X$ operations are implemented by applying physical $\sigma_X$ operators along a loop that wraps around the torus. Mathematically, this corresponds to acting with physical $\sigma_X$ operators on $[\gamma] \in C_1$, where $[\gamma] \neq 0$ and $\partial[\gamma] = 0$. This construction highlights the topological nature of the code, as the logical operators are associated with non-trivial homological cycles on the underlying manifold.

The fundamental polygons of the torus, Klein bottle, and $\mathbb{R}P^2$ are illustrated in Figure~\ref {fig:fundamental_polygons}, with differences arising from the orientation of their periodic boundary conditions. Consequently, the non-contractible cycles associated with logical $\sigma_X$ operations vary across these manifolds. Figure~\ref{fig:logical_Z_operators} provides examples of three distinct cycles, $\gamma_i$ for $i = 1, 2, 3$, on each 2-manifold, distinguished by different colors. 

The simulation of TQEC codes on the Klein bottle is very similar to that of the toric code; the primary modification involves altering the periodic boundary conditions of the lattice on one pair of edges. The detailed methodology for conducting these simulations is described in the following subsection. This approach highlights the adaptability of TQEC codes to different topological manifolds while maintaining computational feasibility.

Recalled from the examples above, $H_1(K; \mathbb{Z}_2) = \mathbb{Z}_2\oplus \mathbb{Z}_2$, and $H_1(\mathbb{R}P^2; \mathbb{Z}_2) = \mathbb{Z}_2$. Consequently, only a single qubit can be encoded when utilizing the topology of $\mathbb{R}P^2$. In contrast, the Klein bottle topology allows for the encoding of two qubits, similar to the toric code. Therefore, in subsequent sections, we will present simulations for the quantum error correction code based on the Klein bottle topology. This analysis evaluates its performance and potential advantages compared to the original toric code.

\subsection{Methodology}
\label{subsec:methodology}

We utilize Python to simulate the physical and logical error rates for qubits on the torus $T^2$ and the Klein bottle $K$. The simulations conducted in this section serve as a proof-of-principle demonstration. Consequently, we restrict our analysis to lattice-structured cell complexes, the same as the original toric code, embedded on the torus $T^2$ and the Klein bottle $K$, as illustrated in Figure~\ref{fig:Toric_code_stabilizers}. The exploration of alternative cell complexes is left for future research.

As discussed earlier, the code is defined on a cell complex represented by a graph $G = (V, E, F)$, where $G$ is generated using the NetworkX package with varied periodic boundary conditions. For simplicity, we employ a square lattice structure. Noise on qubits is modeled using random noise sampled from a uniform distribution with depolarizing noise probability thresholds. The decoding process is performed using the PyMatching package~\cite{pymatching1}, which implements the Minimum Weight Perfect Matching (MWPM) decoder~\cite{pymatching1, pymatching2}. Additionally, most analytical computations are carried out using the NumPy package, while figures are generated with Matplotlib. The fundamental framework of simulation on the toric code is adapted from the sample code provided by PyMatching~\cite{pymatching1, pymatching2}, and the implementation for simulating the Klein bottle code is original.

\subsection{Simulation with Random Qubit Errors}

\begin{figure}
     \centering
     \begin{subfigure}[b]{0.5\textwidth}
         \centering
         \includegraphics[width=\textwidth]{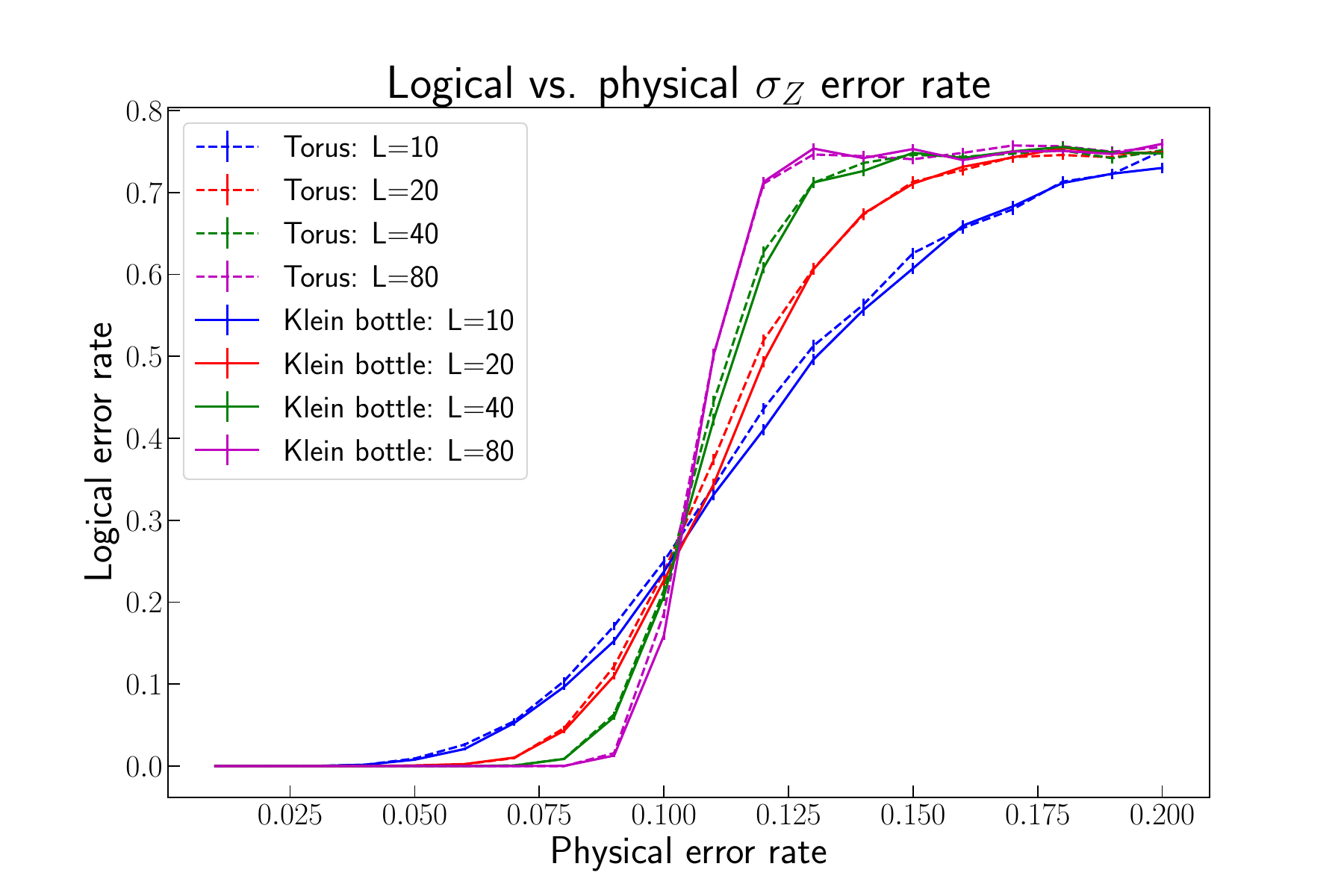}
         \caption{}
         \label{fig:simulation_QBER_torus_k_001_02}
     \end{subfigure}
     \hfill
     \begin{subfigure}[b]{0.5\textwidth}
         \centering
         \includegraphics[width=\textwidth]{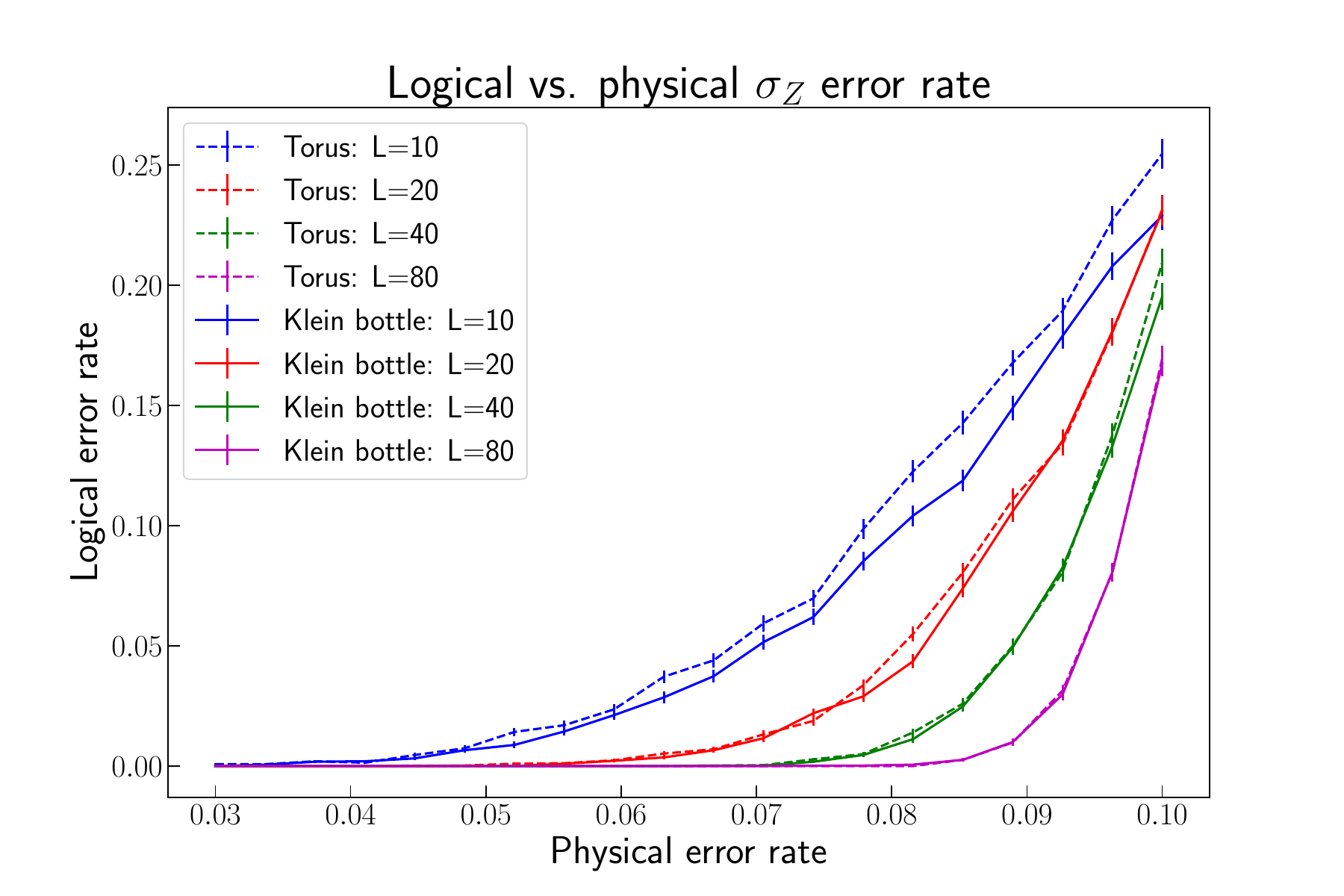}
         \caption{}
         \label{fig:simulation_QBER_torus_k_003_01}
     \end{subfigure}
     \hfill
     \begin{subfigure}[b]{0.5\textwidth}
         \centering
         \includegraphics[width=\textwidth]{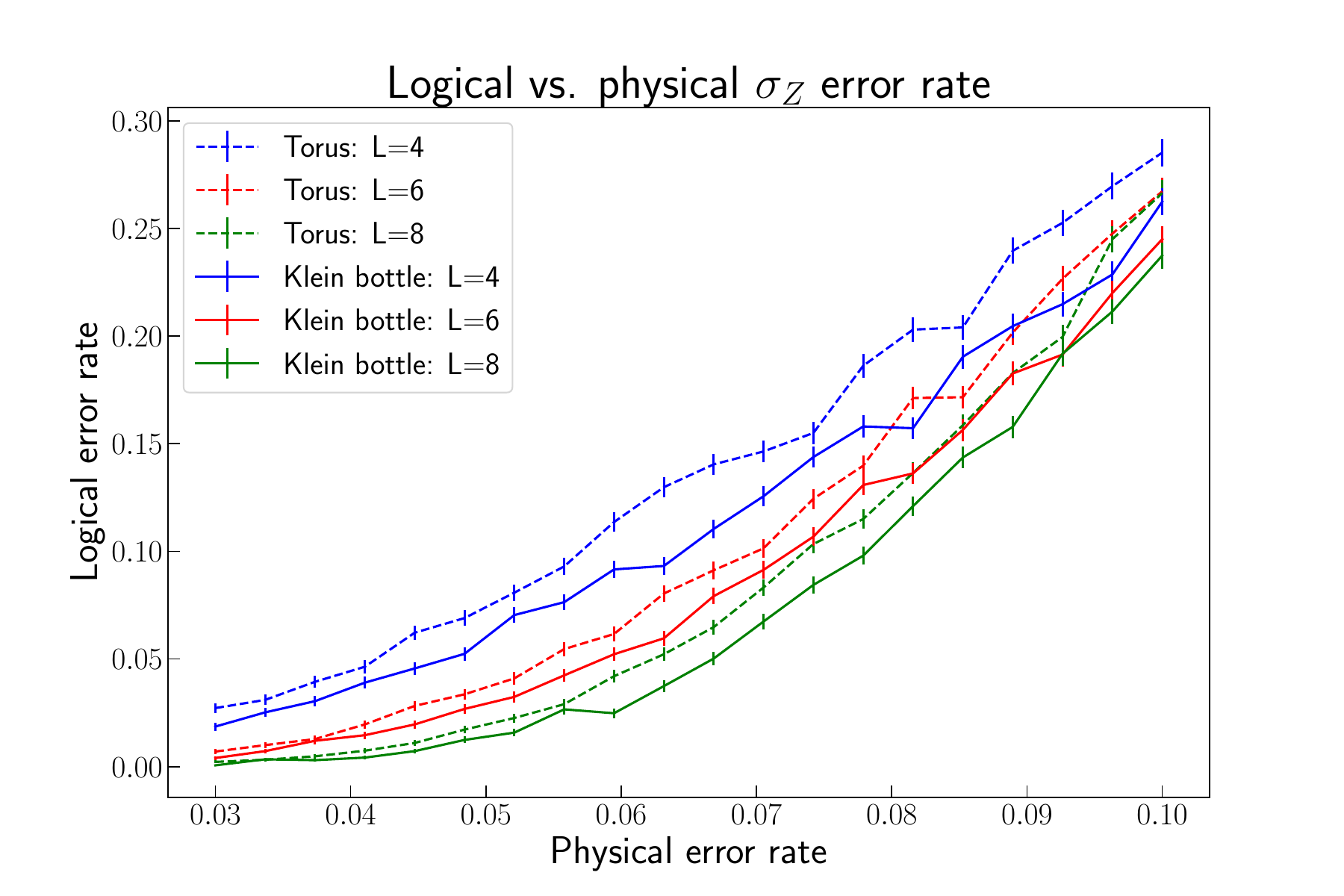}
         \caption{}
         \label{fig:simulation_QBER_torus_k_003_01_small}
     \end{subfigure}
     \caption{\raggedright Logical versus physical $\sigma_Z$ error rates for TQEC code embedded on torus $T^2$ and Klein bottle $K$ for different lattice dimensions assuming perfect syndrome measurements. \textbf{(a)} The physical error rate range from $1\%-20\%$ for $L=\{10, 20, 40, 80\}$. \textbf{(b)} The physical error rate range from $3\%-10\%$ for $L=\{10, 20, 40, 80\}$. \textbf{(c)} The physical error rate range from $3\%-10\%$ for $L=\{4, 6, 8\}$. The TQEC code for qubits on $K$ outperforms that on $T^2$ when $d$ is even. Codes with smaller $d$ show more prominent improvements.}
     \label{fig:normal}
\end{figure}

\begin{figure}
    \centering
    \includegraphics[width=\linewidth]{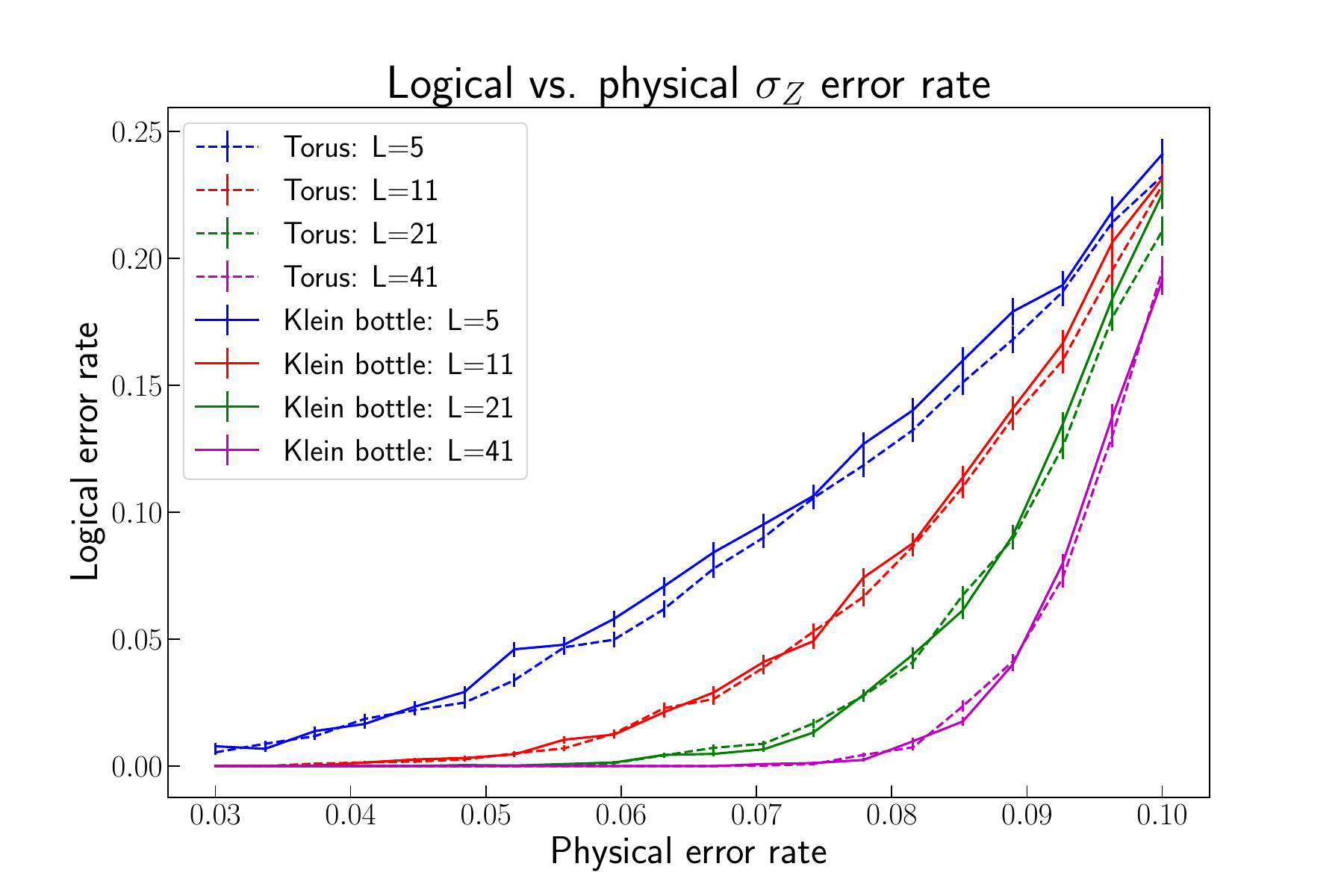}
    \caption{\raggedright Logical versus physical $\sigma_Z$ error rates for TQEC code embedded on torus $T^2$ and Klein bottle $K$ for different lattice dimensions $L=\{5, 11, 21, 41\}$ assuming perfect syndrome measurements. There is no statistically significant distinction between the performance of TQEC codes on $K$ or $T^2$ when $d$ is odd.}
    \label{fig:simulation_QBER_torus_k_003_01_odd}
\end{figure}

In the first part of this study, we simulate the physical and logical $\sigma_Z$ error rates for qubits on cell complexes embedded in the torus $T^2$ and the Klein bottle $K$ using various lattice dimensions $L = \{10, 20, 40, 80\}$. The physical error rate and the corresponding logical error rate are collected and plotted. The physical error rate $p$ is sampled from a uniformly distributed random variable. For each lattice dimension $L$, simulations are conducted over 5000 runs for both $T^2$ and $K$. The error bars in the plots represent 1-sigma uncertainties, and for simplicity, we assume noiseless syndrome measurements. Simulation results incorporating noisy syndromes will be presented in section~\ref{subsec:noisy_syndromes}.

The outcomes of these simulations are shown in Figures~\ref{fig:simulation_QBER_torus_k_001_02} and~\ref{fig:simulation_QBER_torus_k_003_01}. These figures depict the relationship between logical and physical error rates for TQEC codes embedded on $T^2$ and $K$ across different lattice dimensions. Figure~\ref{fig:simulation_QBER_torus_k_001_02} covers a physical error rate range of $1\%-20\%$, while Figure~\ref {fig:simulation_QBER_torus_k_003_01} focuses on a range of $3\%-10\%$. The results indicate that TQEC codes on $T^2$ and $K$ exhibit similar scaling behavior. However, the logical error rate for codes on the Klein bottle is consistently slightly lower across the considered physical error rates. Interestingly, this improvement is more pronounced for smaller lattice dimensions $L$ compared to larger ones, which is even clearer when considering Figure~\ref{fig:simulation_QBER_torus_k_003_01_small}. A qualitative explanation for these observations is provided in Section~\ref{subsec: discussion}.

In all simulations described above, we fix $d$ to even values. To explore the impact of odd lattice dimensions, we conducted additional simulations with $L = \{5, 11, 21, 41\}$ over a physical error rate range of $3\%-10\%$. The results are presented in Figure~\ref {fig:simulation_QBER_torus_k_003_01_odd}. In this case, the TQEC code on the Klein bottle does not demonstrate significant improvement over the Toric code. A qualitative explanation for this behavior is also provided in section~\ref{subsec: discussion}.

\subsection{Simulation with Random Qubit Errors and Noisy Syndromes}
\label{subsec:noisy_syndromes}

\begin{figure}
     \centering
     \begin{subfigure}[b]{0.5\textwidth}
         \centering
         \includegraphics[width=\textwidth]{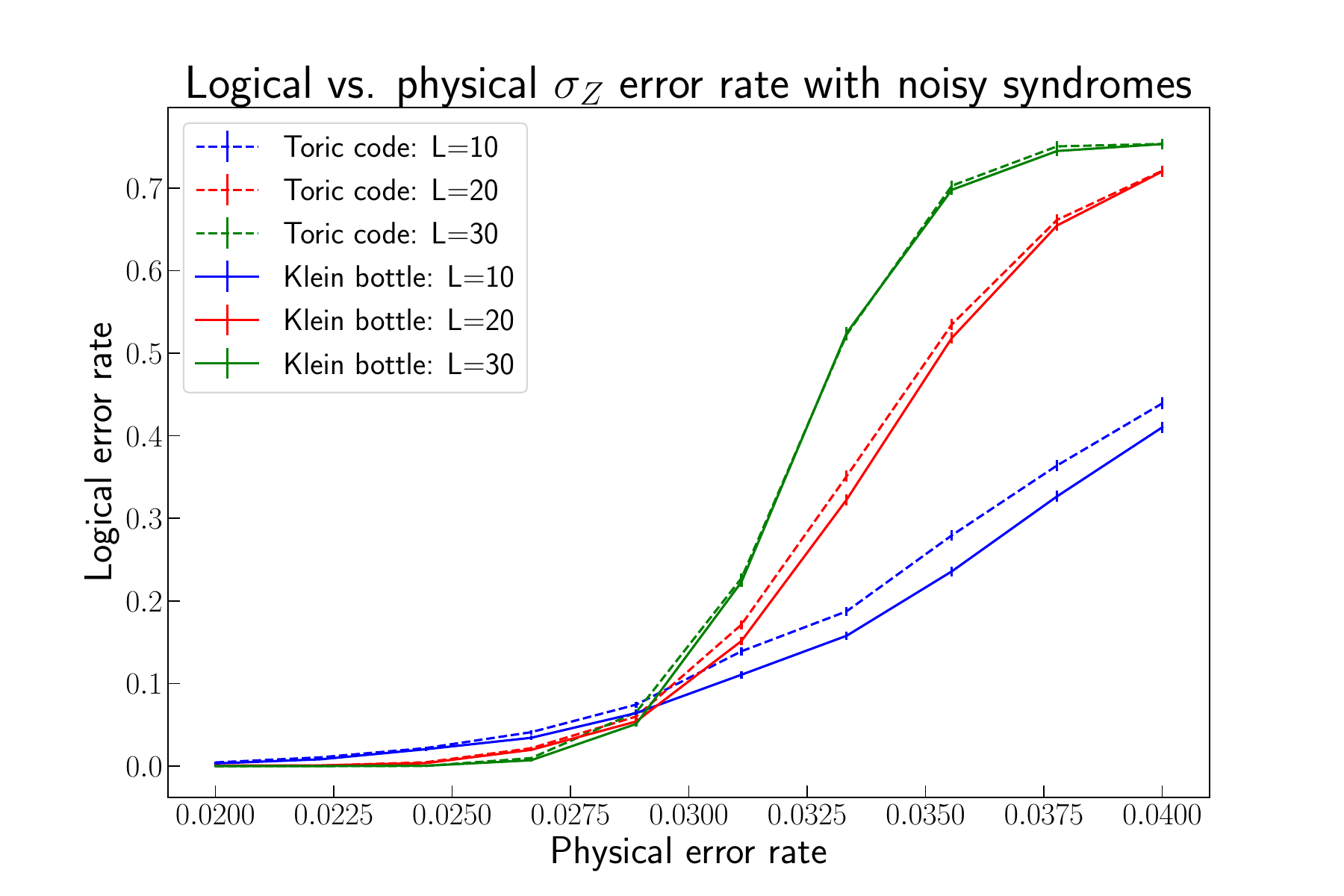}
         \caption{}
         \label{fig:simulation_QBER_torus_k_002_004_noisy}
     \end{subfigure}
     \hfill
     \begin{subfigure}[b]{0.5\textwidth}
         \centering
         \includegraphics[width=\textwidth]{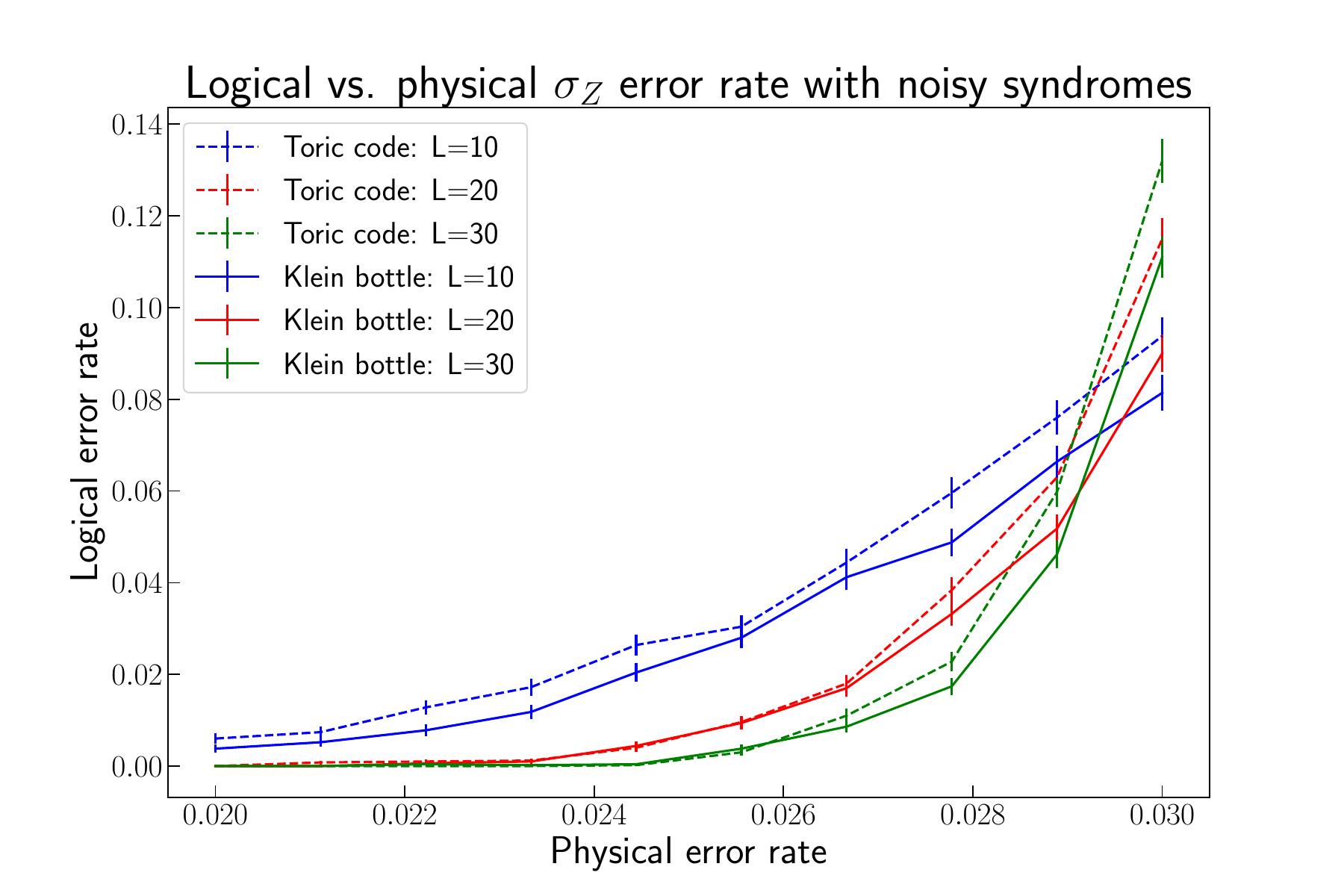}
         \caption{}
         \label{fig:simulation_QBER_torus_k_002_003_noisy}
     \end{subfigure}
     \caption{\raggedright Logical versus physical $\sigma_Z$ error rates for TQEC code embedded on torus $T^2$ and Klein bottle $K$ for different lattice dimensions $L=\{10, 20, 30\}$ with the phenomenological error model for syndrome measurements. \textbf{(a)} The physical error rate ranges from $2\%-4\%$. \textbf{(b)} The physical error rate ranges from $2\%-3\%$. The TQEC code for qubits on the $K$ outperforms that on $T^2$ when $L$ is even. Codes with smaller $L$ show more prominent improvements.}
     \label{fig:noisy_symdrome}
\end{figure}

In addition to simulating qubit errors, noisy syndrome measurements are incorporated into the simulations now. These simulations are conducted for various lattice dimensions $L = \{10, 20, 30\}$, with each data point representing the mean of 5000 runs. The error bars in the plots correspond to 1-sigma uncertainties. The simulation methodology involves the modified example codes from~\cite{pymatching1}. To mitigate errors in syndrome measurements, each measurement is repeated multiple times. Consequently, the decoding process is performed on a three-dimensional lattice instead of a two-dimensional lattice, with the third dimension representing the number of syndrome measurement repetitions~\cite{noisy_syndrome_theory1}. 

Specifically, we adopt the phenomenological error model for syndrome measurements. In this framework, qubits have a probability $p$ of experiencing errors during each measurement cycle. Additionally, there is a probability $p$ that any given syndrome measurement will yield an incorrect result, with $p$ sampled from a uniformly distributed random variable.

The results are presented in Figures~\ref{fig:simulation_QBER_torus_k_002_004_noisy} and~\ref{fig:simulation_QBER_torus_k_002_003_noisy}. A notable observation is that compared to perfect syndrome measurements shown in Figures~\ref{fig:simulation_QBER_torus_k_001_02} and~\ref{fig:simulation_QBER_torus_k_003_01}, the logical error rate increases significantly faster when syndrome measurement errors are present. Despite this increase, TQEC codes defined on the Klein bottle continue to slightly outperform the toric codes. Furthermore, the improvement is more pronounced for smaller lattice dimensions $L$. A qualitative explanation of these observations is provided next.

\subsection{Discussion}
\label{subsec: discussion}

\begin{figure}
     \centering
     \begin{subfigure}[b]{0.5\textwidth}
         \centering
         \includegraphics[width=\textwidth]{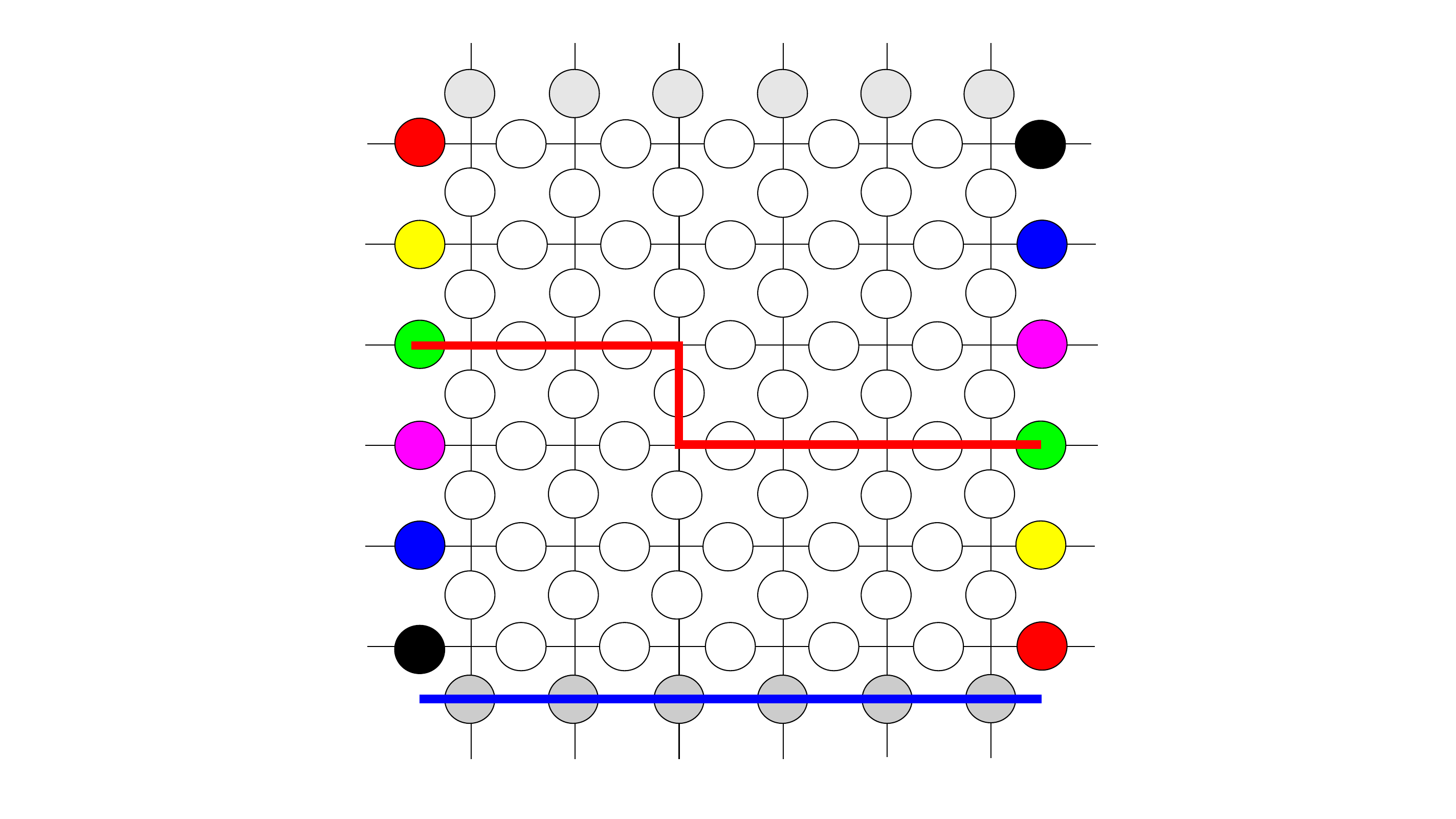}
         \caption{$L=6$}
         \label{fig:Klein_Bottle_logical}
     \end{subfigure}
     \hfill
     \begin{subfigure}[b]{0.5\textwidth}
         \centering
         \includegraphics[width=\textwidth]{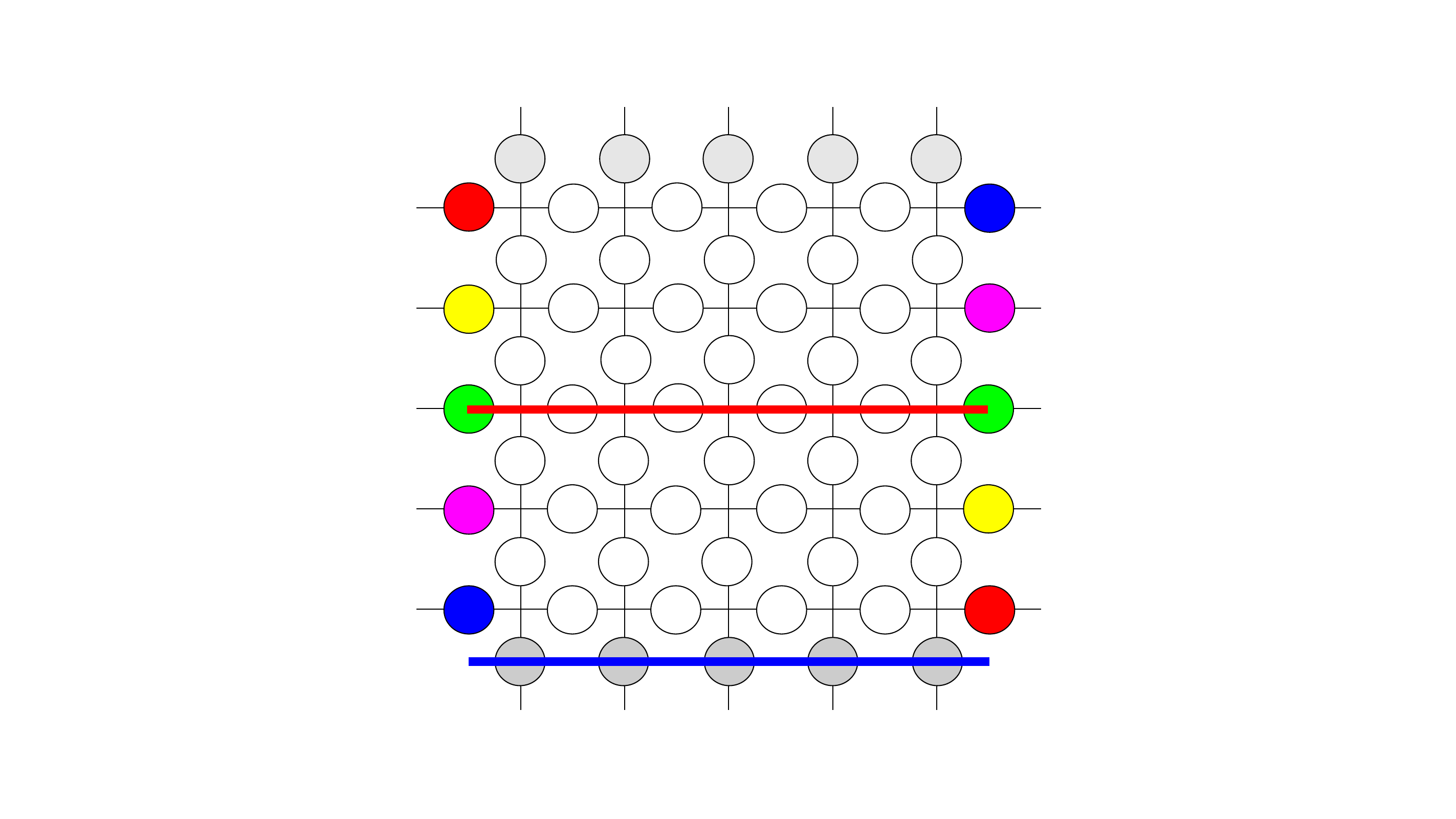}
         \caption{$L=5$}
         \label{fig:Klein_Bottle_logical_odd}
     \end{subfigure}
     \caption{\raggedright The $L\times L$ lattice structure with the Klein bottle topology for different $L$, where the boundary condition of one pair of edges is reversed. When considering the non-contractible loops across the horizontal direction, the blue line shows the shortest loop where logical $\sigma_Z$ acts on, and the red line shows the shortest loop where logical $\sigma_X$ acts on. \textbf{(a)} For even $L=6$. \textbf{(b)} For odd $L=5$.}
     \label{fig:Klein_Bottle_logicals}
\end{figure}

From the above simulations, it is evident that utilizing Klein bottles for TQEC provides a slight improvement over the original toric code when the lattice dimension is even, both with and without noisy syndromes. Notably, this improvement is more pronounced for smaller lattice sizes or smaller lattice dimensions $L$. This observation is particularly intriguing, and in this subsection, we aim to provide a qualitative explanation for this phenomenon.

Figure~\ref{fig:logical_Z_operators} illustrates examples of three distinct homology loops $\gamma_i$, where $i = 1, 2, 3$, for each 2-manifold, represented in different colors. These loops correspond to the logical $\sigma_Z$ operations. It can be observed that for the torus, all $\gamma_i$ loops have the same length. In contrast, for the Klein bottle and $\mathbb{R}P^2$, some homology loops are longer than others, even though the shortest loop (e.g., the green loop) is identical across all three 2-manifolds. Consequently, if the fundamental polygons are of equal size, the average length of homology loops is longer for the Klein bottle and $\mathbb{R}P^2$ compared to the torus. Therefore, given the same physical qubit error rate, a slight improvement in logical error rates can be expected when employing Klein bottle or $\mathbb{R}P^2$ topologies instead of the torus.

For this reason, we expect that the improvement offered by using Klein bottles becomes more significant at smaller lattice dimensions $L$. Consider a square lattice structure as shown in Figure~\ref{fig:Toric_code_stabilizers}. Sorting all logical $\sigma_Z$ operations in ascending order of length yields a sequence $\Gamma = \{\gamma_1, \gamma_2, \dots\}$ such that $\gamma_i \leq \gamma_j$ for $i < j$. For a torus, all logical loops have equal lengths: $\gamma_1 = \gamma_2 = \cdots = \gamma_d = d$. However, for Klein bottles with lattice dimension $L$, the shortest loop satisfies $\gamma_1 = d + 1$. This idea is illustrated in Figure~\ref{fig:Klein_Bottle_logical}, where the shortest homology loop $\beta_1 \in C^1$ is one qubit longer than its counterpart in the toric code. Consequently, we observe more prominent improvements for smaller values of $d$.

An additional noteworthy phenomenon emerges when we consider the implementation presented in Figure~\ref{fig:Klein_Bottle_logicals}, which we considered a lattice-structured cell complex embedded on both the torus and the Klein bottle. As illustrated in Figures~\ref{fig:Klein_Bottle_logical} and~\ref{fig:Klein_Bottle_logical_odd}, the behavior of the code depends critically on whether the lattice dimension $L$ is even or odd. Specifically, when $L$ is even, the shortest $\gamma \in C_1$, represented by the red line in the figures, increases in length by 1 compared to its counterpart in the toric code with the same $L$. This results in an increment of one in the length of the shortest logical $\sigma_X$ operator, which is responsible for decoding logical $\sigma_Z$ errors. Consequently, as shown in Figures~\ref{fig:simulation_QBER_torus_k_001_02} and~\ref{fig:simulation_QBER_torus_k_003_01}, the logical $\sigma_Z$ error rates exhibit improvement over the toric code when $L$ is even. Conversely, for odd values of $L$, the shortest $\gamma \in C_1$ remains identical for both the Klein bottle and the torus, and thus no improvement in logical error rates is observed for odd lattice dimensions $L$.

Furthermore, as illustrated in Figures~\ref{fig:Klein_Bottle_logical} and~\ref{fig:Klein_Bottle_logical_odd}, the shortest 1-cocycle $\beta \in C^1$, depicted as the blue line, has the same length for both the torus $T^2$ and Klein bottle $K$ topologies. Since $\beta$ corresponds to the logical $\sigma_Z$ operators, which are responsible for detecting logical $\sigma_X$ errors, we do not anticipate any improvement in $\sigma_X$ error rates when employing the Klein bottle topology, irrespective of whether $L$ is even or odd. The corresponding simulation results are provided in Appendix~\ref{appendix:klein_bottle_sigmax}.

\section{Conclusion}
In this paper, we have explored the theoretical foundations of Topological Quantum Error Correction (TQEC) through an algebraic topology approach. We began by reviewing key properties and long-established theorems in algebraic topology, subsequently demonstrating the relationship between TQEC capabilities and intersection properties on manifolds, supported by rigorous proofs. Focusing on 2-manifolds $M$, we established a connection between homology and cohomology groups and intersection properties, deriving the necessary and sufficient conditions for a 2-manifold to serve as a TQEC code. Specifically, we proved that a 2-manifold can be used for TQEC for qubits if and only if $H_1(M; \mathbb{Z}_2)$ is non-trivial, as formalized in Theorem~\ref{main_theorem}. This theoretical framework was then generalized to higher-dimensional manifolds, which is formalized in Theorem~\ref{thm:high_dim_main_theorem} stating that an $n$-manifold $M$ can encode qubits on the $i$-th dimension if $H_i(M;\mathbb{Z}_2)$ is non-trivial. Then, we provide a realistic example of defining TQEC codes on the 3-torus with detailed implementation schemes, with performance analysis, illustrating the potential of such higher-dimensional manifolds for TQEC applications. This opens up many wider possibilities for TQEC codes for qubits.

In the final section, we provided examples of TQEC codes on 2-manifolds beyond the toric code. Specifically, we proposed using the Klein bottle as an alternative topological structure for TQEC. We conducted simulations of TQEC codes on the Klein bottle and compared their performance with toric codes under both ideal and noisy syndrome measurement conditions. The simulation results revealed that the Klein bottle architecture offers slightly improved error correction performance compared to the original toric code.

This work builds upon prior studies of TQEC codes, which remain the only experimentally validated quantum error correction paradigm capable of achieving below-threshold error rates by far~\cite{Google1, Google2}. By proposing a broader class of TQEC codes, this paper represents a significant step forward in advancing TQEC, which will be beneficial to fault-tolerant quantum computation and communication.

\begin{acknowledgments} 

Xiang Zou expresses his sincere gratitude to Professor Dror Bar-Natan of the Department of Mathematics, University of Toronto, for his valuable comments and insightful discussions, which have greatly contributed to the development of this work. 

Xiang Zou is supported through the Scholarships for Doctoral Study by the Croucher Foundation. Hoi-Kwong Lo is supported by MITACS, NSERC, Canadian Foundation for Innovation (CFI), the High Throughput and Secure Networks Challenge Program at the National Research Council of Canada (NRC), the NRC Applied Quantum Computing Challenge Program, NUS start-up grant, and CQT Core Research Funding.
\end{acknowledgments}

\section{Data Availability}
The code used for the simulations and all presented data are available upon request.

\section{Author Contributions}
Xiang Zou contributed to the theorem, proofs, and examples in the article. Hoi-Kwong Lo conceptualized, proofread, and modified the work.
\appendix

\section{Proofs of Theorems}
\label{appendix:Proofs}
For readability, we have omitted most of the proofs of the theorems, and in this section, the omitted proofs are provided.

\begin{proof}[Theorem~\ref{Theorem: TQEC intersection}]

    Firstly, using the basic principle of QEC, by proper stabilizer measurements can collapse the set of all possible quantum errors t of discrete quantum errors. Thus, for any possible qubit error $Err$, it can be written as $Err = e_0I+e_1\sigma_X+e_2\sigma_Z+e_3\sigma_X\sigma_Z$ for some constant $e_0, e_1, e_2,e_3$. This means, for QEC purposes, a method to implement logical $\sigma_X$ and $\sigma_Z$ on the code is sufficient.

    Suppose we have a 2-manifold $M$, and there are 2 non-contractible loops on $M$. If the two non-contractible loops had an odd number of intersection points, then the logical $\sigma_X$ and $\sigma_Z$ operators can be defined on each non-contractible loop, respectively. Since they intersect at an odd number of sites, the logical $\sigma_X$ and $\sigma_Z$ operators anti-commute, and this defines a TQEC. This completes the proof of the ``if" part.

    Similarly, for the ``only if'' part, if a loop-based TQEC code exists, then by definition, we need to define the logical $\sigma_X$ and $\sigma_Z$ operators, which operate on two non-contractible loops. Now, to ensure that they anti-commute, the parity of the number of intersection sites must be odd.

\end{proof}

\begin{proof}[Corollary~\ref{Theorem: Simply connected surface}]
    Theorem~\ref{Theorem: TQEC intersection} states that, for any closed surfaces to be used as a TQEC code, we require at least 2 loops that interact in an odd number of points. Theorem~\ref {mod2intersection} states that the modulo-2 intersection number is homotopic invariant. However, for any 2 non-overlapping loops on simply connected surfaces, each can be shrunk down to a point continuously, where they do not intersect. This means that two non-overlapping loops can only transversely intersect at an even number of points. Thus, they cannot be used as TQEC codes.
\end{proof}

\begin{proof}[Theorem~\ref{main_theorem}]

    Notation-wise, since $\mathbb{Z}_2$ is commutative, we will represent it as the modulo-2 addition group that the elements of $\mathbb{Z}_2$ are $\{0, 1\}$, and for simplicity, for $a\in\{0, 1\}$ we represent $\sum_{i=1}^n a:=a^n$.

    Firstly, we prove that the closed and compact 2-manifold $M$ can be used as a TQEC code for qubits only if its first homology group is nontrivial. This is simple, because if $H_1(M; \mathbb{Z}_2)=0$, then by definition, $M$ is simply connected. As shown in Theorem \ref{Theorem: Simply connected surface}, all simply connected surfaces cannot be used for TQEC purposes. This completes the proof of the ``only if " part.

    Then, to prove the ``if" part, considering elements belonging to the $H_1(M;\mathbb{Z}_2)$, $\gamma\in C_1(M;\mathbb{Z}_2), [\gamma]\ne 0 \text{ and } \partial \gamma=0$. And elements belongs to the $H^1(M;\mathbb{Z}_2)$, $\beta\in C^1(M;\mathbb{Z}_2), [\beta]\ne 0 \text{ and } \delta \beta=0$. The Poincaré duality as in the theorem~\ref{Poincaré} states that $H_1(M; \mathbb{Z}_2)\cong H^{1}(M; \mathbb{Z}_2)$. Thus, the pairing given by a function $p:H_1(M; \mathbb{Z}_2)\times H^1(M; \mathbb{Z}_2)\rightarrow \mathbb{Z}_2$ is well-defined, and the $\mathbb{Z}_2$ coefficients ensure that $[\gamma]^2=0$ and $[\beta]^2=0$. By the properties of elements belonging to the first cohomology group, for arbitrary integers $r, s$, we have $[\beta]^{2r+1}([\gamma]^{2s+1})=[\beta]([\gamma])=1$, while $[\beta]^{2r}([\gamma]^{s})=[\beta]^{r}([\gamma]^{2s})=0$. The Poincaré duality states that since $H_1(M;\mathbb{Z}_2)$ is isomorphic to $H^1(M;\mathbb{Z}_2)$, thus, there is a bijective mapping between elements in $C^1$ to elements in $C_{1}$, so since $[\gamma]$ is a class of closed loops, $[\beta]$ are also closed loops on the surface $M$. Therefore, logical $\sigma_X$ can be defined on $[\gamma]$ and logical $\sigma_Z$ can be defined on $[\beta]$, where they anti-commute, which defines a TQEC code. 

    From corollary~\ref{H1duality}, $H_1(M; \mathbb{Z}_2)\cong H^1(M; \mathbb{Z}_2)$. The number of generators determines the number of different classes of 1-chains that can be defined on the surface, and each class of 1-chain has its unique 1-cochains as from equation~\ref{pairing}, which means, each generator in $H_1(M;\mathbb{Z}_2)$ can be used to encode 1-qubit worth of quantum information.

\end{proof}

\begin{proof}[Theorem~\ref{thm:high_dim_main_theorem}]

    The proof is similar to that of 2-manifolds. Notation-wise, since $\mathbb{Z}_2$ is communtative, we will represent it as the modulo-2 addition group that the elements of $\mathbb{Z}_2$ are $\{0, 1\}$, and for simplicity, for $a\in\{0, 1\}$ we represent $\sum_{i=1}^n a:=a^n$.
    
    Firstly, we prove that the closed and compact $n$-manifold $M$ can be used for a TQEC code for qubits on the $i$-cycle and $i$-cocycle only if its $i$-th homology group is nontrivial. The reason being that if $H_i(M; \mathbb{Z}_2)=0$, then $b_i(M;\mathbb{Z}_2)=0$. There is only 1 element in the group $H_i(M; \mathbb{Z}_2)$, which means all closed $i$-cycles are equivalently null-homologous. This implies all $\gamma\in C_i(M;\mathbb{Z}_2), \partial \gamma = 0$ are equivalent, that $\gamma^s=\gamma$ for all positive integer $s$. By the definition of elements belonging to the cohomology group $H^i(M;\mathbb{Z}_2)$, even if there exist $\beta\in C^i(M;\mathbb{Z}_2), [\beta]\ne 0 \text{ and } \delta \beta=0$. For all non-negative integer $r$, $\beta(\gamma^r)=\beta(\gamma)=0$. However, for $\sigma_X$ operator, $\sigma_X^{2r+1}=\sigma_X$ and $\sigma_X^{2r}=\mathbb{I}$ for an arbitrary integer $r$. This implies that $\sigma_X$ cannot be defined on any $i$-cycles; hence, defining a TQEC code is impossible.

    Then, to prove the ``if" part, considering elements belonging to the $H_i(M;\mathbb{Z}_2)$, $\gamma\in C_i(M;\mathbb{Z}_2), [\gamma]\ne 0 \text{ and } \partial \gamma=0$. And elements belongs to the $H^i(M;\mathbb{Z}_2)$, $\beta\in C^i(M;\mathbb{Z}_2), [\beta]\ne 0 \text{ and } \delta \beta=0$. The Poincaré duality as in the theorem~\ref{Poincaré} states that $H_i(M; \mathbb{Z}_2)\cong H^{n-i}(M; \mathbb{Z}_2)$. Thus, the pairing given by a function $p:H_i(M; \mathbb{Z}_2)\times H^i(M; \mathbb{Z}_2)\rightarrow \mathbb{Z}_2$ is well-defined. The $\mathbb{Z}_2$ coefficients ensure that $[\gamma]^2=[\beta]^2=0$. And by the definition of elements belonging to the cohomology group, for arbitrary integers $r, s$, we have $[\beta]^{2r+1}([\gamma]^{2s+1})=[\beta]([\gamma])=1$, while $[\beta]^{2r}([\gamma]^{s})=[\beta]^{r}([\gamma]^{2s})=0$. $[\gamma]$ represent the class of non-trivial closed $i$-cycles and $[\beta]$ represent the class of non-trivial closed $i$-cocycles, which is isomorphic to the class of non-trivial closed $(n-i)$-cycles. Therefore, logical $\sigma_X$ can be defined on $[\gamma]$ and logical $\sigma_Z$ can be defined on $[\beta]$, where they anti-commute, which defines a TQEC code. 

    Since $H_i(M; \mathbb{Z}_2)\cong H^{n-i}(M; \mathbb{Z}_2)$. The rank of $H_i(M;\mathbb{Z}_2)$ determines the number of different classes of $i$-chains that can be defined on $M$, and each class of $i$-chain has its unique $i$-cochains as from equation~\ref{pairing}, which means, each independent generator of $H_1(M;\mathbb{Z}_2)$ can be used to encode 1-qubit worth of quantum information.
\end{proof}

\begin{proof}[Theorem~\ref{thm:high_dim_main_theorem_boundary}]
    The proof is similar to that of Theorem~\ref{thm:high_dim_main_theorem}. Given a compact $n$-manifold $M$ with boundary $\partial M=A\cup B$ and $\partial A = \partial B =A\cap B$, the reason why we decompose $\partial M$ into two parts $A$ and $B=\partial M\setminus A\cup \partial A$ is that, for qubit codes, there are only 2 types of boundary conditions applied, $x$-type boundary and $z$-type boundary as proven by Kitaev in Ref.\cite{Surfacecode1, Kitaev_anyons_computation1}. Thus, when we have $\partial M=A\cup B$, we may use $A$ as the $x$-boundary and $B$ as the $z$-boundary or vice versa, then the proof follows. 
    
    Firstly, we prove that $M$ can be used for a TQEC code for qubits on the $i$-cycle and $i$-cocycle only if $\exists A, s.t., H_i(M, A; \mathbb{Z}_2)\ne0$. The reason being that if $\forall A, H_i(M, A; \mathbb{Z}_2)=0$, then $\forall A, b_i(M, A;\mathbb{Z}_2)=0$. There is only 1 element in the group $H_i(M, A; \mathbb{Z}_2), \forall A$. Which means all closed $i$-cycles are equivalently null-homologous. This implies all $\gamma\in C_i(M, A;\mathbb{Z}_2), \partial \gamma = 0, \forall A$ are equivalent, that $\gamma^s=\gamma$ for all positive integer $s$. This implies that $\sigma_X$ cannot be defined on any $i$-cycles; hence, defining a TQEC code is impossible.

    Then, to prove the ``if" part, suppose $\exists A, s.t., H_i(M, A; \mathbb{Z}_2)\ne0\text{  and  } H_{n-i}(M, B; \mathbb{Z}_2)$. Considering elements belonging to the $H_i(M, A;\mathbb{Z}_2)$, $\gamma\in C_i(M, A;\mathbb{Z}_2), [\gamma]\ne 0 \text{ and } \partial \gamma=0$. And elements belongs to the $H^i(M, A;\mathbb{Z}_2)$, $\beta\in C^i(M;\mathbb{Z}_2), [\beta]\ne 0 \text{ and } \delta \beta=0$. The Lefschetz duality in Theorem~\ref{thm:Lefschetz_duality} states that $H^i(M, A; \mathbb{Z}_2)\cong H_{n-i}(M, B; \mathbb{Z}_2)$. Thus, the pairing given by a function $p:H_i(M, A; \mathbb{Z}_2)\times H^i(M, A; \mathbb{Z}_2)\rightarrow \mathbb{Z}_2$ is well-defined. The $\mathbb{Z}_2$ coefficients ensure that $[\gamma]^2=[\beta]^2=0$. And by the definition of elements belonging to the cohomology group, for arbitrary integers $r, s$, we have $[\beta]^{2r+1}([\gamma]^{2s+1})=[\beta]([\gamma])=1$, while $[\beta]^{2r}([\gamma]^{s})=[\beta]^{r}([\gamma]^{2s})=0$. $[\gamma]$ represent the class of non-trivial closed $i$-cycles and $[\beta]$ represents the class of non-trivial closed $i$-cocycles with the presence of boundary $B$, which is isomorphic to the class of non-trivial closed $(n-i)$-cycles with the presence of boundary $A$. Therefore, logical $\sigma_X$ can be defined on $[\gamma]$ and logical $\sigma_Z$ can be defined on $[\beta]$, where they anti-commute, which defines a TQEC code. 

    Since from the Lefschetz duality in Theorem~\ref{thm:Lefschetz_duality}, $H^i(M, A; \mathbb{Z}_2)\cong H_{n-i}(M, B; \mathbb{Z}_2)$. The rank of $H_i(M, A;\mathbb{Z}_2)$ determines the number of different classes of $i$-chains that can be defined on $M$, and the number of different classes of $i$-cochain equals the rank of $H_{n-i}(M, B;\mathbb{Z}_2)$. Which means, the number of qubits encoded is determined by the smaller of $b_i(M, A;\mathbb{Z}_2)$ and $b_{n-i}(M, B;\mathbb{Z}_2)$.
\end{proof}

\begin{proof}[Theorem~\ref{thm:TQEC_for_Qudit}]
We will prove that such a presentation complex $X_G=(V, E, F)$ built from a group $G=\mathbb{Z}_d^k$ can encode a $d$-dimensional qudits by providing a way to define the logical $X^a$ and $Z^b$ operators of the qudits.

From the construction of a presentation complex from a group $\mathbb{Z}_d$, the idea is to use a d-pointed ``asterisk'' with a $1/d$ twist after a revolution, just as shown in figure~\ref{fig:Z_3_Cell_Complex}, and then attach a disk along the boundary circle of the d-pointed ``asterisk''. The logical $X$ and logical $Z$ are defined to be closed loops on the presentation complex; they wrap around the presentation complex in two different directions, as shown in Figure~\ref{fig:Z_3_operators}.

Similar to the normal qubit code, the Hamiltonian is defined to be the summation of all stabilizers defined in Figure~\ref{fig:Qudit_code_stabilizer}, and both the logical $X$ and $Z$ operators commute with the Hamiltonian~\ref{eq:qudit_code_hamiltonian}. Just a reminder, the single-qudit operator $X$ and $Z$ are defined in equation~\ref{eq:XZ_operator_qudit}. So, $X^aZ^b=\omega^{ab} Z^bX^a$, for $\omega = e^{2\pi i/d}$.

In other words, the logical $X$ operates on the homological cycles $C_1(X_G)$, $\gamma\in C_1(X_G;\mathbb{Z}_d), [\gamma]\ne 0 \text{ and } \partial \gamma=0$, while the logical $Z$ operates on the co-cycles $C^1(X_G)$, $\beta\in C^1(X_G;\mathbb{Z}_d), [\beta]\ne 0 \text{ and } \delta \beta=0$. Since $H_1(X_G)=\mathbb{Z}_d$, there should be $d$ distinct elements in $C_1(X_G;\mathbb{Z}_d)$, which corresponds to revolving around the presentation complex $0, \cdots, d-1$ times. Hence, the logical $X^a$ operates on a $\gamma\in C_1(X_G;\mathbb{Z}_d), [\gamma] = (a \mod d)$, so for any logical $X^a$, there exist a $\gamma\in C_1(X_G;\mathbb{Z}_d)$. 

By the properties of the presentation complex and group elements belonging to the first homology and cohomology groups~\cite{Algebraic_Topology1}, for an arbitrary power of the operator $X^r$ and $Z^s$, denoting $a=r\mod d, b = s\mod d$. We have $[Z]^{r}([X]^{s})=[Z]^a([X])^b$, where the two loops intersects at $ab$ number of sites. Thus, the exchange statistics will be $X^aZ^b=\omega^{ab}Z^bX^a$. 

Since both logical operators commute with the Hamiltonian, they simultaneously diagonalize the quantum state, which is denoted as $\ket{g}$. Then, we have
\begin{equation}
Z^b\ket{g}=\omega^{bg}\ket{g}
\end{equation}

We want to investigate the effect of $X$ acting on the state $\ket{g}$.

\begin{equation}
\begin{split}
Z^bX^a\ket{g}&=\omega^{ab}X^aZ^b\ket{g}\\
&=\omega^{ab} X^a \omega^{bg}\ket{g}\\
&=\omega^{b(a+g)}X^a\ket{g}
\end{split}
\end{equation}
Hence, 
\begin{equation}
    X^a\ket{g} = \ket{(g+a) \mod d}
\end{equation}
Hence, the above definition of the operators works as the logical $X$ and logical $Z$ for the encoded qudit.

Now, the only problem left is whether, when the logical operators get detoured, any of the above arguments will be violated? The answer is no. Remember that the sign ($\pm 1$) of each intersection point is well-defined since the presentation complex is oriented. For two non-detoured logical operators $X$ and $Z$ intersecting at 1 site, we define this intersection point to have a sign $+1$. 

Up to any curve homotopy, the number of intersection points will only be increased by an even number; this is a well-known result~\cite{Algebraic_Topology1, Algebraic_Topology2}, and among the newly emerged intersection points, always half of them have a sign of $+1$ and half are $-1$, this is also a well-known textbook result~\cite{Algebraic_Topology1, Algebraic_Topology2}. This means, in the detoured case, if 2 intersection points are added, one of them would be an intersection between $\{X, Z^{-1}\}$ or $\{X^{-1}, Z\}$, while the other intersection would be between $\{X, Z\}$ or $\{X^{-1},Z^{-1}\}$, then the two intersection will ``cancel out'', and not affecting the overall phase of the logical operators.
    
\end{proof}

\section{Mathematical Foundation}
\label{appendix:Mathematical_foundation}

Several mathematical details are omitted in the main text, and this appendix aims to fill the gaps by providing details for several terms mentioned in the main text.

\subsection{Curve Homotopy and Transversality}
\label{appendix:homotopy_transversality}

Firstly, a curve's homotopy is defined as follows.
\begin{definition}
   A homotopy of paths in $X$ is a family $f_t:[0, 1]\rightarrow X, t\in[0, 1]$, such that $f_t(0)=x_0, f_t(1)=x_1$ for all $t$, and the associated map $F:[0, 1]\times [0, 1] \rightarrow X$ defined by $F(s, t)=f_t(s)$ is continuous.
\end{definition}

Also, we mentioned in section~\ref{intersection_on_2_manifolds} that intersections are only well-defined as two sub-manifolds intersecting transversely, and transversality is defined as below~\cite{Differential_Topology1, Algebraic_Topology1, Algebraic_Topology2}

\begin{definition}[Transversality]
    Two sub-manifolds $A$ and $B$ of a given finite-dimensional smooth manifold are said to intersect transversally if at every point of intersection, their separate tangent spaces at that point together generate the tangent space of the ambient manifold, denoted as $A \pitchfork B$.
\end{definition}

Intuitively, for 2 sub-manifolds $A$ and $B$, a transverse intersection point can be considered the point where the two tangent spaces of $A$ and $B$ at that point do not overlap too much.

\subsection{Simplexes and Cell Complexes}
\label{appendix:simplex_cell_complex}

In section~\ref{sec:TQEC_on_2_manifolds}, the concept of simplexes has been mentioned a few times, and the formal definition is presented here. We followed the definition from section 2.1 in~\cite{Algebraic_Topology1}. An $n$-simplex can be considered as a generalized triangle consisting of $(n+1)$ points $p_0, p_1, \cdots, p_n$, that do not lie on the same hyperplane of dimension $m$ for all $m<n$. Also, if we normalize every point, they can be treated as vectors. Hence, a simplex can be represented as a collection of the length of the vectors $[p_0, p_1, \cdots, p_n]$, for instance, a standard $n$-simplex is 
\begin{equation}
    \Delta^n = \{(t_0, \cdots, t_n)\in\mathbb{R}^{n+1}|\sum_it_i=1 \text{ and } t_i\ge0 \text{ for all i}\}
\end{equation}

Then, a $\Delta$-comple structure on a space $M$ is a collection of maps $\sigma_\alpha:\Delta^n\rightarrow M$, for index $\alpha$.

Furthermore, most of the formation of TQEC codes on manifolds is related to cell complexes embedded in the manifold, as discussed in section~\ref{sec: definition of TQEC}. Here, we give a complete construction process of a cell complex. A cell complex or CW complex, $X$, can be constructed as follows.
\begin{enumerate}
    \item Start with a discrete set $X^0$, where points in $X^0$ are called the 0-cells.
    \item Repeatedly, we form the n-skeleton $X^n$ from $X^{n-1}$ by attaching $n$-cells $e^n_\sigma$ through the ``gluing'' map $\psi_\alpha : S^{n-1}\rightarrow X^{n-1}$. Suppose there are $A_n$ many $n$-cell $e^n_\sigma$ attached
    \begin{equation}
        \left. X^n=X^{n-1}\sqcup \bigsqcup_{\sigma\in A_n}D^n_{\sigma} \middle / x\sim \psi_\sigma\text{ for } x\in\partial D^n_\sigma \right.
    \end{equation}
\end{enumerate}

Where the $\sqcup$ represents the disjoint union. Then, we have
\begin{equation}
    \emptyset \subset X^0\subset X^1\subset X^2 \subset \dots
\end{equation}
Therefore, the chain complex of the cell complex $X$ can be written as $C^{cell}_n(X)=\langle C^n_\sigma\rangle= \langle A_n\rangle$. We name $X^n$ as the $n$-skeleton of the cell complex $X$; this definition is very useful when we discussed the TQEC codes on some higher-dimensional manifolds $X_M$ in Section~\ref{sec: TQEC Codes on Higher Dimensional Manifolds}.

\subsection{Manifolds and Orientation}
\label{subsec:manifolds_and_orientation}

Before we present the definition of oriented manifolds and cell complexes, we first need to define charts and atlases on manifolds. These definitions are widely discussed in numerous mathematical textbooks, such as Ref.~\cite{Differential_Topology1, Algebraic_Topology2}.

\begin{definition}[Chart and (smooth) atlas]
Let \(M\) be a topological space and fix \(n\in\mathbb{N}\).
A chart on \(M\) is a pair \((U,\varphi)\) where \(U\subset M\) is open and
\(\varphi:U\to \varphi(U)\subset \mathbb{R}^n\) is a homeomorphism onto an open subset.

A smooth atlas on \(M\) is a collection \(\mathcal{A}=\{(U_\alpha,\varphi_\alpha)\}\)
of charts such that:
\begin{itemize}
\item The domains cover \(M\), i.e.\ \(\bigcup_\alpha U_\alpha = M\). 
\item Any two charts are smoothly compatible: for all \(\alpha,\beta\) with
\(U_\alpha\cap U_\beta\neq\varnothing\), the transition map
\(\varphi_\beta\circ\varphi_\alpha^{-1}\) (defined on \(\varphi_\alpha(U_\alpha\cap U_\beta)\))
is a smooth map between open subsets of \(\mathbb{R}^n\).
\end{itemize}
\end{definition}

Then, we can use the definition chart and atlas to define the orientability of a manifold.

\begin{definition}[Oriented smooth manifold]
Let \(M\) be an \(n\)-dimensional smooth manifold.
An \emph{orientation} on \(M\) is an equivalence class of atlases
\(\{(U_\alpha,\varphi_\alpha)\}\), where each
\(\varphi_\alpha : U_\alpha \to \mathbb{R}^n\) is a diffeomorphism onto its image,
such that for any pair of overlapping charts
\((U_\alpha,\varphi_\alpha)\), \((U_\beta,\varphi_\beta)\),
the transition map
\(\varphi_\beta \circ \varphi_\alpha^{-1} : \varphi_\alpha(U_\alpha \cap U_\beta)
\to \varphi_\beta(U_\alpha \cap U_\beta)\)
has Jacobian determinant everywhere positive.
The pair \((M,\mathcal{O})\), where \(\mathcal{O}\) is such an orientation,
is called an oriented smooth manifold.
\end{definition}

\begin{definition}[Oriented cell complex]
Let \(X\) be a cell complex and let \(e^n\) be an \(n\)-cell attached via
a characteristic map \(\phi : D^n \to X\).
An \emph{orientation} of the cell \(e^n\) is a choice of generator
of the top relative homology group
\(H_n(D^n,\partial D^n;\mathbb{Z}) \cong \mathbb{Z}\),
or equivalently, a choice of orientation on the model ball \(D^n\).
An oriented cell complex is a cell complex equipped with a choice of
orientation for every cell such that whenever \(e^n\) is attached along
\(\phi|_{\partial D^n} : \partial D^n \to X^{(n-1)}\), the induced
orientation on \(\partial D^n\) agrees with the orientations of the
\((n-1)\)-cells it meets, in the sense that all attaching degrees are \(+1\).
\end{definition}

The definitions can be intuitively understood. A manifold is orientable if it admits a consistent global choice of local orientations, and non-orientable if no such global choice exists. Equivalently, non-orientability can be detected by the existence of a loop along which “clockwise” returns as “counterclockwise,” meaning local orientation flips after going around the loop, which is also equivalent to containing a subset homeomorphic to a Möbius strip.

Therefore, on Non-oriented manifolds, a global choice of local orientations is not well-defined, and the direction of intersections between 2 loops is also not well-defined. However, on oriented surfaces, since one can adopt a well-defined global choice of local orientations, the direction of intersection of several loops can be defined.

\subsection{Pairing and Cup Product}
\label{appendix:pairing}

During our discussion of Poincaré duality in equation~\ref{Poincaré}, for a manifold $M$ and a coefficient ring $R$, we have constructed a map from $H^k(M;R)\rightarrow H_{n-k}(M;R)$ via the cap product $\cap$. Here, we would like to formally define the cap product, the definition can be found in any algebraic topology textbook, such as~\cite{Algebraic_Topology1, Algebraic_Topology2}.

For an arbitrary manifold $M$ and a coefficient ring $R$, the cap product is defined as $\cap : C_k(M;R)\times C^l(M;R) \rightarrow C_{k-l}(M;R)$ for $k\ge l$ by which
\begin{equation}
    \sigma\cap\psi = \psi(\sigma|[v_0, \cdots, v_l])\sigma|[v_l, \cdots, v_k]
\end{equation}
for $\sigma: \Delta^k\rightarrow M$ that $\sigma$ is a map from simplexes to the manifolds, and $\psi \in C^l(M;R)$. Intuitively, the cap product $\cap$ can be thought of as the intersection of chain and cochains on the manifold $M$. 

\subsection{The Betti Number and Homology}
\label{appendix:betti_number}

The Betti number is used to count the number of independent generators of homology groups~\cite{Algebraic_Topology1}. Informally, the $k$-th Betti number of a manifold $M$ is the rank of the $H_k(M)$. It is very useful when we are discussing the error-correcting ability of using $M$ for TQEC, because the number of encoded qubits is closely related to the rank of the homology groups. The formal definition is as follows:

\begin{definition}
Let $M$ be a topological manifold and $H_k(M; R)$ be its $k$-th Homology group with coefficients in the commutative ring $R$. The $k$-th Betti number of $M$, denoted $b_k(M;R)$ is defined as:
\begin{equation}
    b_k(X;R)=dim_R\ H_k(M;R)
\end{equation}

That is, $b_k(M;R)$ is the dimension of the $k$-th homology group of $M$ with coefficients in $R$. Equivalently, when $H_k$ is a finitely generated Abelian group, then:

\begin{equation}
    b_k(X;R)=\text{rank}\ H_k(M;R)
\end{equation}
\end{definition}

Betti Number is used in many other places, including the definition of Euler Characteristics~\cite{Euler_Characteristic1}.

\subsection{Connected Sum Operation}
\label{appendix:connected_sum}

In section~\ref{sec:TQEC_on_2_manifolds}, when providing the examples, we mentioned that all 2-manifolds can be classified into two categories, where all orientable surfaces are homeomorphic to $gT^2$, and all non-orientable surfaces are homeomorphic to $gP^2$. And $gT^2:=(g-1)T^2\#T^2$ and $gP^2 = (g-1)P^2\#\mathbb{R}P^2$, with the $\#$ representing the connected sum operation. We would like to discuss the connected sum of two $n$-manifolds, $M$ and $N$. 

First, we start by a $n$-dimensional ball $B^n$, that defines two inclusion maps $\iota_M:B^n\xhookrightarrow{}M$ and $\iota_N:B^n\xhookrightarrow{}N$. Then, denote the two spaces $M'=M-\iota_M(B^n)$ and $N'=N-\iota_N(B^n)$. Then, we identify the boundary of $M'$ and $N'$, and find a homeomorphic map $f:\partial M'\rightarrow \partial N'$. Thus, formally, the connected sum between $M$ and $N$ is defined as,

\begin{equation}
    M\#N:=\left. (M'\sqcup N') \middle / x\sim f(x) \text{ for } x\in\partial M'\right.
\end{equation}

\subsection{Universal Coefficient Theorem}

We next cite another foundational theorem from cohomology theory~\cite{Algebraic_Topology1, Algebraic_Topology2}.

\begin{theorem}\textbf{(Universal Coefficient Theorem over finite field $F$)}
For a space $X$, for any finite field $F$, there is an isomorphism~\cite{Algebraic_Topology1, Algebraic_Topology2}
    \begin{equation}
        H^n(X;F)\approx Hom_F(H_n(X;F), F)
    \end{equation}
\end{theorem}

Each cohomology class in $H^n(X; F)$ can be uniquely identified with a class of linear functionals on the homology group $H_n(X; F)$, where each functional maps elements of $H_n(X; F)$ to $F$. Consequently, in principle, there exists a class of functionals capable of determining whether a loop $\gamma \in C_1(X; \mathbb{Z}_2)$ satisfies $[\gamma] = 0$ or $[\gamma] \neq 0$, which can be leveraged for TQEC purposes. However, in the absence of Poincaré duality, this class of linear functionals is not directly associated with loops in the dual cell complex $\bar{X}_G$. This contrasts with the case of 2-manifolds, where the dual cell complex can be embedded within the same manifold as a consequence of Poincaré duality. Therefore, according to this theorem, TQEC codes can, in theory, be constructed on more exotic two-dimensional cell complexes, albeit with increased implementation complexity.

There are more subtle things to keep in mind. Consider a cell complex $X$ with $H_1(X) \cong \mathbb{Z}_4$. Since $\mathbb{Z}_4$ contains a 2-torsion subgroup, restricting to $\mathbb{Z}_2$ coefficients yields a non-trivial $H_1(X; \mathbb{Z}_2)$, making such a complex theoretically possible for TQEC for qubits. Conversely, cell complexes that are unsuitable for TQEC exist. For instance, a cell complex $X'$ with $H_1(X') \cong \mathbb{Z}_3$ cannot be used for TQEC for qubits, as $\mathbb{Z}_3$ lacks a 2-torsion subgroup; thus, when restricted to $\mathbb{Z}_2$ coefficients, the resulting homology group is trivial. By Theorem~\ref{main_theorem}, such surfaces are not viable for TQEC for qubits. However, as discussed in~\cite{HomologicalQEC1}, it is potentially suitable for TQEC for qudits with the correct dimension.

\section{Simulation on Klein Bottle Code for Pauli X Error}
\label{appendix:klein_bottle_sigmax}

\begin{figure}
     \centering
     \begin{subfigure}[b]{0.5\textwidth}
         \centering
         \includegraphics[width=\textwidth]{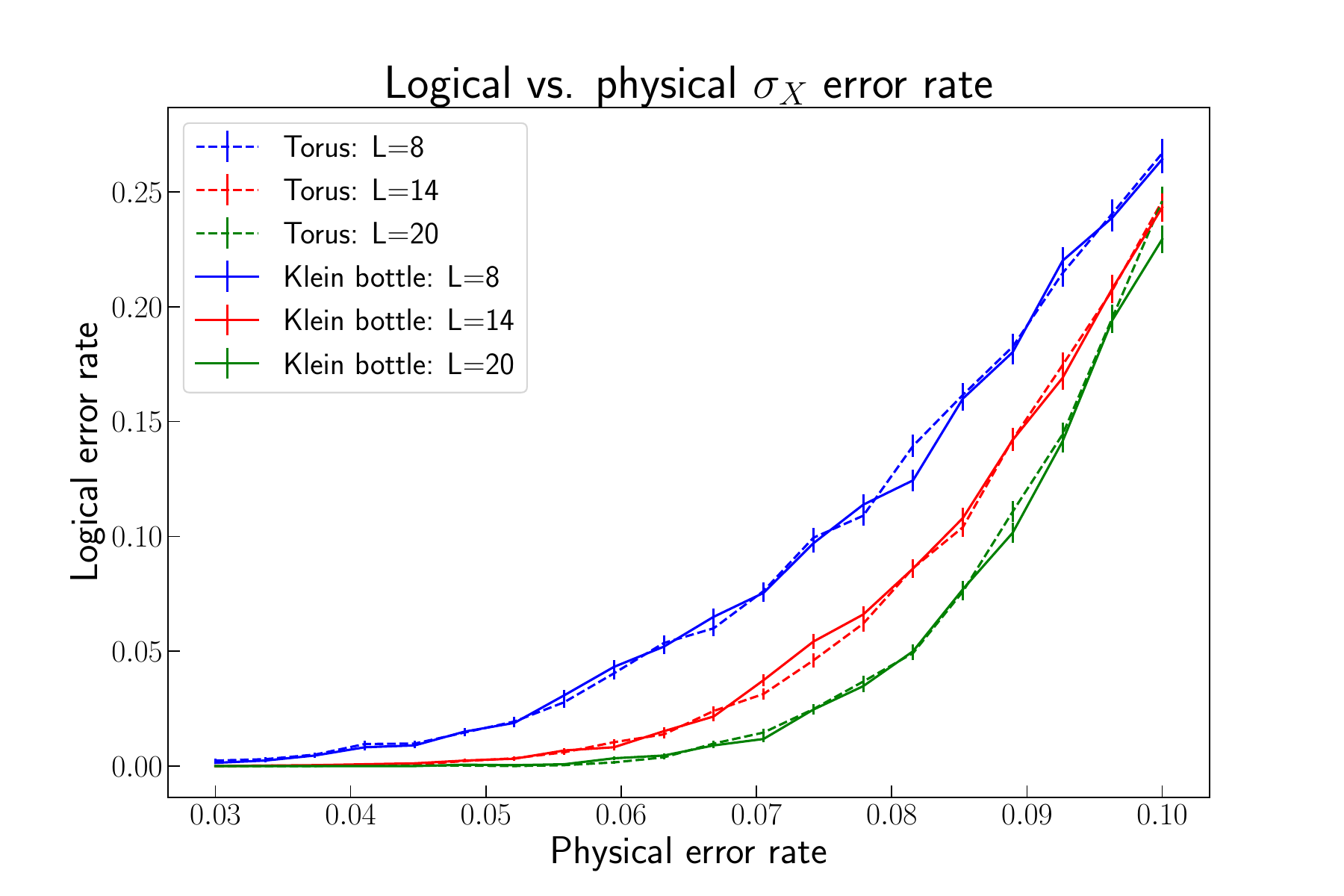}
         \caption{}
         \label{fig:even}
     \end{subfigure}
     \hfill
     \begin{subfigure}[b]{0.5\textwidth}
         \centering
         \includegraphics[width=\textwidth]{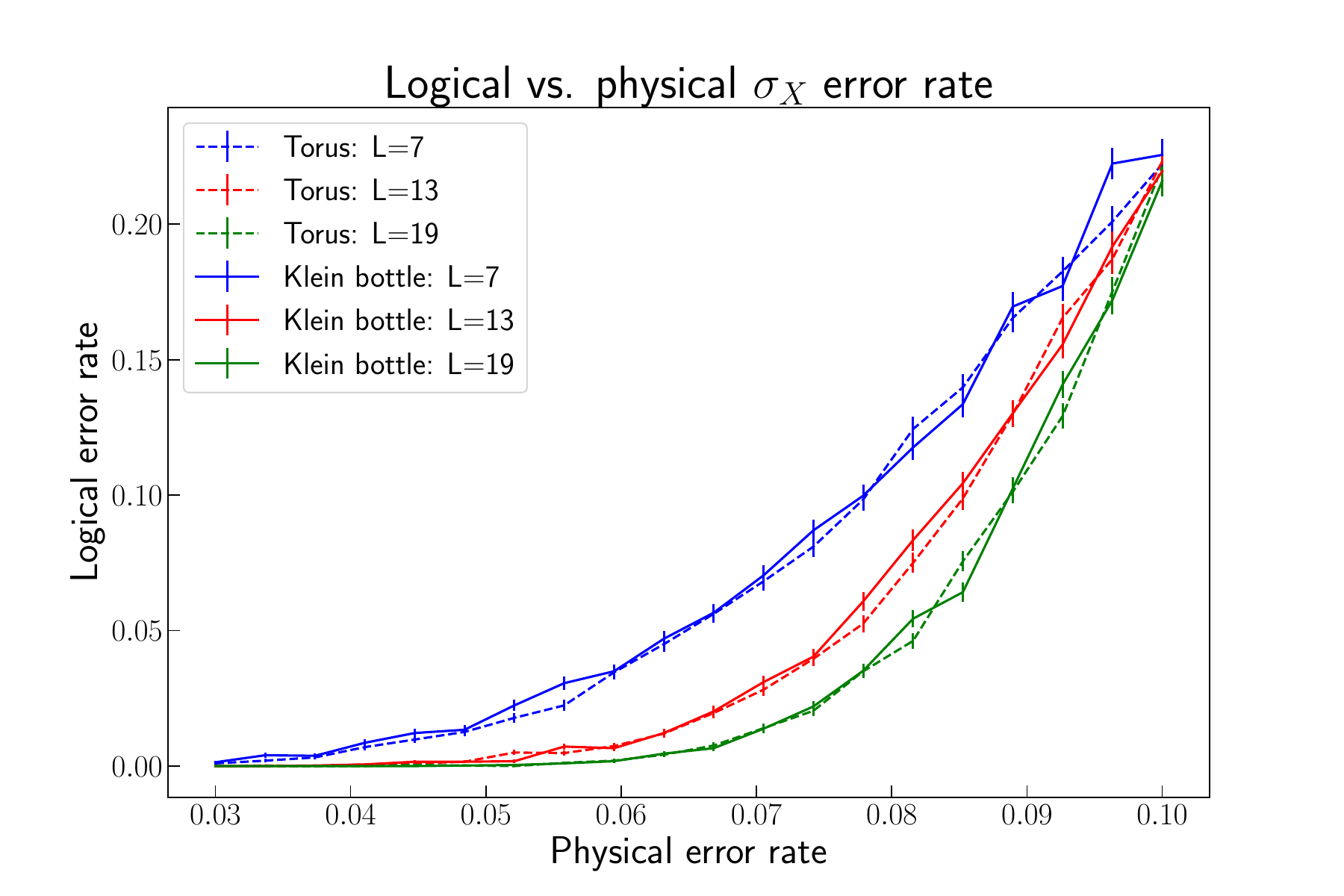}
         \caption{}
         \label{fig:odd}
     \end{subfigure}
     \caption{\raggedright Logical versus physical $\sigma_X$ error rates for TQEC code embedded on torus $T^2$ and Klein bottle $K$ for different lattice dimensions assuming perfect syndrome measurements. (a) $L=\{8, 14, 20\}$ are even. (b) $L=\{7, 13, 19\}$ are odd. There is no statistically significant distinction between the performance of TQEC codes on $K$ or $T^2$ when $d$ is odd.}
     \label{fig:logical_X_errors}
\end{figure}

In this section, we present the simulation results for the logical $\sigma_X$ versus physical $\sigma_X$ error rates. As discussed in Section~\ref{subsec: discussion} and illustrated in Figure~\ref{fig:Klein_Bottle_logical}, the shortest co-cycle $\beta \in C^1$ for both the Klein bottle and the torus remains identical for a given lattice dimension $L$, irrespective of whether $L$ is even or odd. Consequently, we predict that the logical $\sigma_X$ error rates for the toric and Klein bottle codes will be more or less the same under these conditions, as formerly discussed in Section~\ref{subsec: discussion}.

The methodology employed here is consistent with that described in Section~\ref{subsec:methodology}. Specifically, both the physical error rate and the corresponding logical error rate are collected and plotted for analysis. The physical error rate $p$ is sampled from a uniformly distributed random variable. For each lattice dimension $L$, simulations are performed over 5000 runs for both the torus $T^2$ and the Klein bottle $K$. The error bars depicted in the plots represent 1-sigma uncertainties.

The simulation results are presented in Figure~\ref{fig:logical_X_errors}. Figure~\ref{fig:even} displays the logical $\sigma_X$ error rates for lattices with even dimensions $d = \{8, 14, 20\}$, while Figure~\ref{fig:odd} shows the corresponding results for lattices with odd dimensions $d = \{7, 13, 19\}$, for both the torus $T^2$ and the Klein bottle $K$. The results indicate that there is no significant difference in the logical $\sigma_X$ error rates between $T^2$ and $K$, regardless of whether $L$ is even or odd. This observation is consistent with our theoretical predictions.

\bibliography{AlTopQEC}

\end{document}